\documentclass[a4paper,12pt,reqno,twoside, titlepage]{amsart}
\usepackage[margin=2cm]{geometry}
\usepackage{setspace}
\usepackage{amsmath, mathtools, calligra}
\usepackage{amsfonts}
\usepackage{amssymb}
\usepackage{xcolor}
\usepackage{graphicx}
\usepackage{pstricks}
\usepackage{array}
\usepackage{pictexwd}
\usepackage{fancybox}
\usepackage{natbib}
\usepackage{newcent}
\usepackage{enumitem}
\usepackage{tikz}
\usepackage{mathrsfs}
\usepackage{bm}
\usepackage{dsfont}
\usepackage{subcaption}

\setlist{nolistsep}
\newcolumntype{C}[1]{>{\centering\let\newline\\\arraybackslash\hspace{0pt}}m{#1}}
\DeclareMathAlphabet{\mathpzc}{OT1}{pzc}{m}{it}

\newtheorem{theorem}{Theorem}
\newtheorem{proposition}{Proposition}

\newtheorem{corollary}{Corollary}
\newtheorem{lemma}{Lemma}

\def \beq{\begin{eqnarray*}}
\def\eeq{\end{eqnarray*}}

\usepackage[textwidth=30mm]{todonotes}

\def \m{\overline{m}_{\overline{q}^1}}
\def \ms{\overline{m}_{q^*}}
\def \ph{\overline{\varphi}}
\def \la{\overline{\lambda}}
\def \q {\underline{q}^1}
\usepackage{mathtools}

\begin{document}
\title{Contracting over persistent information}
\date{\today}
\thanks{Wei Zhao gratefully acknowledges the support of the HEC Foundation. Claudio Mezzetti thankfully acknowledges financial support from the Australian Research Council Discovery grant DP190102904. Ludovic Renou gratefully acknowledges the support of the Agence Nationale pour la Recherche under grant ANR CIGNE (ANR-15-CE38-0007-01) and through the ORA Project ``Ambiguity in Dynamic Environments'' (ANR-18-ORAR-0005). Tristan Tomala gratefully acknowledges the support of the HEC Foundation and ANR/Investissements d'Avenir under grant ANR-11-IDEX-0003/Labex Ecodec/ANR-11-LABX-0047. }
\author{Wei Zhao}
\address{Wei Zhao, HEC Paris and GREGHEC-CNRS, 1 rue de la Lib\'eration, 78351 Jouy-en-Josas, France}
\email{wei.zhao1(at)hec.fr}
\author{Claudio Mezzetti}
\address{Claudio Mezzetti, School of Economics, The University of Queensland, Level 6, Colin Clark BuildingSt Lucia,
Brisbane Qld 4072, Australia}
\email{c.mezzetti(at)uq.edu.au}
\author{Ludovic Renou}
\address{Ludovic Renou, Queen Mary University of London, CEPR and University of Adelaide, Miles End, E1 4NS, London, UK}
\email{lrenou.econ(at)gmail.com}
\author{Tristan Tomala}
\address{Tristan Tomala, HEC Paris and GREGHEC-CNRS, 1 rue de la Lib\'eration, 78351 Jouy-en-Josas, France}
\email{tomala(at)hec.fr}

\begin{abstract}
We consider a dynamic principal-agent problem, where the \emph{sole} instrument the principal has to incentivize the agent is the disclosure of information. The principal aims at maximizing the (discounted) number of times the agent chooses the principal's most preferred action, e.g., to work hard on the principal's task. We show that there exists an optimal contract, where the principal stops disclosing information as soon as its most preferred action is a static best reply for the agent, or else continues disclosing information until the agent \emph{perfectly} learns the principal's private information. If the agent perfectly learns the state, he learns it in finite time with probability one; the more patient the agent, the later he learns it.


\end{abstract}

\maketitle

\newpage

\section{Introduction}
We consider a dynamic ``principal-agent'' model, where the sole instrument the principal has is information.\footnote{That is, the principal cannot make transfers, terminate the relationship, choose allocations or constrain the agent's choices.} The principal aims at incentivizing the agent to choose the same action -- the principal's most preferred action -- as often as possible, and can only do so by disclosing information about an unknown (binary) state. E.g., the agent is a customer and the principal a service provider, who discloses information about its services to generate as many sales as possible.  We assume that the principal commits to a disclosure policy and we refer to the principal's commitment as the offer of a ``contract.'' The dynamic contracting problem we study is, therefore, a \emph{dynamic persuasion problem}. \medskip

The standard approach in the study of dynamic contracting models (e.g., \citet{spear-srivastava-87}) is to use the agent's continuation value as a state variable. The principal's Bellman equation is then the fixed point of an operator, which satisfies a promised keeping constraint in addition to incentive constraints. In dynamic persuasion models, there is an additional complication, however. The information the principal commits to disclose to the agent generates a \emph{martingale} of beliefs:   the posterior beliefs of the agent must be equal in expectation to his prior beliefs. We thus need to incorporate the agent's beliefs as an additional state variable and to impose the constraint that the belief process is a martingale.  In spite of the increased dimensionality of the principal's problem, we are able to provide a complete characterization of an optimal contract by simultaneously solving for the evolution of the agent's beliefs and promised utility. To the best of our knowledge, we are the first to tackle this difficulty. 

  \medskip

We illustrate the main properties of our optimal policy -- particularly how beliefs evolve over time -- with the help of Figure \ref{fig:beliefsintro}. Figure \ref{fig:beliefsintro} plots four representative evolutions of the agent's belief about the ``high opportunity cost'' state -- the state where the cost to incentivize the agent relative to the benefit is  the highest.  In each panel, the grey region ``OPT'' indicates the region at which choosing the principal's most preferred action is (statically) optimal for the agent.  An arrow pointing from one belief to another indicates how the agent revised his belief within the period following a signal's realization. Multiple arrows originating from the same point thus represent the information disclosed by the policy. Within a period,  the agent takes a decision after having revised his beliefs. Arrows have different colors/patterns.  At all beliefs at the end of continuous black arrows, the agent follows the principal's recommendation.  At all beliefs at the end of dotted magenta arrows, he does not (and chooses what is best given his current belief).

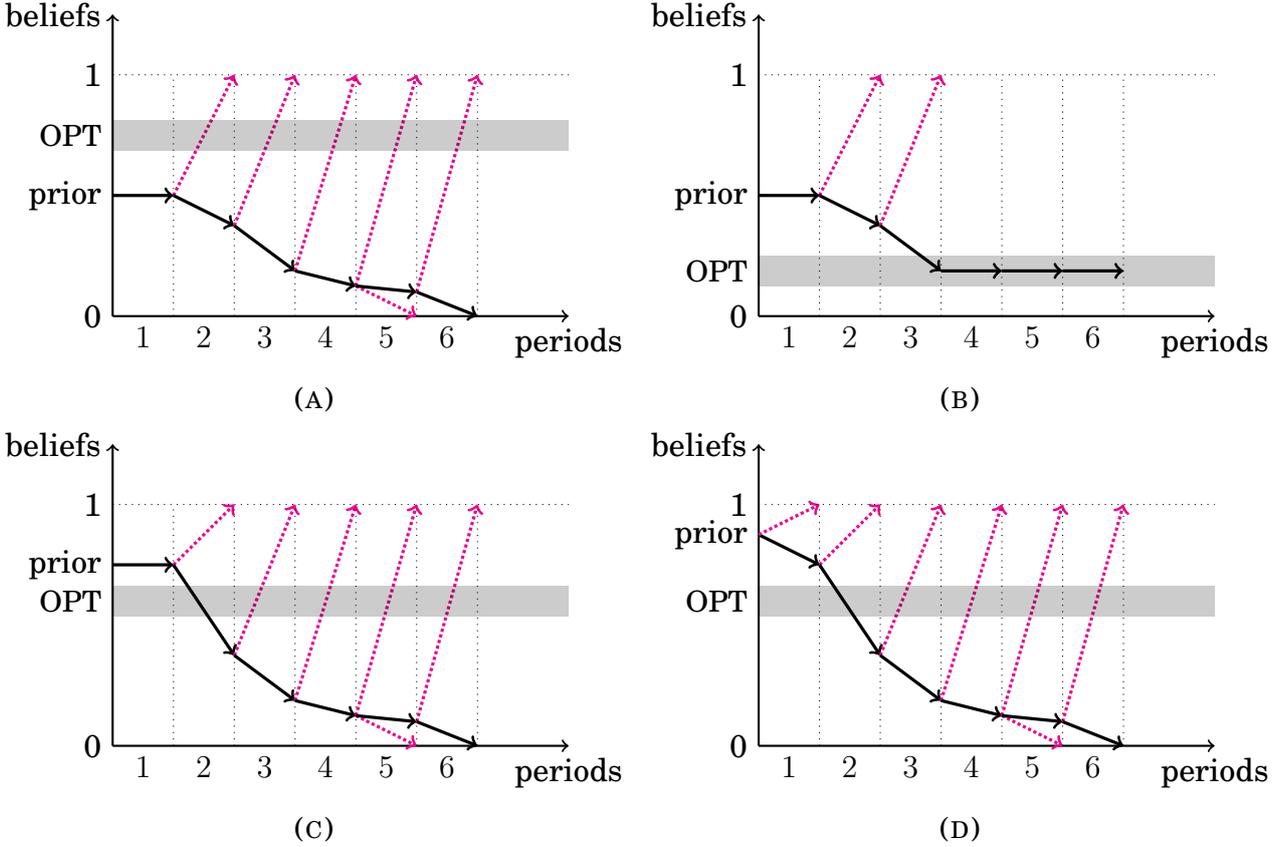
\begin{figure}[h]
\centering
\begin{subfigure}[b]{0.5\textwidth}
\centering
\begin{tikzpicture}[scale = 0.8]
\draw[->, thick](0,0)--(7.5,0) node[below]{periods};
\draw[->,thick](0,0)--(0,5) node[left]{beliefs};
\draw[->, very thick](0,2)--(1,2);
\draw[->,very thick, densely dotted, magenta](1,2)--(2,4);
\draw[->,very thick](1,2)--(2,1.5);
\draw[->,very thick, densely dotted, magenta](2,1.5)--(3,4);
\draw[->,very thick](2,1.5)--(3,0.75);
\draw[->,very thick, densely dotted, magenta](3,0.75)--(4,4);
\draw[->,very thick](3,0.75)--(4,0.5);
\draw[->,very thick, densely dotted, magenta](4,0.5)--(5,0);
\draw[->,very thick,  densely dotted, magenta](4,0.5)--(5,4);
\draw[->,very thick](4,0.5)--(5,0.4);
\draw[->,very thick, densely dotted, magenta](5,0.4)--(6,4);
\draw[->,very  thick](5,0.4)--(6,0);
\node[below] at (0.5,0) {$1$};
\node[below] at (1.5,0) {$2$};
\node[below] at (2.5,0) {$3$};
\node[below] at (3.5,0) {$4$};
\node[below] at (4.5,0) {$5$};
\node[below] at (5.5,0){$6$};
\draw[dotted](0,4)--(7.5,4) node[left] at (0,4){1} ; 
\node[left] at (0,2){prior};
\draw[dotted](3,0)--(3,4);
\draw[dotted](2,0)--(2,4);
\draw[dotted](1,0)--(1,4);
\draw[dotted](4,0)--(4,4);
\draw[dotted](5,0)--(5,4);
\draw[dotted](6,0)--(6,4);
\node[left] at (0,0) {0};
\fill[fill opacity=0.2] (0,2.75)--(0,3.25)--(7.5,3.25)--(7.5,2.75)--cycle;
\node[left] at (0,3){\small{OPT}};
\end{tikzpicture}
\caption{}
\end{subfigure}%
\begin{subfigure}[b]{0.5\textwidth}
\centering
\begin{tikzpicture}[scale=0.8]
\draw[->, thick](0,0)--(7.5,0) node[below]{periods};
\draw[->,thick](0,0)--(0,5) node[left]{beliefs};
\draw[->, very thick](0,2)--(1,2);
\draw[->,very thick, densely dotted, magenta](1,2)--(2,4);
\draw[->,very thick](1,2)--(2,1.5);
\draw[->,very thick, densely dotted, magenta](2,1.5)--(3,4);
\draw[->,very thick](2,1.5)--(3,0.75);
\draw[->,very thick](3,0.75)--(4,0.75);
\draw[->,very thick](4,0.75)--(5,0.75);
\draw[->,very  thick](5,0.75)--(6,0.75);
\node[below] at (0.5,0) {$1$};
\node[below] at (1.5,0) {$2$};
\node[below] at (2.5,0) {$3$};
\node[below] at (3.5,0) {$4$};
\node[below] at (4.5,0) {$5$};
\node[below] at (5.5,0){$6$};
\draw[dotted](0,4)--(7.5,4) node[left] at (0,4){1} ; 
\node[left] at (0,2){prior};
\draw[dotted](3,0)--(3,4);
\draw[dotted](2,0)--(2,4);
\draw[dotted](1,0)--(1,4);
\draw[dotted](4,0)--(4,4);
\draw[dotted](5,0)--(5,4);
\draw[dotted](6,0)--(6,4);
\node[left] at (0,0) {0};
\fill[fill opacity=0.2] (0,0.5)--(0,1)--(7.5,1)--(7.5,0.5)--cycle;
\node[left] at (0,0.75){\small{OPT}};
\end{tikzpicture}
\caption{}
\end{subfigure}

\begin{subfigure}[b]{0.5\textwidth}
\centering
\begin{tikzpicture}[scale = 0.8]
\draw[->, thick](0,0)--(7.5,0) node[below]{periods};
\draw[->,thick](0,0)--(0,5) node[left]{beliefs};
\draw[->, very thick](0,3)--(1,3);
\draw[->,very thick, densely dotted, magenta](1,3)--(2,4);
\draw[->,very thick](1,3)--(2,1.5);
\draw[->,very thick, densely dotted, magenta](2,1.5)--(3,4);
\draw[->,very thick](2,1.5)--(3,0.75);
\draw[->,very thick, densely dotted, magenta](3,0.75)--(4,4);
\draw[->,very thick](3,0.75)--(4,0.5);
\draw[->,very thick, densely dotted, magenta](4,0.5)--(5,0);
\draw[->,very thick,  densely dotted, magenta](4,0.5)--(5,4);
\draw[->,very thick](4,0.5)--(5,0.4);
\draw[->,very thick, densely dotted, magenta](5,0.4)--(6,4);
\draw[->,very  thick](5,0.4)--(6,0);
\node[below] at (0.5,0) {$1$};
\node[below] at (1.5,0) {$2$};
\node[below] at (2.5,0) {$3$};
\node[below] at (3.5,0) {$4$};
\node[below] at (4.5,0) {$5$};
\node[below] at (5.5,0){$6$};
\draw[dotted](0,4)--(7.5,4) node[left] at (0,4){1} ; 
\node[left] at (0,3){prior};
\draw[dotted](3,0)--(3,4);
\draw[dotted](2,0)--(2,4);
\draw[dotted](1,0)--(1,4);
\draw[dotted](4,0)--(4,4);
\draw[dotted](5,0)--(5,4);
\draw[dotted](6,0)--(6,4);
\node[left] at (0,0) {0};
\fill[fill opacity=0.2] (0,2.65)--(0,2.15)--(7.5,2.15)--(7.5,2.65)--cycle;
\node[left] at (0,2.4){\small{OPT}};
\end{tikzpicture}
\caption{}
\end{subfigure}%
\begin{subfigure}[b]{0.5\textwidth}
\centering
\begin{tikzpicture}[scale = 0.8]
\draw[->, thick](0,0)--(7.5,0) node[below]{periods};
\draw[->,thick](0,0)--(0,5) node[left]{beliefs};
\draw[->, very thick](0,3.5)--(1,3);
\draw[->,very thick, densely dotted, magenta](0,3.5)--(1,4);
\draw[->,very thick, densely dotted, magenta](1,3)--(2,4);
\draw[->,very thick](1,3)--(2,1.5);
\draw[->,very thick, densely dotted, magenta](2,1.5)--(3,4);
\draw[->,very thick](2,1.5)--(3,0.75);
\draw[->,very thick, densely dotted, magenta](3,0.75)--(4,4);
\draw[->,very thick](3,0.75)--(4,0.5);
\draw[->,very thick, densely dotted, magenta](4,0.5)--(5,0);
\draw[->,very thick,  densely dotted, magenta](4,0.5)--(5,4);
\draw[->,very thick](4,0.5)--(5,0.4);
\draw[->,very thick, densely dotted, magenta](5,0.4)--(6,4);
\draw[->,very  thick](5,0.4)--(6,0);
\node[below] at (0.5,0) {$1$};
\node[below] at (1.5,0) {$2$};
\node[below] at (2.5,0) {$3$};
\node[below] at (3.5,0) {$4$};
\node[below] at (4.5,0) {$5$};
\node[below] at (5.5,0){$6$};
\draw[dotted](0,4)--(7.5,4) node[left] at (0,4){1} ; 
\node[left] at (0,3.5){prior};
\draw[dotted](3,0)--(3,4);
\draw[dotted](2,0)--(2,4);
\draw[dotted](1,0)--(1,4);
\draw[dotted](4,0)--(4,4);
\draw[dotted](5,0)--(5,4);
\draw[dotted](6,0)--(6,4);
\node[left] at (0,0) {0};
\fill[fill opacity=0.2] (0,2.65)--(0,2.15)--(7.5,2.15)--(7.5,2.65)--cycle;
\node[left] at (0,2.4){\small{OPT}};
\end{tikzpicture}
\caption{}
\end{subfigure}%

\caption{Evolution of actions and beliefs over time}\label{fig:beliefsintro}
\end{figure}

The following are general properties of our optimal policy. The first observation to make is that the agent repeatedly updates his belief until either he perfectly learns the state or choosing the principal's most preferred action becomes (statically) optimal.  Moreover, if the agent learns the state, he learns it in finite time.  We provide an explicit characterization of the priors at which the agent eventually learns the state. 

Second, along the paths at which the agent plays the principal's most preferred action, his beliefs about the ``high'' state are decreasing. Intuitively, the optimal contract exploits the asymmetry in opportunity costs and lowers the perceived opportunity cost -- hence making it easier to incentivize the agent -- by sometimes informing him when the opportunity cost is high.\footnote{To be precise, under our policy, upon receiving the signal ``the opportunity cost is high,'' the agent learns that this is indeed true. However, the signal is not sent with probability one. This corresponds to the (magenta/dotted) arrows ending at 1.} 

Third, with the exception of panel (\textsc{d}), the policy does not disclose information to the agent at the first period. Thus, adopting the definition of persuasion as the act of changing the agent's beliefs prior to him making any decision, information disclosure rewards the agent for following the recommendation, but does not persuade him in panels (\textsc{A}), (\textsc{B}) and (\textsc{C}). Yet, as panel (\textsc{d}) illustrates, the policy sometimes needs to persuade the agent. For instance, if the promise of full information disclosure at the next period wouldn't incentivize the agent, then persuading the agent is necessary, that is, the policy must generate a strictly positive value of information for the agent. There are other circumstances at which persuading the agent may be necessary. Persuasion can reduce the agent's expected opportunity cost of following the principal's recommendation sufficiently enough to compensate for the loss to the principal due to providing information at the start of the relationship. 

Finally, with the exception of panel (\textsc{b}), the policy does not induce the agent to believe that playing the principal's most preferred action is optimal. This is markedly different from what we would expect from the static analysis of \citet{Kamenica2011}.   Intuitively, the ``static'' policy is sub-optimal because it does not extract all the informational surplus it creates, that is, it creates a strictly positive value of information, but does not extract it all. (The participation constraint of the agent does not bind.)  Even in panel (\textsc{b}), the beliefs do not jump immediately to the ``OPT'' region. In fact, the belief process may approach the ``OPT'' region only asymptotically.

 \medskip

\textbf{Related literature.} The paper is part of the literature on Bayesian persuasion, pioneered by \citet{Kamenica2011}, and recently surveyed by \citet{Kamenica-survey}. The three  most closely related papers are \citet{Ball2019dynamic}, \citet{Ely2019},  and \citet{orlov2018persuading}. In common with our paper, these papers study the optimal disclosure of information in dynamic games and show how the disclosure of information can be used as an incentive tool. The observation that information can be used to incentivize agents is not new and dates back to the literature on repeated games with incomplete information, e.g., \citet{Aumann-Maschler-Stearns}. See \citet{GaricanoRayo2017} and \citet{fudenberg-rayo-2019} for some more recent papers exploring the role of information provision as an incentive tool. \medskip

The classes of dynamic games studied differ considerably from one paper to another, which makes comparisons difficult. In \citet{Ely2019}, the agent has to repeatedly decide whether to continue working on a project or to quit (i.e., unlike our paper, there only two actions); quitting ends the game. The principal aims at maximizing the number of periods the agent works on the project and can only do so by disclosing information about its complexity, modeled as the number of periods required to complete it. Thus, their dynamic game is a quitting game, while ours is a repeated game.  When the project is either easy or difficult (i.e., when there are two states), the optimal disclosure policy initially persuades the agent that the task is easy, so that he starts working. (Naturally, if the agent is sufficiently convinced that the project is easy, there is no need to persuade him initially.)  If the project is in fact difficult, the policy then discloses it at a later date, when completing the project is now within reach. A main difference with our optimal disclosure policy is that information comes in lumps in \citet{Ely2019}, i.e., information is disclosed only at the initial period and at a later period, while information is repeatedly disclosed in our model.\footnote{When there are more than two states, the optimal policy discloses information more frequently in \citet{Ely2019}. The frequency of disclosure is thus a consequence of the dimensionality of the state space in their model, while it is not so in our model.} Another  main difference is as follows. In \citeauthor{Ely2019}, only when the promise of full information disclosure at a later date is not enough to   incentivize the agent to start working does the principal persuade the agent initially. This is not so with our policy: the principal persuades the agent in a larger set of circumstances. This initial persuasion reduces the cost of incentivizing the agent in future periods.\medskip

\citet{orlov2018persuading} also consider a quitting game, where the principal aims at delaying the quitting time as far as possible. The quitting time is the time at which the agent decides to exercise an option, which has different values to the principal and the agent. The principal chooses a disclosure policy informing the agent about the option's value. When, as in this paper, the principal  commits to a long-run policy, the optimal policy is to fully reveal the state with some delay. (Note that the principal is referred to as the agent in their work.) This policy is not optimal in \citet{Ely2019}, or in our paper. See \citet{au2015dynamic}, \citet{BizzottoForthcoming}, \citet{Che2020}, \citet{Henry2019} and \citet{smolin2018} for  other papers on information disclosure in quitting games, where the agent either waits and obtains additional information or takes an irreversible action and stops the game.\medskip

\citet{Ball2019dynamic} studies a continuous time model of information provision, where the state changes over time and  payoffs are the ones of the quadratic example of \citet{Crawford1982}. Ball shows that the optimal disclosure policy requires the sender to disclose the current state at a later date, with the delay shrinking over time. The main difference between his work and ours is the persistence of the state (also, we consider two different classes of games). When the state is fully persistent, as in \citet{Ely2019} and our model, full information disclosure with delay is not optimal in general. (See the discussion of Example 1 in Section \ref{sec:final}.)  \medskip 

Finally, there are a few papers on dynamic persuasion, where the agent takes an action repeatedly. However, either the agent is myopic, e.g., \citet{Ely2017} and \citet{Renault2017}, or the principal cannot commit, e.g., \citet{Escude2020}.

\section{The problem}
A principal and an agent interact over an infinite number of periods, indexed by $t \in \{1,2,\dots\}$.  At the first stage, before the interaction starts, the principal learns a payoff-relevant state $\omega \in \Omega=\{\omega_0,\omega_1\}$, while the agent remains uninformed.  The prior probability of $\omega$ is $p_0(\omega)>0$. At each period $t$, the principal sends a signal $s \in S$ and, upon observing the signal $s$, the agent takes decision $a \in A$. The sets $A$ and $S$ are finite. The cardinality of $S$ is as large as necessary for the principal to be unconstrained in his signaling.\footnote{From \citet{makris-renou}, it is enough to have the cardinality of $S$ as large as the cardinality of $A$.}
\medskip

We assume that there exists $a^* \in A$ such that the principal's payoff is strictly positive whenever $a^*$ is chosen and zero, otherwise. The principal's payoff function is thus $v: A \times \Omega \rightarrow \mathbb{R}$, with $v(a^*,\omega_0) >0$,  $v(a^*,\omega_1) >0$ and $v(a,\omega_0)= v(a,\omega_1) =0$ for all $a \in A \setminus \{a^*\}$.  The agent's payoff function is $u: A \times \Omega \rightarrow \mathbb{R}$. The (common) discount factor is $\delta \in (0,1)$. \medskip

We write $A^{t-1}$ for $\underbrace{A \times \dots \times A}_{t-1 \text{\;times\;}}$ and $S^{t-1}$ for $\underbrace{S \times \dots \times S}_{t-1 \text{\;times\;}}$, with generic elements $a^t$ and $s^t$, respectively. A behavioral strategy for the principal is a collection of maps $\tau=(\tau_t)_{t=1}^{\infty}$, with $\tau_t: A^{t-1} \times S^{t-1} \times \Omega \rightarrow \Delta(S)$. Similarly, a behavioral strategy for the agent is a collection of maps $\sigma=(\sigma_t)_{t=1}^{\infty}$ with $\sigma_t: A^{t-1} \times S^{t-1} \times S \rightarrow \Delta(A)$. \medskip

We write $\mathbf{V}(\tau,\sigma)$ for the principal's payoff and $\mathbf{U}(\tau,\sigma)$ for the agent's payoff under the strategy profile $(\sigma,\tau)$. The objective is to characterize the maximal payoff the principal achieves if he commits to a strategy $\tau$, that is,
\[\sup_{(\tau,\sigma)} \mathbf{V}(\tau,\sigma), \]
subject to
\[\mathbf{U}(\tau,\sigma)  \geq \mathbf{U}(\tau,\sigma'),\] for all $\sigma'$.\medskip

Several comments are worth making. First, we interpret the strategy the principal commits to as a \emph{contract} specifying, as a function of the state, the information to be disclosed at each history of realized signals and actions. That is, the contract specifies a statistical experiment at each history of realized signals and states. The principal chooses the contract prior to learning the state. An alternative interpretation is that neither the principal nor the agent know the state, but the principal has the ability to conduct statistical experiments contingent on past signals and actions.  We can partially dispense with the commitment assumption. Indeed, since the choices of statistical experiments are observable, we can construct strategies that incentivize the principal to implement the specified statistical experiments.\footnote{The simplest such strategy is to have the agent play $a \neq a^*$ in all future periods after a deviation.} Second, the only additional information the agent obtains each period is the outcome of the statistical experiment. Third, the state is fully persistent and the principal perfectly monitors the action of the agent. Finally, the only instrument available to the principal is information. The principal can neither remunerate the agent nor terminate the relationship nor allocate different tasks to the agent. We purposefully make all these assumptions to address our main question of interest: what is the optimal way to incentivize the agent with information only? \medskip

\section{Optimal contracts}
This section fully characterizes an optimal contract and discusses its most salient properties. We first start with a recursive formulation of the principal's problem.

\subsection{A recursive formulation}
The first step in deriving an optimal contract is to reformulate the principal's problem as a recursive problem. To do so, we introduce two state variables. The first state variable we introduce is promised continuation payoff. It is well-known that classical dynamic contracting problems admit recursive formulations if we introduce promised continuation payoff as a state variable and  impose promise-keeping constraints, e.g., \citet{spear-srivastava-87}. The second state variable we introduce is beliefs. We now turn to the formal reformulation of the problem.\footnote{A nearly identical reformulation  already appeared in \citet{Ely2015}, one of the working versions of \citet{Ely2017}.  We remind the reader that \citet{Ely2017} analyzes the interaction between a long-run principal and a sequence of short-run agents. (See also \citet{Renault2017}.) While discussing the extension of his model to the interaction between a long-run principal and a long-run agent, \citet{Ely2015} derived a recursive  reformulation nearly identical to ours. However, he didn't study further the reformulated problem. We start from the recursive formulation and use it to derive an optimal policy. See Section \ref{app-recursive-formulation} for a detailed comparison of the two formulations.} \medskip

We first need some additional notation. We write $p(\omega)$ for the probability of $\omega$. Throughout, we abuse notation and denote $p$ the belief that the state is $\omega_1$. We let $u(a,p):= \sum_{\omega} p(\omega)u(a,\omega)$ be the agent's expected payoff of choosing $a$ when his belief is $p$, $m(p):= \max_{a \in A} u(a,p)$ be the agent's best payoff when his belief is $p$, and
$M(p):=\sum_{\omega}p(\omega)\max_{a \in A}u(a,\omega)$ be the agent's expected payoff if he learns the state prior to choosing an action. It is worth noting that $m$ is a piece-wise linear convex function, that $M$ is linear and that $m(p) \leq M(p)$ for all $p$. Similarly, we let $v(a,p)$ be the principal's payoff when the agent chooses $a$ and the principal's belief is $p$. Finally, let $P:=\{p \in [0,1]: m(p)=u(a^*,p)\}$, be the set of beliefs at which $a^*$ is optimal. If non-empty, the set $P$ is a closed interval $[\underline{p},\overline{p}]$. \medskip

Let $\mathcal{W} \subseteq [0,1] \times \mathbb{R}$ such that $(p,w) \in \mathcal{W}$ if and only if $w \in [m(p),M(p)]$. Throughout, we restrict attention to functions $V: \mathcal{W} \rightarrow \mathbb{R}$, with the interpretation that $V(p,w)$ is the principal's payoff if he promises a continuation payoff of $w$ to the agent when the agent's current belief is $p$. \medskip

The principal's maximal payoff is $V^*(p_0,m(p_0))$, where $V^*$ is the unique fixed point of the contraction $T$, defined by 
\begin{eqnarray*}
T(V)(p,w): =
\begin{cases}
\max_{\big((\lambda_s, (p_s,w_s),a_s)  \in [0,1]\times \mathcal{W} \times A\big)_{s \in S}} \sum_{s \in S} \lambda_s [(1-\delta)v(a_s,p_s)+ \delta V (p_s,w_s)], \\
\text{subject to:}  \\
(1-\delta)u(a_s,p_s)+ \delta w_s \geq m(p_s),\\
\sum_{s \in S} \lambda_s[ (1-\delta)u(a_s,p_s)+ \delta w_s] \geq w, \\
\sum_{s \in S} \lambda_s p_s=p,  \sum_{s \in S} \lambda_s=1.\\
\end{cases}
\end{eqnarray*}

We briefly comment on the maximization program. Recalling that $S$ is the set of possible signals $s$, at each $(p,w)$, a policy prescribes the probability $\lambda_s$ that the realized signal is $s$ and conditional on $s$, the belief $p_s$, the promised utility $w_s$, and the recommended action $a_s$. The first constraint is an incentive constraint: the agent must have an incentive to play $a_s$ when $w_s$ is the agent's promised continuation payoff and $p_s$ the agent's belief. To understand the right-hand side, observe that the agent can always play a static best reply to any belief, so that his expected payoff must be at least $m(p_s)$ when his current belief is $p_s$.\footnote{More precisely, if the agent's belief at period $t$ is $p_{t}$, he obtains the payoff $m(p_{t})$ by playing a static best-reply. Since the function $m$ is convex and beliefs follow a martingale, his expected payoff is therefore at least $(1-\delta)\sum_{t' \geq t}\delta^{t'-t}\mathbb{E}[m(\mathbf{p}_{t'})|\mathcal{F}_t] \geq m(p_t)$, where $\mathcal{F}_t$ is the agent's filtration at period $t$.} Conversely, if the contract specifies action $a_s$  and the agent does not execute that action, the contract can specify  no further information revelation, in which case the agent's payoff is at most $m(p_s)$.  Therefore, $m(p_s)$ is the agent's min-max payoff. The second constraint is the  promise-keeping constraint: if the principal promises the continuation payoff $w$ at a period, the contract must honor that promise in subsequent periods.  The third constraint states that the policy selects a splitting of $p$; that is, a distribution over posteriors with expectation $p$. \medskip

Throughout, we slightly abuse notation and denote by $\tau$ a policy, that is, a function from $\mathcal{W}$ to $([0,1] \times \mathcal{W} \times A)^{|S|}$. A policy is feasible if it specifies a feasible tuple  $((\lambda_s, (p_s,w_s),a_s))_{s \in S}$ for each $(p,w)$, i.e., a tuple satisfying the constraints of the maximization problem $T(V)(p,w)$.\medskip

Two important observations are worth making. First, for any function $V$,  $T(V)$ is a concave function in $(p,w)$. Concavity reflects the fact that the more information the principal discloses, the harder it is to reward the agent in the future. Second, $T(V)$ is a decreasing function in $w$, that is, the more the principal promises to the agent, the harder it is to incentivize the agent to play $a^*$.
\footnote{A real-valued function $f$ is increasing (resp., strictly increasing) if $x > y$ implies that $f(x) \geq f(y)$ (resp., $f(x)>f(y)$). The function $f$ is (resp., strictly) decreasing if $-f$ is (resp., strictly) increasing.} We will repeatedly make use of these two properties, which we formally record in the following proposition.\medskip

\begin{proposition}\label{prop:concave}
The value function $V^*$ is concave in both arguments and decreasing in $w$.
\end{proposition}

Proposition \ref{prop:concave} together with the recursive formulation has a number of additional implications, which are formally stated in Proposition \ref{prop2}. 

First, if the principal induces the posterior $p_s$ while recommending the action $a_s$ and promising the continuation payoff $w_s$, the principal should not have an incentive to further disclose information in that period. 

Second, if the principal does not recommend $a^*$ at a period, then the principal never recommends $a^*$ at a subsequent period, that is, the principal's continuation value is zero. In other words, as soon as an action other than $a^*$ is played, the principal stops incentivizing the agent to play $a^*$.  The intuition is simple. Suppose to the contrary that the principal were to recommend $a_s \neq a^*$  after the signal $s$ at period $t$ and $a^*$ at the next period. Consider a policy change, where the principal anticipates the disclosure of the information, which incentivizes the agent to play $a^*$ at period $t+1$, to period $t$. This policy change is feasible and increases the principal's payoff, a contradiction. 

Third, there is at most one signal $s^*$ at which the principal recommends the agent to play $a^*$. Moreover, whenever the principal recommends $a^*$, the agent is indifferent between obeying the recommendation  or deviating. In other words, the promised continuation payoff does not leave rents to the agent. Intuitively, if two signals recommended $a^*$, the principal would not lose from merging them into one. Let us call $s^*$ the unique signal at which the agent is recommended $a^*$. If the agent were given a positive rent when signal $s^*$ realizes, the principal would benefit by a change in policy that reduces the agent's promised utility associated with $s^*$ (since the value function is decreasing in promised utility).  For that change in policy to be feasible, the change must increase the promised utility when some other signal $s \neq s^*$ is realized. As we have already seen, this does not affect the principal's payoff (since the principal obtains a zero payoff in all periods, which follows a recommendation different from $a^*$).

\begin{proposition}\label{prop2} For all $(p,w)$, there exists a solution $(\lambda_s, p_s,w_s,a_s)_{s \in S}$ to $T(V^*)(p,w)$ such that
\begin{description}
\item[(i)] 
    For all $s \in S$ such that $\lambda_s >0$, we have
    \[\left( {1 - \delta } \right)v\left( {{a_s},{p_s}} \right) + \delta {V^ * }\left( {{p_s},{w_s}} \right) = {V^ * }\left( {{p_s},\left( {1 - \delta } \right)u\left( {{a_s},{p_s}} \right) + \delta {w_s}} \right).\]
 \item[(ii)]  
 For all $s \in S$ such that $\lambda_s>0$ and $a_s \neq a^*$, $V^*(p_s,w_s)=0$.
 \item[(iii)]
 There exists  at most one signal $s^* \in S$ such that $\lambda_{s^*}>0$ and $a_{s^*}=a^*$.   Moreover,
 \[(1-\delta)u(a_{s^*},p_{s^*})+ \delta w_{s^*} =m(p_{s^*}).\]
 \end{description}
\end{proposition}

Proposition \ref{prop2} states key properties that an optimal policy possesses. We conclude this section with a partial converse, that is, we state properties that guarantee optimality of a policy. To do so, we need some additional notation.  We first let $Q^1$ be the set of beliefs at which the agent has an incentive to play $a^*$ if he is promised full information disclosure at the next period, that is,
\[Q^1:=\{p \in [0,1]: (1-\delta)u(a^*,p)+\delta M(p) \geq m(p)\}. \]
If $Q^1$ is empty, then all policies are optimal as the principal can never incentivize the agent to play $a^*$. If $Q^1$ is non-empty, then it is a closed interval $[\underline{q}^1,\overline{q}^1]$. Note that $\q=0$ if and only if $a^*$ is optimal at $p=0$. For a graphical illustration, see Figure \ref{fig:Q^1}.
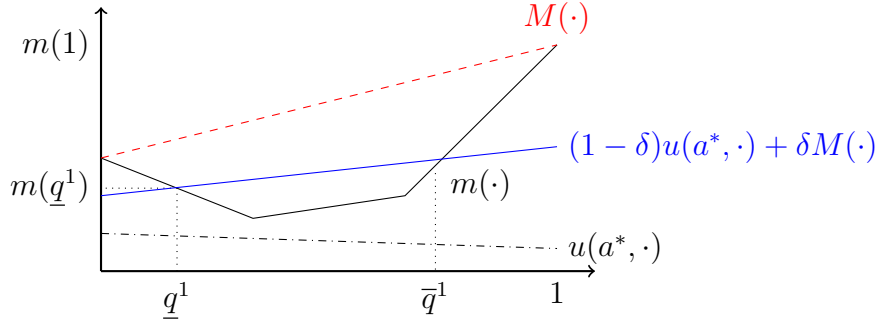
\begin{figure}[h]
\begin{center}
\begin{tikzpicture}
\draw[->,thick] (0,0)--(0,3.5);
\draw[->,thick] (0,0) -- (6.5,0);
\draw (0,1.5)--(2,0.7)--(4,1)--(6,3);
\node[below] at (5,1.5) {$m(\cdot)$};
\draw[dashed, red] (0,1.5)--(6,3) node[above] {$M(\cdot)$};
\draw[dashdotted] (0,0.5)--(6,0.3) node[right] {$u(a^*,\cdot)$};
\node[below] at (1,0) {$\underline{q}^1$};
\node[left] at (0,1.1) {$m(\underline{q}^1)$};
\draw[dotted] (1,0)--(1,1.1)--(0,1.1);
\node[below] at (6,0) {$1$};
\node[left] at (0,3) {$m(1)$};
\draw[blue] (0,1)--(6,1.65) node[right] {$(1-\delta) u(a^*,\cdot) + \delta M(\cdot)$};
\draw[dotted] (4.4,1.3)--(4.4,0) node[below] {$\overline{q}^{1}$};
\end{tikzpicture}
\end{center}
\caption{Construction of the set $Q^1$}\label{fig:Q^1}
\end{figure}
\medskip

Second, for all $p \in Q^1$, we write $\bold{w}(p) \in [m(p),M(p)]$ for the continuation payoff that makes the agent indifferent between playing action $a^*$ and receiving the continuation payoff $\bold{w}(p)$ in the future, and playing a best reply to the belief $p$ forever, that is, $\bold{w}(p)$ solves:
\[(1-\delta)u(a^*,p) + \delta \bold{w}(p) = m(p).\]

\begin{theorem}\label{theo1}
Consider any feasible policy inducing the value function $\tilde{V}$.  If $\tilde{V}$ is concave in both arguments, decreasing in $w$ and satisfies
 \[ \tilde{V}(p,m(p)) \geq (1-\delta)v(a^*,p)+ \delta \tilde{V}(p,\bold{w}(p)),\] 
for all $p \in Q^1$, then the policy is optimal.
\end{theorem}

\begin{proof} We argue that $\tilde{V}$ is the fixed point of the operator $T$, hence $\tilde{V}=V^*$. Let $(\lambda_s, p_s,w_s,a_s)_{s \in S}$ be a solution to the  maximization problem $T(\tilde{V})(p,w)$. We first start with the following observation. Consider any $s$ such that $a_s \neq a^*$. We have
\begin{eqnarray*}
(1-\delta)v(a_s,p_s)+\delta \tilde{V}(p_s,w_s)=\delta \tilde{V}(p_s,w_s) \leq  \tilde{V}(p_s,w_s) \leq \tilde{V}(p_s, (1-\delta)u(a_s,p_s)+\delta w_s),
\end{eqnarray*}
where the last inequality follows from the fact that $\tilde{V}$ is decreasing in $w$ and $m(p_s) \leq  (1-\delta)u(a_s,p_s)+\delta w_s \leq (1-\delta) m(p_s) + \delta w_s \leq w_s$. 

Consider now any $s$ such that $a_s=a^*$. Since $(\lambda_s, p_s,w_s,a_s)_{s \in S}$ is feasible, we have that 
\[(1-\delta)u(a^*,p_s)+\delta w_s \geq m(p_s), \]
hence $p_s \in Q^1$ and, therefore, 
\[ \tilde{V}(p_s,m(p_s)) \geq (1-\delta)v(a^*,p_s)+ \delta \tilde{V}\Big(p_s,\underbrace{\frac{-(1-\delta)u(a^*,p_s)+m(p_s)}{\delta}}_{\bold{w}(p_s)}\Big). \]
The concavity of $\tilde{V}$ implies that 
\[\tilde{V}(p_s,(1-\delta)u(a^*,p_s)+\delta w_s) -\tilde{V}(p_s,m(p_s)) \geq \delta \Big[\tilde{V}(p_s,w_s)-\tilde{V}\Big(p_s,\bold{w}(p_s)\Big)\Big],\]
where we use the identity $(1-\delta)u(a^*,p_s)+\delta w_s - m(p_s) = \delta (w_s -\bold{w}(p_s))$ and observation (a) about concave functions in Section \ref{math:prelim}.

Combining the above two inequalities implies that   
\[\tilde{V}(p_s,(1-\delta)u(a^*,p_s)+ \delta w_s) \geq (1-\delta)v(a^*,p_s)+ \delta \tilde{V}(p_s,w_s). \]

It follows that
\begin{eqnarray*}
T( \tilde{V} )(p,w) & = & \sum_{s \in S} \lambda_s \left[(1 - \delta)v(a_s,p_s) + \delta \tilde{V} (p_s,w_s) \right] \\
 & \leq & \sum_{s \in S} \lambda_s \left[\tilde{V} ( p_s,(1-\delta)u(a_s,p_s) + \delta w_s) \right]\\
 & \leq & \tilde{V}\left(\sum_{s \in S}\lambda_s p_s, \sum_{s \in S}\lambda_{s}((1-\delta)u(a_s,p_s) + \delta w_s))\right)\\
 & \leq & \tilde{V}(p,w),
\end{eqnarray*}
where the second  inequality follows from the concavity of $\tilde{V}$ and the third inequality from $\tilde{V}$ being decreasing in $w$.\medskip

Conversely, since the policy inducing $\tilde{V}$ is feasible, we must have that $T(\tilde{V})(p,w) \geq \tilde{V}(p,w)$ for all $(p,w)$. This completes the proof. 
\end{proof}

\subsection{An optimal policy}
The objective of this section is to define a policy, which we later prove to be optimal. We denote  by $a^{p}$  a maximizer of $u(\cdot,p)$.
Without loss of generality, assume that $\frac{m(1) -u(a^*,1)}{v(a^*,1)} \geq \frac{m(0) -u(a^*,0)}{v(a^*,0)}$. (A symmetric argument applies if the reverse inequality holds.)  Note that if $a^*$ is optimal for the agent at $p=1$, i.e., $m(1)=u(a^*,1)$, then $a^*$ is also optimal at $p=0$ and, consequently, $a^*$ is optimal at all beliefs, i.e., $P=[0,1]$. In what follows, we exclude this trivial case and assume that $1 \notin P$. We can restate the condition as $\frac{v(a^*,0)}{v(a^*,1)} \geq \frac{m(0) -u(a^*,0)}{m(1) -u(a^*,1)}$, that is, the principal's benefit of $a^*$ in state $\omega_0$ relative to state $\omega_1$ is higher than the agent's opportunity cost in state $\omega_0$ relative to state $\omega_1$. 
 \medskip

Define the functions $\lambda: \mathcal{W} \rightarrow [0,1]$ and $\varphi: \mathcal{W} \rightarrow [0,1]$, with $(\lambda(p,w), \varphi(p,w))$ the unique solution to:
\begin{eqnarray}
\begin{pmatrix}
p \\ w
\end{pmatrix}
=
\lambda(p,w)
\begin{pmatrix}
\varphi(p,w) \\ m(\varphi(p,w))
\end{pmatrix}
+
(1-\lambda(p,w))
\begin{pmatrix}
1 \\ m(1)
\end{pmatrix}.\label{eq1}
\end{eqnarray}
for all $w>m(p)$, and $\lambda(p,m(p)), \varphi(p,m(p)))=(1,p)$. Geometrically, the solution \[(\varphi(p,w),m(\varphi(p,w))\] is the unique intersection between the line connecting $(p,w)$ and $(1,m(1))$ and the graph of the piecewise linear function $m(p)$.\footnote{Each piece corresponds to an optimal action.} See Figure \ref{fig:varphi} for an illustration. Intuitively, at all $(p,w)$, this corresponds to a policy which induces the belief and continuation payoff $(\varphi(p,w),m(\varphi(p,w)))$ with probability $\lambda(p,w)$ and the belief and continuation payoff $(1,m(1))$ with the complementary probability (in which case the agent learns that the state is $\omega_1$).\medskip

\begin{figure}[h]
\centering
\begin{tikzpicture}
\draw[->,thick] (0,0)--(0,3.5);
\draw[->,thick] (0,0) -- (6.5,0);
\draw (0,1.5)--(2,0.7)--(4,1)--(6,3);
\node[below] at (5,1.5) {$m(\cdot)$};
\draw[dashed] (0,1.5)--(6,3) node[above] {$M(\cdot)$};
\draw[dashed] (3,0)--(3,1.5)--(0,1.5);
\node[below] at (3,0) {$p$};
\node[left] at (0,1.5) {$w$};
\draw[dashed] (6,3)--(3,1.5)--(0.5,0.25);
\draw[dashed] (6,0)--(6,3)--(0,3);
\node[below] at (6,0) {$1$};
\node[left] at (0,3) {$m(1)$};
\draw[dashed] (1.7,0)--(1.7,.85)--(0,.85);
\node[below] at (1.7,0) {$\varphi(p,w)$};
\node[left] at (0,0.85) {$m(\varphi(p,w))$};
\end{tikzpicture}
\caption{Construction of $\lambda$ and $\varphi$}\label{fig:varphi}
\end{figure}
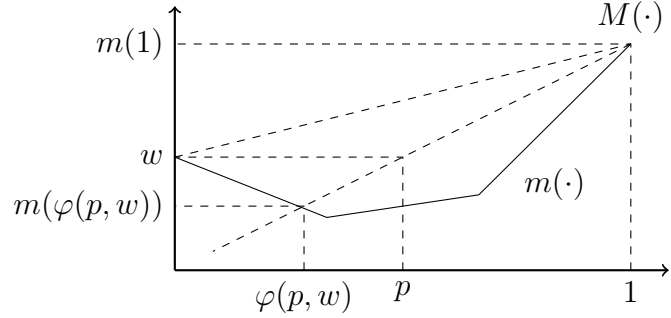
\medskip

We now define a family of policies $(\tau_q)_{q \in [\underline{q}^1,\overline{q}^1]}$ and show later the existence of $q^* \in [\underline{q}^1,\overline{q}^1]$ such that the policy $\tau_{q^*}$ is optimal. 
 
 For each $q \in [\underline{q}^1,\overline{q}^1]$, there are four regions to consider:
\begin{eqnarray*}
\mathcal{W}^1_q&:=&\Big\{(p,w): p \in [0,\underline{q}^1), w\leq \frac{\underline{q}^1-p}{\underline{q}^1}m(0) + \frac{p}{\underline{q}^1} m(\underline{q}^1)\Big\},\\
\mathcal{W}^2_q&:=&\Big\{(p,w): p \in (q,1], \frac{1-p}{1-q}m(q) + \frac{p-q}{1-q}m(1)<w \leq
\frac{1-p}{1-\underline{q}^1}m(\underline{q}^1) + \frac{p-\underline{q}^1}{1-\underline{q}^1}m(1) \Big\}\\
& & \bigcup \Big\{(p,w): p \in [\underline{q}^1,q], w \leq
\frac{1-p}{1-\underline{q}^1}m(\underline{q}^1) + \frac{p-\underline{q}^1}{1-\underline{q}^1}m(1) \Big\},\\
\mathcal{W}^3_q&:=&\Big\{(p,w): p \in (q,1], w  \leq \frac{1-p}{1-q}m(q) + \frac{p-q}{1-q}m(1)\Big\},\\
\mathcal{W}^4_q &:=& \mathcal{W} \setminus (\mathcal{W}^{1}_q \cup \mathcal{W}^2_q \cup \mathcal{W}^3_q).
\end{eqnarray*}

The four regions partition the set $\mathcal{W}$. Figure \ref{fig:4regions} illustrates the four regions with $\mathcal{W}_q^1$ the black region, $\mathcal{W}^2_q$ the region with vertical lines, $\mathcal{W}_q^3$ the gray region, and $\mathcal{W}_q^4$ the region with northwest lines. 
 It is worth observing that the regions $\mathcal{W}_q^1$ and $\mathcal{W}_q^4$ do not depend on the parameter $q$, while the other two do. The policy $\tau_q$ differs from one region to another.

\begin{figure}[h]
\begin{center}
\begin{tikzpicture}
\draw[->,thick] (0,0)--(0,3.5);
\draw[->,thick] (0,0) -- (6.5,0);
\draw (0,1.5)--(0.5,1)--(2,0.7)--(4,1)--(6,3);
\node[below] at (5,1.5) {$m(\cdot)$};
\draw[dashed] (0,1.5)--(6,3) node[above] {$M(\cdot)$};
\node[below] at (1,0) {$\underline{q}^1$};
\draw[dotted] (0.9,0)--(0.9,0.9);
\node[below] at (6,0) {$p$};
\node[left] at (0,3) {$w$};
\draw[dotted] (4.4,1.3)--(4.4,0) node[below] {$\overline{q}^{1}$};


\usetikzlibrary{patterns}
\draw[pattern={north west lines}] (0,1.5)--(0.9,0.9)--(6,3);
\filldraw[black](0,1.5)--(0.5,1)--(0.9,0.9);
\draw[pattern={vertical lines}](0.9,0.9)--(2,0.7)--(3.5,0.93)--(6,3);
\filldraw[gray!85!](3.5,0.93)--(4,1)-- (6,3);

\draw[dotted] (3.5,0.95)--(3.5,0) node[below] {$q$};
\end{tikzpicture}
\end{center}
\caption{The regions $\mathcal{W}_1$, $\mathcal{W}_2$, $\mathcal{W}_3$ and $\mathcal{W}_4$.}\label{fig:4regions}
\end{figure}
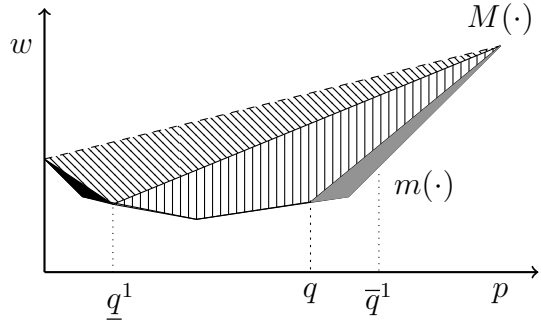
\medskip

Region $\mathcal{W}_q^1$  corresponds to the points $(p,w)$ below the line connecting $(0,m(0))$ to $(\q,m(\q))$. When $(p,w)$ is in this region, the policy splits $p$ into $0$ (i.e., discloses that the state is $\omega_0$) and $\underline{q}^1$ with probability $\frac{\q -p}{\underline{q}^1}$ and $\frac{p}{\underline{q}^1}$, respectively. Conditional on $0$, the policy recommends $a^{0}$ and promises a continuation payoff of $m(0)$. Conditional on $\underline{q}^1$, the policy recommends $a^*$ and promises a continuation payoff of $\bold{w}(\underline{q}^1)$. The agent is thus made indifferent between playing $a^*$ and receiving $\bold{w}(\underline{q}^1)$ in the future, and playing a best reply to the belief $\underline{q}^1$ forever.  Note that the  leaves rents to the agent -- since $\frac{\q -p}{\underline{q}^1}m(0)+\frac{p}{\underline{q}^1}m(\underline{q}^1)>w$. Intuitively, since $\underline{q}^1$ is the lowest belief at which the agent is willing to take action $a^*$ at the current period in exchange for full information at the next period, the principal must create a strictly positive information value at any belief $p< \underline{q}^1$.\medskip

Region $\mathcal{W}_q^2$ corresponds to the points $(p,w)$ below the line connecting $(\q,m(\q))$ and $(1,m(1))$ but above the line connecting $(q,m(q))$ and $(1,m(1))$. When $(p,w)$ is in this region, the policy splits $p$ into $1$ (i.e., reveals that the state is $\omega_1$) and $\varphi(p,w)$ with probability $1-\lambda(p,w)$ and $\lambda(p,w)$, respectively. Conditional on $\varphi(p,w)$, the policy recommends action $a^*$ and promises a continuation payoff of $\bold{w}(\varphi(p,w))$. Conditional on $1$, the policy recommends action $a^{1}$ and promises a continuation payoff of $m(1)$. In this region, the policy repeatedly incentivizes the agent to play $a^*$ with the promise of future information disclosure and 
does not leave any rents to the agent -- since $(1-\lambda(p,w))m(1)+\lambda(p,w)m(\varphi(p,w))=w$. 
 \medskip

Region $\mathcal{W}_q^3$ corresponds  to the points $(p,w)$ below the line connecting $(q,m(q))$ and $(1,m(1))$.

When $(p,w)$ is in this region, the policy splits
$p$ into $q$ and $1$ with probability $\frac{1-p}{1-q}$ and $\frac{p-q}{1-q}$, respectively. Conditional on $1$,
the policy recommends $a^{1}$ and promises a continuation payoff of $m(1)$. Conditional on $q$, the policy recommends $a^*$ and promises a continuation payoff of $\bold{w}(q)$. The agent is thus made indifferent between playing $a^*$ and receiving $\bold{w}(q)$ in the future, and playing a best reply to the belief $q$ forever. 
The policy in this region is analogous to the one in region $\mathcal{W}_q^1$. When $q=\overline{q}^1$ the reason for the analogy is immediate, as $\overline{q}^1$ is the highest belief at which the agent is willing to take action $a^*$ at the current period in exchange for full information at the next period. The reason why this policy may also be optimal for $q<\overline{q}^1$, as we shall see later, is to minimize the cost of incentivizing the agent relative to the benefit to the principal. In this region, the principal provides information value and leave rents to the agent.\medskip

Region $\mathcal{W}_q^4$ corresponds to all other points. 
When $(p,w) \in \mathcal{W}^4_q$, the policy  splits $p$ into $0$, $\underline{q}^1$, and $1$ with probability $\lambda_0$, $\lambda_{\underline{q}^1}$ and $\lambda_1$, respectively. Conditional on $0$ (resp., $1$), the policy recommends action $a^{0}$ (resp, $a^{1})$ and promises a continuation payoff of $m(0)$ (resp., $m(1)$). Conditional on $\underline{q}^1$, the policy recommends action $a^*$ and promises a continuation payoff of $\bold{w}(\underline{q}^1)$. The probabilities $(\lambda_0, \lambda_{\underline{q}^1},\lambda_1) \in \mathbb{R}_+ \times  \mathbb{R}_+  \times \mathbb{R}_+$ are the unique solution to:
\begin{equation*}
\lambda_0
\begin{pmatrix}
0 \\ m(0) \\ 1
\end{pmatrix}
+
\lambda_{\underline{q}_1}
\begin{pmatrix}
\underline{q}^1 \\ m(\underline{q}^1) \\ 1
\end{pmatrix}
+
\lambda_1
\begin{pmatrix}
1 \\ m(1) \\ 1
\end{pmatrix}
=
\begin{pmatrix}
p \\ w \\ 1
\end{pmatrix}.
\end{equation*}
A solution exists since $\mathcal{W}^4_q$ is the convex hull of $(0,m(0))$, $(\underline{q}^1,m(\underline{q}^1))$ and $(1,m(1))$. In this region, the policy leaves no rents to the agent. This completes the description of the policy $\tau_q$.\medskip

Before proving the existence of $q^* \in [\underline{q}^1,\overline{q}^1]$ such that $\tau_{q^*}$ is an optimal policy, we highlight the evolution of beliefs in $\tau_{q}$ for an arbitrary $q$. When the prior belief that the state is $\omega_1$ is $p_0<\underline{q}^1$, the policy starts by disclosing some information; with some probability it reveals that the state is $\omega_0$ and with the complementary probability it pushes the belief up to $\underline{q}^1$. Similarly, when the prior belief is $p_0>q$, the policy starts by disclosing some information; with some probability it reveals that the state is $\omega_1$ and with the complementary probability it pushes the belief down do $q$. When the prior belief is $\underline{q}^1<p_0\leq q$, and hence $(p_0,m(p_0))$ is in region $\mathcal{W}^2_q$, the policy recommends $a^*$ and promises the continuation utility $\bold{w}(p_0) \geq m(p_0)$. (Note that if $a^*$ is statically optimal at $p_0$, then   $\bold{w}(p_0) = m(p_0)$.) If the pair $(p_0, \bold{w}(p_0))$ remains in region $\mathcal{W}^2_q$, then -- unless playing $a^*$ is statically optimal for the agent at $p_0$ -- the policy discloses that the state is $\omega_1$ with some probability and with the complementary probability recommends that the agent play $a^*$ and pushes the belief down to $p_1:=\varphi(p_0,m(p_0))$, while promising utility $\bold{w}(p_1)$. The process continues until the pair of belief and promised utility reaches region $\mathcal{W}^4_q$, at which stage one last step occurs. The policy either reveals what the state is (either $\omega_0$ or $\omega_1$) or it lowers the belief to $\underline{q}^1$ and recommends $a^*$ one last time before revealing the state at the next period. If it is ever the case that the decreasing belief in region $\mathcal{W}^2_q$ reach a $p$ at which it is statically optimal for the agent to play $a^*$, then the principal reveals no further information and recommends $a^*$ forever. This is what happens in panel panel $(\textsc{B})$ of Figure \ref{fig:beliefsintro}. Panel $(\textsc{D})$, on the other hand,  represents the evolution of beliefs starting from a prior belief in the region $\mathcal{W}_q^3$, transitioning to the region $\mathcal{W}^2_q$ at the next period and staying there for three periods  and transitioning then to the region $\mathcal{W}^4_q$. Panels $(\textsc{A})$ and $(\textsc{C})$ represent similar belief evolutions, except that they start from a prior belief in region $\mathcal{W}_q^2$.  To illustrate further the policy, we now present an example.

\textbf{\textit{Example 1.}} The agent has three possible actions $a_0$, $a_1$ and $a^*$, with $a_0$ (resp., $a_1$) the agent's optimal action when the state is $\omega_0$ (resp., $\omega_1$). The prior probability of $\omega_1$ is $1/3$ and the discount factor is $1/2$. The per-period payoffs are in Table \ref{tab:ex1}, with the first coordinate corresponding to the principal's payoff.

\begin{table}[h]
\centering 
\caption{Payoff table of Example 1}\label{tab:ex1}
\begin{tabular}{|c|c|c|c|}
\hline
 & $a_0$ & $a_1$ & $a^*$\\  \hline
 $\omega_0$ & $0,1$ & $0,0$ & $1,1/2$\\ \hline
 $\omega_1$ & $0,0$ & $0,2$ & $1,1/2$ \\ \hline
 \end{tabular}
\end{table}

To describe the  policy $\tau_q$, note that  $M(p)=1+p$, $m(p)=\max(1-p,2p)$ and $\bold{w}(p)=2\max(2p,1-p)-(1/2)$. Therefore, $Q^1=[1/6,1/2]$. Assume that $q=1/2$ (we will  show that this is optimal in this example). Let us start with the pair $(p_0,m(p_0))= (1/3,2/3)$, which is in region $\mathcal{W}^2_{1/2}$. The policy recommends $a^*$ to the agent and promises a continuation payoff of $\bold{w}(1/3)=5/6$. The next state is therefore $(1/3,5/6)$, which is again in $\mathcal{W}^2_{1/2}$. If the agent had been obedient, the policy then splits the prior probability $1/3$ into  $3/11$ and $1$ with probability $22/24$ and $2/24$, respectively. To see this, note that we indeed have:

\begin{eqnarray*}
\begin{pmatrix}
\frac{1}{3} \\ \frac{5}{6}
\end{pmatrix}
=
\frac{22}{24}
\begin{pmatrix}
\frac{3}{11} \\ m(\frac{3}{11})
\end{pmatrix}
+
\frac{2}{24}
\begin{pmatrix}
1 \\ m(1)
\end{pmatrix}.
\end{eqnarray*}

Conditional on the posterior $3/11$, the policy recommends $a^*$ to the agent and promises a continuation payoff of $\bold{w}(3/11)=21/22$. Conditional on the posterior $1$, the policy recommends $a_1$ and promises a continuation payoff of $m(1)=2$. Therefore, the next state is either $(3/11,21/22)$ or $(1,2)$, with the former again in $\mathcal{W}^2_{1/2}$. In the latter case, the policy yet again recommends $a_1$ and a continuation payoff of $2$. In the former case, the policy splits  $3/11$ into $7/39$ and $1$, with probability $39/44$ and $5/44$, respectively. Conditional on the posterior $7/39$, the policy recommends $a^*$ to the agent and promises a continuation payoff of  $\bold{w}(7/39)=89/78$. Conditional on the posterior $1$, the policy recommends $a_1$ and promises a continuation payoff of $m(1)=2$. Finally, at the state $(7/39,89/78)$, which is in region $\mathcal{W}^4_{1/2}$, the policy does a penultimate split of $7/39$ into $0$, $1/6$ and $1$ with probability $113/156$ $18/156$ and $25/156$, respectively. Conditional on the posterior $1/6$, the policy recommends $a^*$ and promises a continuation payoff of $7/6$, i.e., full information disclosure at the next period. Thus, the policy fully discloses the state in finite time to the agent. See Figure \ref{fig:beliefs} for the evolution of the beliefs at the beginning of each period. At all beliefs other than $0$ and $1$, the agent is recommended to play $a^*$. The principal's expected payoff is $1285/1536$, i.e., about $0.83$.

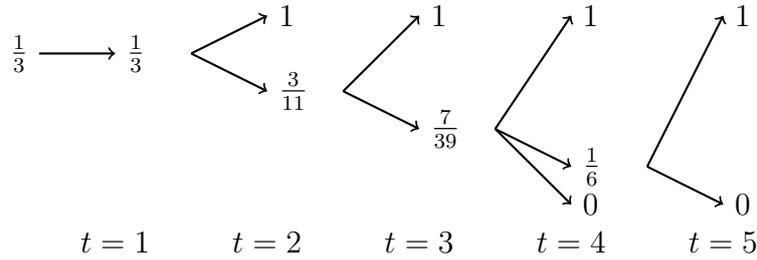
\begin{figure}[h]
\begin{center}
\begin{tikzpicture}
\node[left] at (0,2.5) {$\frac{1}{3}$};
\draw[->, thick] (0,2.5)--(1,2.5) node[right] {$\frac{1}{3}$};
\draw[->, thick] (2,2.5)--(3,3) node[right] {$1$};
\draw[->, thick] (2,2.5)--(3,2) node[right] {$\frac{3}{11}$};
\draw[->, thick] (4,2)--(5,3) node[right] {$1$};
\draw[->, thick] (4,2)--(5,1.5) node[right] {$\frac{7}{39}$};
\draw[->, thick] (6,1.5)--(7,3) node[right] {$1$};
\draw[->, thick] (6,1.5)--(7,1) node[right] {$\frac{1}{6}$};
\draw[->, thick] (6,1.5)--(7,0.5) node[right] {$0$};
\draw[->, thick] (8,1)--(9,3) node[right] {$1$};
\draw[->, thick] (8,1)--(9,0.5) node[right] {$0$};
\node[] at (1,0) {$t=1$};
\node[] at (3,0) {$t=2$};
\node[] at (5,0) {$t=3$};
\node[] at (7,0) {$t=4$};
\node[] at (9,0) {$t=5$};
\end{tikzpicture}
\end{center}
\caption{Evolution of the beliefs.}\label{fig:beliefs}
\end{figure}

\medskip

As Example 1 illustrates, the  policy repeatedly incentivizes the agent to play $a^*$ with the promise of future information disclosure. Recall that the agent's opportunity cost of playing $a^*$, when his belief is $p$, is $m(p)-u(a^*,p) \geq 0$, which is strictly positive unless $a^*$ is (statically) optimal at $p$. Thus, whenever the principal asks the agent to play $a^*$, the principal owes the agent some future compensation.

In our model, the only instrument the principal has to compensate the agent is information, and  the most value the principal can create is $M(p)-m(p)$.  Crucially, information creation is irreversible, that is, any information value created at a period is information value that cannot be created at future periods. Therefore, the best the principal can do is to create as little value of information as possible to compensate the agent and to retain the remaining value for future compensation. The best is thus to promise a continuation payoff of $\mathbf{w}(p)$. Whenever possible, this is precisely what the policy does. (Whenever this is not possible, the policy needs to disclose some information first; this happens when $p \notin Q^1$.)

Now, there are many ways to disclose just enough information to compensate the agent. The policy we construct does so by informing the agent when the state is $\omega_1$. The rationale is two-fold. First, the agent cannot be asked to play $a^*$ in all contingencies if he is to be compensated, so he must play another action. Since it is inconsequential to the principal which action the agent plays when he is not playing $a^*$,  the principal is free to choose this action. Second, the lower the agent's belief, the lower the cost of incentivizing the agent to play $a^*$ relative to the principal's benefit. Putting these two observations together, we are looking for two beliefs $(p',p'')$ such that (i) the agent is asked to play $a^*$ at $p'$, (ii) $p'<p$ since we want the agent to play $a^*$ at the lowest belief, and (iii) the probability of $p'$ is as high as possible. This problem is reminiscent of static persuasion problems. The best splitting is to have $p'$ as close as possible to $p$ and $p''$ as far as possible, that is, we want $(p',p'')=(\varphi(p,\mathbf{w}(p)),1)$.  Whenever possible, the policy we construct precisely does that. However, this is not always possible. Indeed, when $\varphi(p,\mathbf{w}(p))< \underline{q}^1$, the agent cannot be incentivized to play $a^*$ at $\varphi(p,\mathbf{w}(p))$.\footnote{Recall that $\underline{q}^1$ is the lowest belief at which the agent can be incentivized to play $a^*$.}  In that situation, the policy splits $p$ into $0$, $\underline{q}^1$ and $1$. In doing so, the principal insures that the agent will play $a^*$ one more time and yet is compensated enough.

\medskip

\subsection{Construction of $q^*$ and optimality}  Let $V_q: \mathcal{W} \rightarrow \mathbb{R}$ be the value function induced by the policy $\tau_q$. Note that for all $q$, $V_q(1,m(1)) = 0$ since $a^*$ is not optimal at $p=1$ and $V_q(0,m(0)) =0$ if $a^*$ is not optimal at $p=0$ (resp., $=v(a^*,0)$) if $a^*$ is optimal at $p=0$).  Also, $V_q(\q,m(\q)) = (1-\delta) v(a^*,\q)$ if $\q>0$ (resp., $V_q(0,m(0))=v(a^*,0)$ if $\q=0$, since $a^*$ is then optimal at $p=0$). Therefore, any two policies $\tau_q$ and $\tau_{q'}$ induce the same values at all $(p,w) \in \mathcal{W}^1_q \cup \mathcal{W}^4_q=\mathcal{W}^1_{q'} \cup \mathcal{W}^4_{q'}$. (Remember that the regions $\mathcal{W}^1_q$ and  $\mathcal{W}^4_q$ do not vary with $q$, see Figure \ref{fig:4regions}.)

Similarly, any two policies $\tau_q$ and $\tau_{q'}$ induce the same values at all $(p,w) \in \mathcal{W}^2_{\min(q,q')}$. Thus, in particular, $\tau_q$ and $\tau_{\overline{q}^1}$ induce the same values at all $(p,w) \in \mathcal{W} \setminus \mathcal{W}^3_q$. Finally, at all $(p,w) \in \mathcal{W}^3_q$, $V_q(p,w)= \frac{1-p}{1-q}V_{q}(q,m(q))= \frac{1-p}{1-q}V_{\overline{q}^1}(q,m(q))$. 
(See Section \ref{app:value-function} for more details.)  \medskip

We are now ready to state our main result. Let 
\[q^*=\sup \left\{p \in [\underline{q}^1,\overline{q}^1]: V_{\overline{q}^1}(p,m(p)) \geq V_{\overline{q}^1}(p,w) \text{\;for all\;} w \right\}.\]

\begin{theorem}\label{theo2:opti}
The policy $\tau_{q^*}$ is optimal and, therefore, $V_{q^*}=V^*$.
\end{theorem}

To understand the role of $q^*$, recall that for all $p \in [q^*,1]$, the policy leaves rents to the agent.\footnote{That is, the agent is promised a payoff of $\frac{1-p}{1-q^*}m(q^*) + \frac{p-q^*}{1-q^*} m(1)>m(p)$.} To minimize the rents left to the agent, we therefore would like to have $q^*$ as high as possible, i.e, equals to $\overline{q}^1$, the highest belief at which the agent is willing to play $a^*$ in exchange for full information disclosure at the next period. However, $V_{\overline{q}^1}(\cdot,m(\cdot))$ is not guaranteed to be concave in $p$, a necessary condition for optimality. To see that $V^*(\cdot,m(\cdot))$ must be concave in $p$,  consider any pair $(p,p') \in [0,1] \times [0,1]$ and $\alpha \in [0,1]$. We have
\begin{eqnarray*}
\alpha V^*(p,m(p)) + (1-\alpha)V^*(p',m(p')) & \leq & V^*(\alpha p + (1-\alpha) p', \alpha m(p) + (1-\alpha) m(p')) \\
&\leq &  V^*(\alpha p + (1-\alpha) p', m(\alpha p + (1-\alpha) p')),
\end{eqnarray*}
where the first inequality follows from the concavity of $V^*$ in both arguments and the second from $V^*$ decreasing in $w$ and the convexity of $m$.  The optimal choice of $q^*$ is thus the largest $q$, which guarantees $V_q(\cdot,m(\cdot))$ to be concave.

More precisely, as we show in Section \ref{app:theo2-opti}, the definition of $q^*$ guarantees that $V_{q^*}$ is concave in both arguments and decreasing in $w$, so that $V_{q^*}(\cdot,m(\cdot))$ is a concave function of $p$. We also prove that $V_{q^*}(p,m(p)) \geq V_{\overline{q}^1}(p,m(p))$ for all $p$. Since it is clearly the smallest such function, $V_{q^*}$ is the concavification of $V_{\overline{q}^1}$.  In particular, $q^*=\overline{q}^1$ if  $V_{\overline{q}^1}(\cdot,m(\cdot))$ is already concave. Figure \ref{fig:concavification} illustrates the concavification in the context of Example 1.  

\begin{figure}[h]
\centering
\includegraphics[scale=0.15]{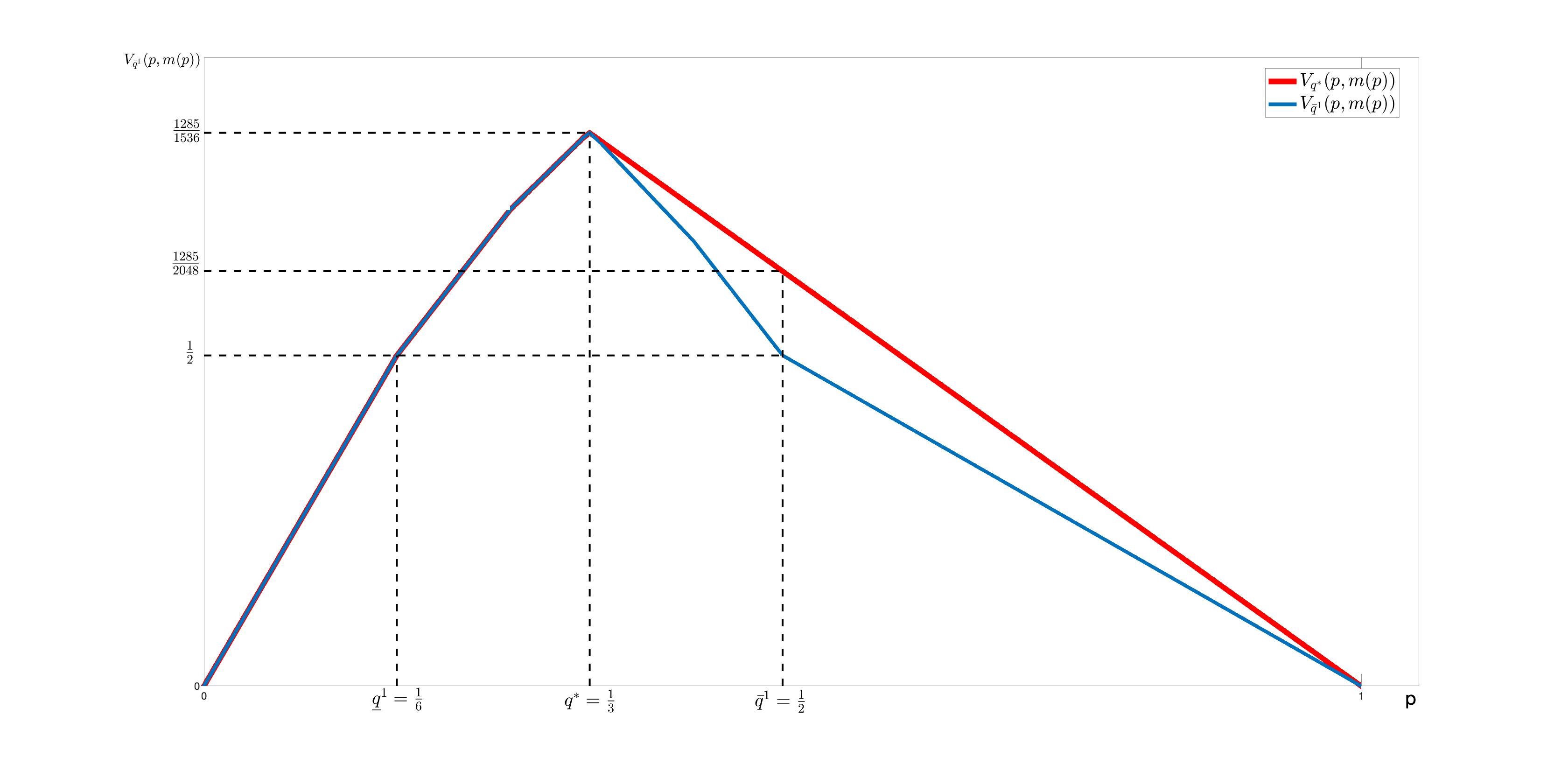}
\caption{The concavification of $V_{\overline{q}^1}(\cdot,m(\cdot))$  in Example 1}\label{fig:q*}\label{fig:concavification}
\end{figure}

The policy we construct leaves rents to the agent for all priors in $[0,\underline{q}^1) \cup(q^*,1]$, that is, the (ex-ante) participation constraint does not bind. This is quite natural for all priors in $[0,1]\ \setminus Q^1$ since the agent cannot be incentivized to play $a^*$ even once. In the language of \citet{Ely2019}, ``the goalposts need to move,'' that is, one needs to disclose information at the ex-ante stage to persuade the agent to play $a^*$. However, our policy also leaves rents for all priors in $(q^*, \overline{q}^1]$. The intuitive reason is that the initial information disclosure  reduces  the cost of incentivizing the agent in subsequent periods  sufficiently enough to compensate for the initial loss. (When the realized posterior is $1$, the agent never plays $a^*$, thus creating the loss.)

\subsection{Properties of the policy} The policy we construct has a number of noteworthy features, which we now explore.\medskip

\textbf{Information disclosure.} The policy discloses information gradually over time, with beliefs evolving until either the agent learns the state or believes that $a^*$ is (statically) optimal. We can be more specific.  First, we consider the instances when the policy converges with positive probability to a belief $p \in P=[\underline{p},\overline{p}]$, the set of beliefs at which $a^*$ is statically optimal. Let $Q^{\infty}= [\underline{p},\overline{q}^{\infty}]$, with $\overline{q}^{\infty}$  the solution to
\[m(\overline{q}^{\infty})= (1-\delta)u(a^*,\overline{q}^{\infty}) + \delta \left(\frac{1-\overline{q}^{\infty}}{1-\underline{p}}m(\underline{p})+ \frac{\overline{q}^{\infty}-\underline{p}}{1-\underline{p}}m(1)\right),\]
if $P$ is non-empty, and $Q^{\infty}= \emptyset$, otherwise. Note that $P \subseteq Q^{\infty}$. See Figure \ref{fig:q-infty} for a graphical illustration.

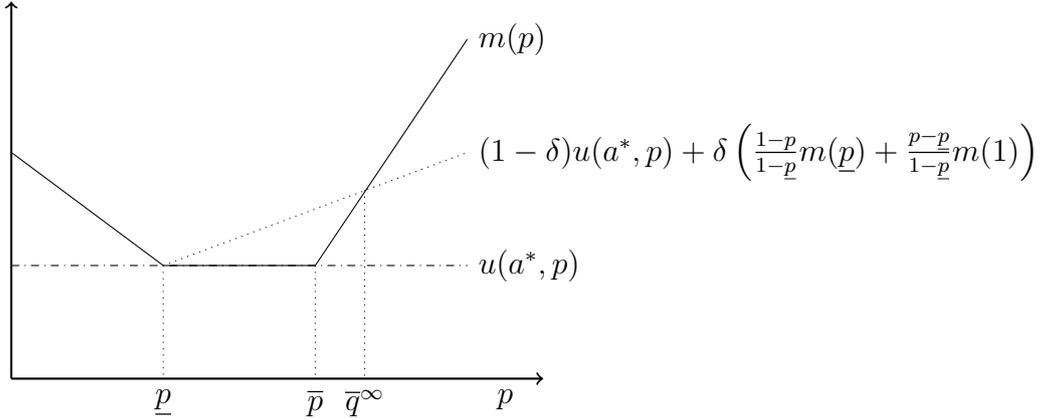
\begin{figure}[h]
\begin{center}
\begin{tikzpicture}
\draw[->,thick] (0,0)--(0,5);
\draw[->,thick] (0,0) -- (7,0);
\draw (0,3)--(2,1.5)--(4,1.5)--(6,4.5) node[right] {$m(p)$};
\draw[dashdotted] (0,1.5)--(6,1.5) node[right] {$u(a^*,p)$};
\draw[dotted](2,1.5)--(6,3) node[right]{$(1-\delta)u(a^*,p) + \delta \left(\frac{1-p}{1-\underline{p}}m(\underline{p})+ \frac{p-\underline{p}}{1-\underline{p}}m(1)\right)$ };
\draw[dotted] (4.65,2.5)--(4.65,0) node[below] {$\overline{q}^{\infty}$};
\draw[dotted] (4,1.5)--(4,0) node[below] {$\overline{p}$};
\draw[dotted] (2,1.5)--(2,0) node[below] {$\underline{p}$};
\node[below] at (6.5,0) {$p$};
\end{tikzpicture}
\end{center}
\caption{Construction of $\overline{q}^{\infty}$}\label{fig:q-infty}
\end{figure}

Intuitively, the set $Q^{\infty}$ has the ``fixed-point property,''  that is, if one starts with a prior $p \in Q^{\infty}$ and promised utility $\bold{w}(p)$, then the belief $\varphi(p,\mathbf{w}(p)) \in Q^{\infty}$. To see this, note that the pair $(p, \bold{w}(p))$ is in region $\mathcal{W}^2_q$. Since $\varphi(p,\mathbf{w}(p)) \leq p$ (with a strict inequality if $p \notin P$), we then have a decreasing sequence of beliefs converging to an element in $P$. This is because, at
all priors $p \in Q^{\infty}$, the policy splits $p$ into $p'=\varphi(p,\mathbf{w}(p))$ and $1$, then splits  $p'$ into $p''=\varphi(p',\mathbf{w}(p'))$ and $1$, etc.  The decreasing sequence $(p, p' ,p'',\dots)$ converges, either in finite time or asymptotically, to a belief in $P$, at which no further splitting occurs and the agent plays $a^*$ forever.  See panel (\textsc{b}) of Figure \ref{fig:beliefsintro} for an an illustration.

Second, at all priors in $ [0,1]\setminus Q^{\infty}$, there exists $T_{\delta} < \infty$ such that the belief process is absorbed in the degenerate beliefs $0$ or $1$ after at most $T_{\delta}$ periods. In other words, the agent learns the state for sure in finite time. The number of periods $T_{\delta}$ corresponds to the maximal number of periods the agent can be incentivized to play $a^*$.  (We provide an explicit computation in Section \ref{app:value-function}.) In Example $1$, it is $4$. Moreover, the date $T_{\delta}$ is increasing in $\delta$ and converges to $+\infty$ as $\delta$ converges to $1$. (Note that the convergence is uniform in that it does not depend on $p \in [0,1]\setminus Q^{\infty}$.)

 \medskip 


\textbf{The economics of our policy.} We now provide some further economic insights as to why our policy is optimal. Let $(\sigma,\tau)$ be a profile of strategies and denote $\mathbb{P}_{\sigma,\tau}(\cdot|\omega)$ the distribution over signals and actions induced by $(\sigma,\tau)$ conditional on $\omega$. The principal's expected payoff is: 
\begin{eqnarray*}
& (1-\delta) \sum_{\omega}\left(p_{0}(\omega) \left(\sum_{t}\sum_{s^{t},a^{t-1}}\delta^{t-1}\mathbb{P}_{\sigma,\tau}(s^{t},a^{t-1}|\omega)\tau_t(a^*|s^t,a^{t-1})\right)v(a^*,\omega)\right) = \\
& \lambda^* v^*(a^*,p^*),
\end{eqnarray*}
where 
\[\lambda^*: =(1-\delta) \sum_{\omega}p_{0}(\omega) \left(\sum_{t}\sum_{s^{t},a^{t-1}}\delta^{t-1}\mathbb{P}_{\sigma,\tau}(s^{t},a^{t-1}|\omega)\tau_t(a^*|s^t,a^{t-1})\right)\]
is the discounted probability of recommending action $a^*$ and 
\begin{eqnarray*}
p^* & :=& \frac{(1-\delta) p_{0}(\omega_1) \left(\sum_{t}\sum_{s^{t},a^{t-1}}\delta^{t-1}\mathbb{P}_{\sigma,\tau}(s^{t},a^{t-1}|\omega_1)\tau_t(a^*|s^t,a^{t-1})\right)}{\lambda^*},
\end{eqnarray*}
is the average discounted probability of $\omega_1$ when  $a^*$  is played. Notice that $p^*$ cannot be lower than $\underline{q}^1$ since the agent would never play $a^*$ at beliefs lower than $\q$. Similarly, let $p^{\dagger}$ be the average discounted probability of $\omega_1$ when $a^*$ is not recommended. Since the belief process is a martingale, $\lambda^*p^* + (1-\lambda^*)p^{\dagger}=p_0$.

We now turn our attention to the agent's expected payoff. Since the agent's static payoff is bounded from above by $M(p)$ when his belief is $p$,   his ex-ante expected payoff is bounded from above by: 
\begin{equation}\label{eq:upper-bound}
\lambda^* u(a^*,p^*)+ (1-\lambda^*)M(p^{\dagger})= \lambda^*[u(a^*,p^*)-M(p^*)]+ M(p_0),
\end{equation}
where we use the linearity of $u(a^*,\cdot)$ and $M(\cdot)$. Moreover, since the agent's payoff must be at least $m(p_0)$, there exists a positive number $c \geq 0$ such that 
\[\lambda^*(u(a^*,p^*)-M(p^*))+ M(p_0) -c=m(p_0). \] 
The number $c$ captures two effects. First, the optimal solution may leave some rents to the agent, so that the agent's payoff is $m(p_0)+c_1$ for some $c_1 \geq 0$. Second, the agent's actual payoff is bounded away from the upper bound derived in (\ref{eq:upper-bound}) by some positive number $c_2\geq 0$. Thus, $c=c_1+c_2$.

We can then rewrite the principal's expected payoff as: 
\[\frac{ v(a^*,p^*)}{M(p^*)-u(a^*,p^*)}(M(p_0)-m(p_0)-c). \]
The first term  captures the  benefit of incentivizing the agent to play $a^*$ relative to the cost. Since $\frac{v(a^*,0)}{v(a^*,1)} \geq \frac{m(0)-u(a^*,0)}{m(1)-u(a^*,1)}$,  it is decreasing in $p^*$.\footnote{This follows from the observation that $M(p^*)-u(a^*,p^*)=p^*[(m(1)-m(0))- (u(a^*,1)-u(a^*,0))]+m(0)-u(a^*,0)$, $v(a^*,p^*)= p^*(v(a^*,1)-v(a^*,0))+v(a^*,0)$, and simple algebra.} Ceteris paribus, the lower the average belief at which the agent plays $a^*$, the higher the principal's expected payoff.  In a sense, this term represents the ``debt'' the principal accumulates  over time as the agent repeatedly plays $a^*$, that is, this is how much the principal owes the agent in exchange for him playing $a^*$. The second term captures how the principal repays his debt with his only instrument: information. The term $M(p_0)-m(p_0)$ is the maximal value of information the principal can create. Ceteris paribus, the principal's payoff is decreasing in $c$, that is, the best is to leave no rents to the agent and to create as much information as necessary to repay the agent. Notice that $c=0$ is only achieved by both leaving no rents to the agent and having the agent informed of the state when he does not play $a^*$.  

The principal, however, cannot freely choose $p^*$ and $c$: he needs to trade-off  a lower $p^*$ for a lower $c$. In other words, inducing a lower belief at which the agent plays $a^*$ necessitates to compensate the agent with some information creation. We can now understand better our policy. For $p_0 \in [0,\underline{q}^1]$, our policy consists in splitting $p_0$ into $0$ and $\underline{q}^1$. At $\q$, the agent is incentivized to plays $a^*$ with the promise of full information disclosure at the next period. The policy thus maximizes the benefit of incentivizing the agent to play $a^*$ relative to the cost (i.e., $p^*=\q$ is the lowest possible), attains the upper bound derived in (\ref{eq:upper-bound}) for the agent's expected payoff (i.e., $c_2=0$), and leaves as little rents as possible (i.e., $c_1$ is minimized).  A similar argument applies to $p_0 \in [\overline{q}^1,1]$. For all $p_0 \in [\underline{q}^1,q^*]$, the policy leaves no rents to the agent and uses to the full extent possible the information available to the principal, i.e., $c=0$. In addition, the policy minimizes $p^*$ conditional on $c=0$. Finally, for all $p_0 \in (q^*,\overline{q}^1]$, a similar logic applies. The policy leaves strictly positive rents to the agent, i.e., $c_1>0$. The gain is to reduce the average cost of incentivizing $a^*$ in the future, which out-weights the cost.  This is achieved by lowering the agent's average belief.

\section{Final Discussion}\label{sec:final}
We conclude the paper with an informal comparison of our policy with two other policies, which feature prominently in the dynamic persuasion literature. (See Appendix \ref{app:other-policies}  for a formal discussion.) \medskip 

We first consider the Kamenica-Gentzkow's policy (for short, KG's policy). The KG's policy aims to persuade the agent to choose $a^*$ as often as possible by disclosing information at the initial stage only. Consider Example 1. If the principal commits to disclose information at the initial stage only, the principal's payoff is 0 (since $a^*$ is never statically optimal). The KG's policy is clearly not optimal in this example. The intuition for why is straightforward. The principal would gain by conditioning his information disclosure on the agent's actions, but this is not permitted by the KG's policy. (As we show in Appendix \ref{app:other-policies}, the KG's policy is, however, optimal in all problems with two actions.) \medskip 

Another policy the principal could commit to is to fully disclose the state with  probability $\alpha$ at period $t$ (and withhold all information with the complementary probability) if the agent plays $a^*$ at period $t-1$. (If the agent deviates, the harshest  punishment is to withhold all information in all subsequent periods.) This is a recursive policy. Consider again Example 1. Since the normalized expected payoff from the harshest punishment is $2/3$, if we write $V$ (resp., $U$) for the principal (resp., agent) payoff, the best such policy is to choose $\alpha$ so as to maximize
\[V=\frac{1}{2}1+ \frac{1}{2} (1-\alpha) V, \]
subject to
\[U=\frac{1}{2}\Big(\frac{1}{2}\Big)+ \frac{1}{2}\Big[(1-\alpha)U + \alpha \frac{4}{3} \Big] \geq \frac{2}{3}. \]
The principal's best payoff is $V=4/5$ with $\alpha=1/4$. This policy 
is not optimal. As we have seen, the principal obtains a payoff of $V=1285/1536$ at the optimum. The key intuition why this is so is that the policy of random full disclosure does not exploit the asymmetries in the opportunity cost of playing the principal's most preferred action with respect to the state. In fact, we show in Appendix \ref{app:other-policies} that if there are no asymmetries, i.e.,$\frac{v(a^*,0)}{v(a^*,1)} = \frac{m(0)-u(a^*,0)}{m(1)-u(a^*,1)}$, then the random policy is also optimal. \medskip

An analogous policy, which is also not optimal, is to ``reward'' the agent with full information disclosure for playing $a^*$ sufficiently often at the beginning of the relationship, say up to period $T^*$. This policy of fully disclosing the state with delay plays a prominent role in the work of \citet{Ball2019dynamic} and \citet{orlov2018persuading}. 
In Example 1, this policy selects the largest $T^*$ such that
\[(1-\delta)\left(\frac{1}{2}\big(\delta^0+\delta^1+\dots+\delta^{T^*-1}\big) + \frac{4}{3}\big(\delta^{T^*}+\dots \big) \right)\geq \frac{2}{3},\]
With such a simple strategy, $T^*=\lfloor \ln(5)/\ln(2) \rfloor = 2$ and the principal's payoff is $3/4$. This  policy does worse than the policy of random full disclosure, but only because of an ``integer constraint''.  Intuitively, random full disclosure performs better  because it makes it possible to incentivize the agent to play $a^*$ a discounted number of periods slightly larger than $2$, namely $\ln(5)/\ln(2)$. (In continuous time, the two policies would be equivalent.)

\appendix
\section{Proofs}
\subsection{Mathematical preliminaries}\label{math:prelim}
We collect without proofs some useful results about concave functions. Let $f: [a,b] \rightarrow \mathbb{R}$ be a concave function and $a \leq x < y <z \leq b$. The following properties hold: 
\begin{itemize}
\item[(a)] $\frac{f(y)-f(x)}{y-x} \geq \frac{f(z)-f(y)}{z-y}$,
\item[(b)] $\frac{f(y)-f(a)}{y-a} \geq \frac{f(z)-f(a)}{z-a}$,
\item[(c)] $\frac{f(b)-f(x)}{b-x} \geq \frac{f(b)-f(y)}{b-y}$.
\item[(d)] $\frac{f(y)-f(x)}{y-x}\geq \frac{f(y+\Delta)-f(x+\Delta)}{y-x}$ 
for all $\Delta \geq 0$ such that $y + \Delta \leq b$. 
\end{itemize}
Note that property (a) implies (d) and is true irrespective of whether $x+\Delta \gtreqqless y$. We will repeatedly use these properties in most of the following proofs. 
\medskip

To prove Lemma \ref{V-concave-in-w}, we will use the following property: if $f: [a,b] \rightarrow \mathbb{R}$ satisfies $\frac{f(x)-f(a)}{x-a} \geq \frac{f(y)-f(a)}{y-a}$ for all $a < x \leq y \leq b$, then $f$ is concave.

\subsection{Recursive formulation: Theorem 4 of \citet[p. 44]{Ely2015}}\label{app-recursive-formulation}
We first note that the operator $T$ is monotone, i.e., for all $V \geq V'$, $T(V) \geq T(V')$. It also satisfies $T(V + c) \leq T(V) + \delta c$ for all positive constant $c\geq 0$, for all $V$. Hence, it is indeed a contraction by Blackwell's theorem. \medskip

\citet{Ely2015}  proves that the principal's maximal payoff is $\max_{w \in [m(p_0),M(p_0)]}\widehat{V}^{*}(p_0,w)$, with $\widehat{V}^{*}$ the unique fixed point of the contraction $\widehat{T}$, with $\widehat{T}$ differing from $T$ in that the promise-keeping constraint is as an equality; all other constraints are the same.  Note that the operator $\widehat{T}$ is also monotone.

We now argue that  that $V^*(p_0,m(p_0))= \max_{w \in [m(p_0),M(p_0)]}\widehat{V}^{*}(p_0,w)$. (Note that we are not arguing that $T =\widehat{T}$.)

As a preliminary observation, note that $T(V)(p,w) \geq \widehat{T}(V)(p,w)$ for all $(p,w) \in \mathcal{W}$, for all $V$. Let $w_0 \in \arg\max_{w \in [m(p_0),M(p_0)]}\widehat{V}^{*}(p_0,w)$.  We have that
  \begin{eqnarray*}
  V^*(p_0,m(p_0)) \geq V^*(p_0,w_0) = T(V^*)(p_0,w_0) &\geq & \widehat{T}(V^*)(p_0,w_0) \geq  
  \widehat{T}^2(V^*)(p_0,w_0) \geq \dots \geq \\
&   \geq &   \lim_{n \rightarrow \infty}\widehat{T}^n(V^*)(p_0,w_0) = \widehat{V}^*(p_0,w_0), 
  \end{eqnarray*}
where the first inequality follows from  $V^*$ being decreasing in $w$. 

Conversely, let $(\lambda^*_s, p^*_s,w^*_s,a^*_s)_{s \in S}$ be a maximizer of $T(V^*)(p_0,m(p_0))$. We have that
\[ M(p_0) = \sum_{s \in S}\lambda^*_s M(p^*_s) \geq \sum_{s \in S} \lambda^*_s[ (1-\delta)u(a^*_s,p^*_s)+ \delta w^*_s] :=\widehat{w}_0 \geq \ \sum_{s \in S}\lambda^*_s m(p^*_s) \geq m(p_0),\]
hence $(\lambda^*_s, p^*_s,w^*_s,a^*_s)_{s \in S}$ is a maximizer for $T(\widehat{V}^*)(p_0,\widehat{w}_0)$ and, consequently, \[V^*(p_0,m(p_0))  = \widehat{V}^*(p_0,\widehat{w}_0) \leq \max_{w \in [m(p_0),M(p_0)]}\widehat{V}^{*}(p_0,w).\]

\subsection{Proposition \ref{prop2}}
We break Proposition \ref{prop2} into several lemmata.

\begin{lemma}
    \label{lem:state-wise optimality} Let $(\lambda_s, p_s,w_s,a_s)_{s \in S}$ be a solution to the maximization program $T(V^*)(p,w)$.
    For all $s \in S$ such that $\lambda_s >0$, we have
    \[\left( {1 - \delta } \right)v\left( {{a_s},{p_s}} \right) + \delta {V^ * }\left( {{p_s},{w_s}} \right) = {V^ * }\left( {{p_s},\left( {1 - \delta } \right)u\left( {{a_s},{p_s}} \right) + \delta {w_s}} \right).\]
\end{lemma}
\begin{proof} By contradiction, assume that there exists $s' \in S$ such that $\lambda_{s'}>0$ and \[\left( {1 - \delta } \right)v\left( a_{s'},p_{s'} \right) + \delta V^ * \left( p_{s'},w_{s'} \right) < V^ * \left( p_{s'},\left( 1 - \delta \right)u\left( a_{s'},p_{s'} \right) + \delta w_{s'} \right).\]
Let $(\lambda^*_s, p^*_s,w^*_s,a^*_s)_{s \in S}$ be the policy, which achieves $V^*(p_{s'}, (1-\delta)u(a_{s'},p_{s'}) + \delta w_{s'})$, and consider the new policy
\[((\lambda_s, p_s,w_s,a_s)_{s \in S \setminus \{s'\}}, (\lambda_{s'} \lambda^*_s, p^*_s,w^*_s,a^*_s)_{s \in S}).\]
By construction, the new policy is feasible. Moreover, we have that
\begin{eqnarray*}
\sum_{s \in S \setminus \{s'\}} \lambda_s [(1-\delta)v(a_s,p_s)+ \delta V^* (p_s,w_s)]  + \lambda_{s'}
\sum_{s \in S} \lambda^*_s [(1-\delta)v(a^*_s,p^*_s)+ \delta V^* (p^*_s,w^*_s)] = \\
\sum_{s \in S \setminus \{s'\}} \lambda_s [(1-\delta)v(a_s,p_s)+ \delta V^* (p_s,w_s)]  + \lambda_{s'}
V^*(p_{s'}, (1-\delta)u(a_{s'},p_{s'}) + \delta w_{s'})> \\
\sum_{s \in S} \lambda_s [(1-\delta)v(a_s,p_s)+ \delta V^* (p_s,w_s)],
\end{eqnarray*}
a contradiction with the optimality of $(\lambda_s, p_s,w_s,a_s)_{s \in S}$.

Since  the fixed point satisfies $V^*(p_s,(1-\delta)u(a_s,p_s)+\delta w_s) \geq (1-\delta)v(a_s,p_s)+\delta V^*(p_s,w_s)$, we have the desired result.
\end{proof}

\begin{lemma}
    \label{lem:shirk at boundary}
   Let $(\lambda_s, p_s,w_s,a_s)_{s \in S}$ be a solution to the maximization program $T(V^*)(p,w)$.
    For all $s \in S$ such that $\lambda_s >0$, $V^*(p_s,w_s)=0$ if $a_s \neq a^*$.\end{lemma}
\begin{proof}
Let $s \in S$ such that $\lambda_s>0$ and $a_s \neq a^*$. We have
\begin{eqnarray*}
(1-\delta)v(a_s,p_s)+ \delta V^*(p_s,w_s) =  \delta V^*(p_s,w_s) \geq V^*(p_s,(1-\delta)u(a_s,p_s)+\delta w_s) \geq V^*(p_s, w_s),
\end{eqnarray*}
where the first inequality follows from  Lemma \ref{lem:state-wise optimality} and the second follows from $V^*$ decreasing in $w$ and $w_s \geq
u(a_s,p_s)$ for
\[(1-\delta)u(a_s,p_s)+\delta w_s \geq m(p_s), \]
to hold. It follows that $V^*(p_s,w_s) =0$.
\end{proof}

\begin{lemma}
    \label{lem:effort at one signal}
    Let $(\lambda'_s, p'_s,w'_s,a'_s)_{s \in S'}$ be a solution to the maximization program $T(V^*)(p,w)$. There exists another solution $(\lambda_s, p_s,w_s,a_s)_{s \in S}$ such that $a_s=a^*$ for at most one $s \in S$ with $\lambda_s>0$.
\end{lemma}
\begin{proof}
Let $(\lambda'_s, p'_s,w'_s,a'_s)_{s \in S'}$ be a solution to the maximization program $T(V^*)(p,w)$. Let $S^* \subseteq S'$ be the set of signals such that $a_s=a^*$ and $\lambda_s>0$. If $S^*$ is empty, there is nothing to prove. If $S^*$ is non-empty, define $p^*$ as
\[\sum_{s \in S^*} \Big(\frac{\lambda'_s}{\sum_{s \in S^*} \lambda'_s}\Big) p_s =p^*,\]
and $\sum_{s \in S^*} \lambda'_s = \lambda^*$.
From the concavity of $V^*$, we have that
\begin{eqnarray*}
\sum_{s \in S^*} \lambda'_s (v(a^*,p'_s) (1-\delta)+ \delta V^*(p'_s,w'_s)) & = & \lambda^* \Big(v(a^*,p^*)(1-\delta) + \delta
\sum_{s \in S^*}\Big(\frac{\lambda'_s}{\lambda^*}\Big)V(p'_s,w'_s)\Big) \\
&\leq &   \lambda^* \Big(v(a^*,p^*)(1-\delta) + \delta
V(p^*,w^*)\Big),
\end{eqnarray*}
where
\[w^*= \sum_{s \in S^*} \Big(\frac{\lambda'_s}{\sum_{s \in S^*} \lambda'_s}\Big) w'_s.\]
Notice that $w^* \in [m(p^*),M(p^*)]$ since the convexity of $m$ implies
\[ M(p^*)= \sum_{s \in S^*}\Big(\frac{\lambda'_s}{\sum_{s \in S^*} \lambda'_s}\Big) M(p'_s) \geq \sum_{s \in S^*}\Big(\frac{\lambda'_s}{\sum_{s \in S^*} \lambda'_s}\Big) w_s \geq  \sum_{s \in S^*}\Big(\frac{\lambda'_s}{\sum_{s \in S^*} \lambda'_s}\Big) m(p'_s) \geq
m(p^*).\]

It is routine  to verify that the new contract
\[((\lambda'_s, p'_s,w'_s,a'_s)_{s \in S'\setminus S^*},(\lambda^*,p^*,a^*,w^*)) \]
is feasible and, therefore, also optimal. \end{proof}

\begin{lemma}
    \label{lem:never back-load incentive}
   Let $(\lambda'_s, p'_s,w'_s,a'_s)_{s \in S'}$ be a solution to the maximization program $T(V^*)(p,w)$. There exists another solution $(\lambda_s, p_s,w_s,a_s)_{s \in S}$ such that
   \[(1-\delta)u(a_s,p_s)+\delta w_s= m(p_s), \]
  for all $s$ such that $\lambda_s>0$ and $a_s=a^*$.
   \end{lemma}
\begin{proof}
Assume that there exists $s^*$ such that $\lambda'_{s^*}>0$, $a_{s^*}=a^*$ and
\[(1-\delta)u(a'_{s^*},p'_{s^*})+\delta (w'_{s^*}-\varepsilon)\geq  m(p'_{s^*})\]
for some $\varepsilon>0$.
From Lemma \ref{lem:effort at one signal}, we may assume that there is a single such $s^*$. Consider the new tuple $(\lambda_s, p_s,w_s,a_s)_{s \in S}$, where $w_{s^*}=w'_{s^*} -\varepsilon$, $w_{\tilde{s}} = w'_{\tilde{s}} +\frac{\lambda_{s^*}}{\lambda_{\tilde{s}}}\varepsilon$ for some $\tilde{s}\neq s^*$ such that $\lambda_{\tilde{s}}>0$, $w_s=w'_s$ for all $s \in S \setminus\{s^*,\tilde{s}\}$, and $(\lambda_s,p_s,a_s)=(\lambda'_s,p'_s,a'_s)$ for all $s$. This new contract is  feasible and increases the principal's payoff.
\end{proof}

\subsection{Value functions }\label{app:value-function}
This section characterizes the value function $V_q$ induced by the policy $\tau_q$. As explained in the text, it is enough to characterize $V_{\overline{q}^1}$. We first start with the definition of important subsets of $[0,1]$.\medskip 

\subsubsection{Construction of the sets $Q^k$}

Let $Q^0: =[0,1]$. We define inductively the set $Q^k \subseteq [0,1]$, $k \geq 0$.   We write $\underline{q}^k$ (resp., $\overline{q}^k$) for $\inf Q^k$ (resp., $\sup Q^k$).  For any $ k \geq 0$, define the function $U^k: [\underline{q}^k,1] \rightarrow \mathbb{R}$:
\begin{equation*}
U^k(q) := \frac{1-q}{1-\underline{q}^k}m(\underline{q}^k) + \frac{q-\underline{q}^k}{1-\underline{q}^k}m(1),
\end{equation*}
with the convention that $U^k \equiv m(1)$ if $\underline{q}^k=1$.
Note that $U^0(q)=M(q)$ and $U^k(q) \geq m(q)$ for all $k$. We define $Q^{k+1}$ as follows:
\begin{eqnarray*}
Q^{k+1}& = & \{q \in Q^k: (1-\delta)u(a^*,q)+ \delta U^k(q) \geq m(q) \}.
\end{eqnarray*}

For a graphical illustration, see Figure \ref{fig:threshold}.

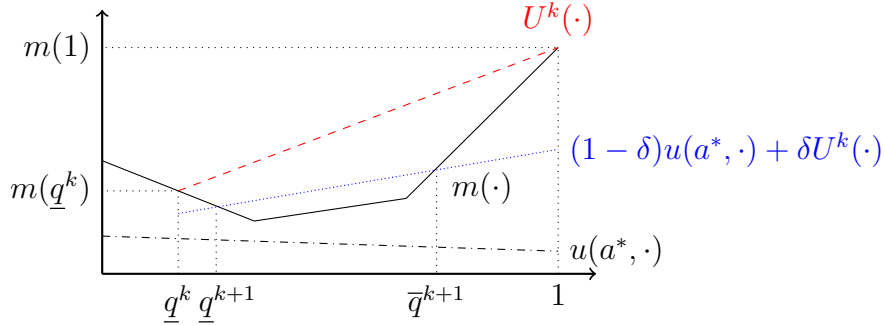
\begin{figure}[h]
\begin{center}
\begin{tikzpicture}
\draw[->,thick] (0,0)--(0,3.5);
\draw[->,thick] (0,0) -- (6.5,0);
\draw (0,1.5)--(2,0.7)--(4,1)--(6,3);
\node[below] at (5,1.5) {$m(\cdot)$};
\draw[dashed, red] (1,1.1)--(6,3) node[above] {$U^{k}(\cdot)$};
\draw[dashdotted] (0,0.5)--(6,0.3) node[right] {$u(a^*,\cdot)$};
\node[below] at (1,0) {$\underline{q}^k$};
\node[left] at (0,1.1) {$m(\underline{q}^k)$};
\draw[dotted] (1,0)--(1,1.1)--(0,1.1);
\node[below] at (6,0) {$1$};
\node[left] at (0,3) {$m(1)$};
\draw[dotted] (6,0)--(6,3)--(0,3);
\draw[densely dotted, blue] (1,0.8)--(6,1.65) node[right] {$(1-\delta) u(a^*,\cdot) + \delta U^k(\cdot)$};
\draw[dotted] (1.5,0.8)--(1.5,0) node[below] at (1.65,0) {$\underline{q}^{k+1}$};
\draw[dotted] (4.4,1.3)--(4.4,0) node[below] {$\overline{q}^{k+1}$};
\end{tikzpicture}
\end{center}
\caption{Construction of the thresholds}\label{fig:threshold}
\end{figure}

Few observations are worth making. First, we have that $ P \subseteq Q^k$ for all $k$.
Second, we have a decreasing sequence, i.e., $Q^{k+1} \subseteq Q^k$ for all $k$. Third, if $Q^k$ and $P$ are non-empty, then they are closed intervals. Fourth, the limit  $Q^{\infty}= \lim_{k \rightarrow \infty}  Q^k= \bigcap_{k} Q^k$ exists and includes $P$.  Moreover, if $P \neq \emptyset$, then $\underline{q}^{\infty}=\underline{p}$, where $\underline{p}:=\inf P$. If $P= \emptyset$, then $Q^{\infty} = \emptyset$. Consequently, there exists $k^*< \infty$ such that $\emptyset= Q^{k^*+1} \subset Q^{k^*} \neq \emptyset$.\medskip

The first to the third observations are readily proved, so we concentrate on the proof of the fourth observation. The limit exists as we have a decreasing sequence of  sets. \medskip

We  prove that if $P = \emptyset$, then $Q^{\infty}=\emptyset$. So, assume that $P=\emptyset$. We first argue that it cannot be that $Q^k=Q^{k-1} \neq \emptyset$ for some $k \geq 0$. To the contrary, assume that $Q^k=Q^{k-1} \neq \emptyset$ for some $k \geq 0$, hence $Q^{k'} =Q^{k-1}$ for all $k' \geq k$. From the convexity and continuity of $m$ and the linearity of $u$, $Q^{k-1}$ is the closed interval $[\underline{q}^{k-1}, \overline{q}^{k-1}]$, with the two boundary points solution to
\[(1-\delta)u(a^*,q)+ \delta U^{k-2}(q) = m(q). \]
Therefore, if $(\underline{q}^{k},\overline{q}^k)= (\underline{q}^{k-1}, \overline{q}^{k-1})$, we have that:
\begin{eqnarray*}
m(\underline{q}^{k-1}) & = & (1-\delta)u(a^*,\underline{q}^{k-1})+\delta m(\underline{q}^{k-1}),\\
m(\overline{q}^{k-1})& = & (1-\delta)u(a^*,\overline{q}^{k-1})+ \delta \Big[\frac{1-\overline{q}
^{k-1}}{1-\underline{q}^{k-1}}m(\underline{q}^{k-1}) + \frac{\overline{q}^{k-1}-\underline{q}^{k-1}}{1-\underline{q}^{k-1}}m(1)\Big],\\
 & \leq & 
(1-\delta)u(a^*,\overline{q}^{k-1})+\delta m(\overline{q}^{k-1}).
\end{eqnarray*}
This implies that $u(a^*,\underline{q}^{k-1})= m(\underline{q}^{k-1})$ and $u(a^*,\overline{q}^{k-1})=m(\overline{q}^{k-1})$ and, therefore,
$\emptyset \neq Q^{k-1} \subseteq P$, a contradiction. 

We thus  have an infinite sequence of strictly decreasing non-empty closed intervals. Let $\varepsilon:=\min_{p \in [0,1]} m(p)-u(a^*,p)$. Since $P =\emptyset$, we have that $\varepsilon>0$. For all $p \in Q^{\infty}$, for all $k$, 
\begin{eqnarray*}
m(p) &\leq& (1-\delta) u(a^*,p)+ \delta U^k(p),\\
&\leq & (1-\delta) (m(p) - \varepsilon) + \delta U^k(p).\\
\end{eqnarray*}
Assume that $Q^{\infty}$ is non-empty and let $\underline{q}^{\infty}$ its greatest lower bound. Since $\underline{q}^{\infty} \in Q^k$ for all $k$, we have that $U^k(\underline{q}^{\infty}) \geq m(\underline{q}^{\infty}) + \varepsilon (1-\delta)/\delta$ for all $k$.  
Since $\lim_{k \rightarrow \infty} U^k(\underline{q}^{\infty}) = m(\underline{q}^{\infty})$, we have that $m(\underline{q}^{\infty}) \geq m(\underline{q}^{\infty}) + \varepsilon (1-\delta)/\delta$, a contradiction.\medskip 

We now prove that if $P \neq \emptyset$, then $\underline{q}^{\infty} = \underline{p}$. From above, we have that if $Q^k=Q^{k-1} \neq \emptyset$ for some $k \geq 0$, hence $Q^{k'} =Q^{k-1}$ for all $k' \geq k$, then $P=Q^k$ since $P \subseteq Q^k$. If we have an infinite sequence of strictly decreasing sets, for all $q \in Q^{\infty}$, 
\begin{equation*}
(1-\delta) u(a^*,q) + \delta \Big[\frac{1-q}{1-\underline{q}^{\infty}}m(\underline{q}^{\infty})+ \frac{q-\underline{q}^{\infty}}{1-\underline{q}
^{\infty}}m(1)\Big] \geq m(q).
\end{equation*}
Taking the limit  $q \downarrow \underline{q}^{\infty}$, we obtain that $u(a^*,\underline{q}^{\infty}) =m(\underline{q}^{\infty})$, i.e., $\underline{q}^{\infty} \in P$.  Hence, $\underline{q}^{\infty}=\underline{p}$.

\subsubsection{Value functions} We first derive $V_{\overline{q}^1}$for all $(p,w) \in \mathcal{W} \setminus \mathcal{W}^2_{\overline{q}^1}$. 

To start with, $V_{\overline{q}^1}(1,m(1)) = 0$ since $a^*$ is not optimal at $p=1$. Similarly, $V_{\overline{q}^1}(0,m(0)) =0$ if $a^*$ is not optimal at $p=0$, while $V_{\overline{q}^1}(0,m(0)=v(a^*,0)$ if $a^*$ is optimal at $p=0$.  Also, $V_{\overline{q}^1}(\q,m(\q)) = (1-\delta) v(a^*,\q)$ if $\q>0$; while  $V_{\overline{q}^1}(0,m(0))=v(a^*,0)$ if $\q=0$, since $a^*$ is then optimal at $p=0$.

 With the function $V_{\overline{q}^1}$ defined at these three points, it is then defined at all points $(p,w)$ in $\mathcal{W}^1_{\overline{q}^1} \cup \mathcal{W}^4_{\overline{q}^1}$. In particular, it is easy to show that 
\[V_{\overline{q}^1}(\q,w)=\frac{M(\q)-w}{M(\q)-m(\q)}(1-\delta) v(a^*,\q) =  \frac{M(\q)-w}{M(\q)-u(a^*,\q)} v(a^*,\q),\]
for all $w \in [m(\q),M(\q)]$.

At all points $(p,w) \in \mathcal{W}^3_{\overline{q}^1}$, 
\[V_{\overline{q}^1}(p,w)=\frac{1-p}{1-\overline{q}^1} V_{\overline{q}^1}(\overline{q}^1,m(\overline{q}^1)). \]
Therefore, $V_{\overline{q}^1}$ is well-defined at all $(p,w) \in \mathcal{W} \setminus \mathcal{W}^2_{\overline{q}^1}$.
\medskip

At all points $(p,w) \in \mathcal{W}^2_{\overline{q}^1}$,  $V_{\overline{q}^1}(p,w)$ is defined via the recursive equation: 
\[V_{\overline{q}^1}(p,w)=\lambda(p,w)[(1-\delta)v(a^*,\varphi(p,w))+ \delta V_{\overline{q}^1}(\varphi(p,w),\bold{w}(\varphi(p,w))]=\lambda(p,w)V_{\overline{q}^1}(\varphi(p,w),m(\varphi(p,w))).\] 

Since $V_{\overline{q}^1}(p,w)= \lambda(p,w) V_{\overline{q}^1}(\varphi(p,w),m(\varphi(p,w))$, the value function is well-defined at all $(p,w)$ if it is well-defined at all $(p,m(p))$, which we now prove.\medskip

By construction of the sets $Q^k$, observe that if $p \in Q^k \setminus Q^{k+1}$, then $\bold{w}(p) \in (U^k(p),U^{k+1}(p)]$ and, therefore, $\varphi(p,\bold{w}(p)) \in [\underline{q}^{k-1},\underline{q}^k) \subset Q^{k-1}\setminus Q^{k}$. Moreover, $\varphi(\overline{q}^{k},\bold{w}(\overline{q}^k))=\underline{q}^k$. We now use these observations to complete the derivation of $V_{\overline{q}^1}$. \medskip

For all $p \in Q^1\setminus Q^2$, we have that $\bold{w}(p) \in Q^0 \setminus Q^1$, so that $(p,\bold{w}(p)) \in \mathcal{W}^4_{\overline{q}^1}$.  Since 
\[ V_{\overline{q}^1}(p,m(p))  =  (1-\delta)v(a^*,p) + \delta V_{\overline{q}^1}(p,\bold{w}(p)),\]
$V_{\overline{q}^1}(p,m(p))$ is well-defined for all $p \in Q^1 \setminus Q^2$. By induction, assume that it is well-defined for all $p \in \bigcup_{\ell<k}Q^{\ell}\setminus Q^{\ell+1}$. We argue that it is well-defined for all $p \in Q^{k}\setminus Q^{k+1}$. Fix any $p \in Q^k \setminus Q^{k+1}$. From our initial observation, $\varphi(p,\bold{w}(p)) \in [\underline{q}^{k-1},\underline{q}^k)$ and, therefore, $V_{\overline{q}^1}(p,m(p))$ is well-defined since
\begin{eqnarray*}
V_{\overline{q}^1}(p,m(p)) & = & (1-\delta)v(a^*,p) + \delta V_{\overline{q}^1}(p,\bold{w}(p)) \\ 
& = & (1-\delta)v(a^*,p)+ \lambda(p,\bold{w}(p)) \underbrace{V_{\overline{q}^1}(\varphi(p,\bold{w}(p)),m(\varphi(p,\bold{w}(p))))}_{\text{defined by the induction step}}.
\end{eqnarray*}

Therefore, $V_{\overline{q}^1}(p,m(p))$ is well-defined for all $p \in \bigcup_{\ell}Q^{\ell}\setminus Q^{\ell+1} = Q^1 \setminus Q^{\infty}$. It remains to argue that it is well-defined for all $p \in Q^{\infty}$. \medskip 

 From the definition of $Q^{\infty}$, we have that  
$\bold{w}(p) \leq \frac{1-p}{1-\underline{q}^{\infty}} m(\underline{q}^{\infty}) + \frac{p-\underline{q}^{\infty}}{1-\underline{q}^{\infty}}m(1)$ and, therefore, $\varphi(p,\bold{w}(p)) \in Q^{\infty}$. In other words, if $p \in Q^{\infty}$, then $\varphi(p,\bold{w}(p)) \in Q^{\infty}$, so that the restriction of $V_{\overline{q}^1}(\cdot,m(\cdot))$ to $Q^{\infty}$ is entirely defined by its value on $Q^{\infty}$ via the contraction: 
\begin{eqnarray*}
V_{\overline{q}^1}(p,m(p)) & = & (1-\delta) v(a^*,p)+ \delta \lambda(p,\bold{w}(p))V_{\overline{q}^1}(\varphi(p,\bold{w}(p)),m(\varphi(p,\bold{w}(p))). 
\end{eqnarray*}
The unique solution to this fixed point problem is given by:
\[V_{\overline{q}^1}(p,m(p))= v(a^*,p)-\frac{m(p)-u(a^*,p)}{m(1)-u(a^*,1)}v(a^*,1),\]
for all $p \in Q^{\infty}$. 
To see this,  with a slight abuse of notation, write $(\lambda,\varphi)$ for $(\lambda(p,w),\varphi(p,\bold{w}(p)))$, and note that:
\begin{flalign*}
& (1-\delta)v(a^*,p) + \delta \lambda \left[v(a^*,\varphi)-\frac{m(\varphi)-u(a^*,\varphi)}{m(1)-u(a^*,1)}v(a^*,1)\right]= \\
&(1-\delta)v(a^*,p) + \delta\left[v(a^*,p)- (1-\lambda) v(a^*,1)\right] \\
 &- \frac{m(p)-(1-\lambda)m(1)-u(a^*,p)(1-\delta)}{m(1)-u(a^*,1)}v(a^*,1) + \delta \frac{u(a^*,p) - (1-\lambda)u(a^*,1)}{m(1)-u(a^*,1)}v(a^*,1)=\\
& v(a^*,p)-\frac{m(p)-u(a^*,p)}{m(1)-u(a^*,1)}v(a^*,1), 
\end{flalign*} 
where we use the identities $\lambda \varphi + (1-\lambda)1=p$, $\lambda m(\varphi) + (1-\lambda)m(1)=\bold{w}(p)$, and $\delta \bold{w}(p)= m(p)-(1-\delta)u(a^*,p)$. 

This completes the characterization of $V_{\overline{q}^1}$. Note that $V_{\overline{q}^1}$ and, therefore, all value functions $V_q$, are continuous functions.

\subsubsection{Value functions: another representation}  We now present another construction of $V_q$. For any $q \in [\underline{q}^1,\overline{q}^1]$, define the function $\overline{m}_{q}: [0,1] \rightarrow \mathbb{R}$ as
\begin{eqnarray*}
\begin{cases}
\left(1-\frac{p}{\underline{q}^1}\right)m(0) + \frac{p}{\underline{q}^1}m(\underline{q}^1)  & \text{\;if\;} p \in [0,\underline{q}^1] ,\\
m(p) & \text{\;if\;} p \in (\underline{q}^1,q] ,\\
 \frac{1-p}{1-q}m(q) + \frac{p-q}{1-q}m(1) & \text{\;if\;} p \in (q,1].
\end{cases}
\end{eqnarray*}

 Note that $\overline{m}_{q}$ is convex, $\overline{m}_{q}(p) \geq m(p)$ for all $p \in [0,1]$, $\overline{m}_{q}(0)=m(0)$ and
 $\overline{m}_{q}(1)=m(1)$.  For a graphical illustration, see Figure \ref{fig:m-bar}. 
 
 \begin{figure}[h]
\centering
\begin{tikzpicture}
\draw[->,thick] (0,0)--(0,3.5);
\draw[->,thick] (0,0) -- (6.5,0);
\draw (0,1.5)--(0.5,1)--(2,0.7)--(4,1)--(6,3);
\node[below] at (5,1.5) {$m(\cdot)$};
\draw[dashed, red] (0,1.5)--(6,3) node[above] {$M(\cdot)$};

\node[below] at (1,0) {$\underline{q}^1$};
\draw[dotted] (0.9,0)--(0.9,0.9);
\node[below] at (6,0) {$1$};
\node[left] at (0,3) {$m(1)$};
\draw[dotted] (6,0)--(6,3)--(0,3);

\draw[dotted] (4.4,1.3)--(4.4,0) node[below] {$\overline{q}^{1}$};
\draw[dotted] (3.5,0.95)--(3.5,0) node[below] {$q$};
\draw[thick, cyan](0,1.5)--(0.9,0.93)--(2,0.7)--(3.5,0.93)--(6,3) node[above] at (3.8,1.4) {$\overline{m}_{q}(\cdot)$} ;
\end{tikzpicture}
\caption{The function $\overline{m}_{q}$}\label{fig:m-bar}
\end{figure}
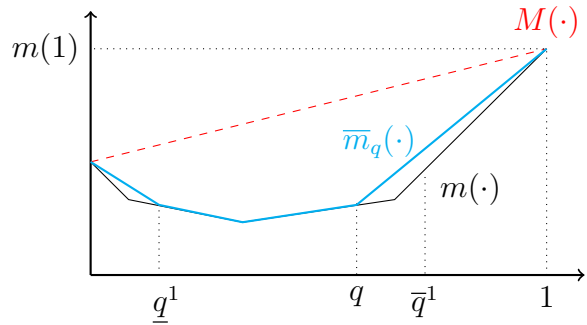

It is straightforward to check that we have the following formula:
\begin{eqnarray}\label{eq:value-function-def}
V_q(p,w) & = & \overline{\lambda}(p,w)V_q(\overline{\varphi}(p,w),\overline{m}_{q}(\overline{\varphi}(p,w)),
\end{eqnarray}
where the functions $\la$ and $\ph$ are defined as in the main text, but with $\overline{m}_{q}$  instead of $m$. See Equation (\ref{eq1}).

\subsection{Theorem \ref{theo2:opti}}\label{app:theo2-opti}
To prove Theorem \ref{theo2:opti}, we prove the following proposition and invoke Theorem \ref{theo1}.

\begin{proposition}\label{V-concave} Let $V_{q^*}$ be the value function induced by the policy $\tau^*$, with 
\[q^*=\sup \left\{p \in Q^1: V_{\overline{q}^1}(p,m(p)) \geq V_{\overline{q}^1}(p,w) \text{\;for all\;} w \right\}.\] Then, 
$V_{q^*}$ is concave in $(p,w)$, decreasing in $w$, and satisfies: 
\[V_{q^*}(p,m(p)) \geq (1-\delta)v(a^*,p) + \delta V_{q^*}(p^*,\bold{w}(p)), \]
for all $p \in Q^1$. 
\end{proposition}

We start with two preliminary observations. \medskip

\textsc{Observation A.} For all $q \in [\underline{q}^1,\overline{q}^1]$, we have the following identity:
\begin{eqnarray*}
V_{q}(p,w) = \frac{1-p}{1-p'}V_{q}\left(p', \frac{1-p'}{1-p}w+\frac{p'-p}{1-p}\overline{m}_{q}(1)\right).
\end{eqnarray*}
The proof is as follows. Let $w'= \frac{1-p'}{1-p}w+\frac{p'-p}{1-p}\overline{m}_{q}(1)$. 

Assume that $w' > \overline{m}_{q}(p')$. Since 
\begin{eqnarray*}
\la(p',w') 
 \begin{pmatrix} \ph(p',w') \\ \overline{m}_{q}(\ph(p',w')) \end{pmatrix} + \left(1-\la(p',w')\right)
 \begin{pmatrix} 1 \\ \overline{m}_{q}(1)\end{pmatrix}&=& \begin{pmatrix} p' \\ w' \end{pmatrix},
\end{eqnarray*}
we have

\begin{eqnarray*}
 \frac{1-p}{1-p'} \la(p',w') 
 \begin{pmatrix} \ph(p',w') \\ \overline{m}_{q}(\ph(p',w')) \end{pmatrix} + \left(1-\frac{1-p}{1-p'} \la(p',w')\right)
 \begin{pmatrix} 1 \\ \overline{m}_{q}(1) \end{pmatrix}&=& \begin{pmatrix} p \\ w \end{pmatrix}.
\end{eqnarray*}
Therefore, $\la(p,w)=  \frac{1-p}{1-p'} \la(p',w') $ and $\ph(p',w')=\ph(p,w)$ since the solution $(\la(p',w'),\ph(p',w'))$ is unique when $w' >m_{q}(p')$. The statement then follows from Equation (\ref{eq:value-function-def}).

Assume that $w' = \overline{m}_{q}(p')$. From the convexity of $\overline{m}_q$, this requires that $w = \overline{m}_q(p)$, so that $\overline{m}_{q}(p')=\frac{1-p'}{1-p}\overline{m}_{q}(p)+\frac{p'-p}{1-p}\overline{m}_{q}(1)$. The result follows from continuity as: 
\begin{eqnarray*}
V_{q}(p,\overline{m}_{q}(p))&=&\lim_{w \rightarrow \overline{m}_{q}(p)}V_{q}(p,w), \\
& = & \lim_{w \rightarrow \overline{m}_{q}(p)}\frac{1-p}{1-p'}V_{q}\left(p', \frac{1-p'}{1-p}w+\frac{p'-p}{1-p}\overline{m}_{q}(1)\right),\\
&= &  \frac{1-p}{1-p'}V_{q}\left(p', \frac{1-p'}{1-p}\overline{m}_{q}(p)+\frac{p'-p}{1-p}\overline{m}_{q}(1)\right),\\
& = & \frac{1-p}{1-p'}V_{q}\left(p', \overline{m}_{q}(p')\right).
\end{eqnarray*}

Note that this implies that 
\begin{eqnarray*}
V_{q}(p,\mathbf{w}(p)+ c)= \la(p,\mathbf{w}(p))V_{q}\left(\ph(p,\mathbf{w}(p)), \overline{m}_{q}(\ph(p,\mathbf{w}(p)))+ \frac{c}{\la(p,\mathbf{w}(p))}\right),
\end{eqnarray*}
where $c$ is a positive constant. 
\medskip

\textsc{Observation B.} The value function $V_{\overline{q}^1}(p,\cdot): [\m(p), M(p)] \rightarrow \mathbb{R}$ is concave in $w$, for each $p$. See Lemma \ref{V-concave-in-w} in section \ref{sec:concave-in-w}.

\medskip

\subsubsection{Proposition \ref{V-concave}(a)} 
We prove that  $V_{q^*}$ is decreasing in $w$. To start with, fix $p \in [0,1]$ and $(w,w') \in [\ms(p),M(p)] \times  [\ms(p),M(p)] $, with $w' >w$. 

First, assume that $p \leq q^*$. If $w = \ms(p)$, then $V_{q^*}(p,w') \leq V_{q^*}(p,w)$ by construction of $q^*$. If $w> \ms(p)$, we have that
\begin{eqnarray*}
\frac{V_{q^*}(p,w') - V_{q^*}(p,w)}{w'-w}&= &\frac{V_{\overline{q}^1}(p,w') - V_{\overline{q}^1}(p,w)}{w'-w} \\
&\leq &\frac{V_{\overline{q}^1}(p,w)-V_{\overline{q}^1}(p,\ms(p))}{w-\ms(p)}\\
&= & \frac{V_{q^*}(p,w)-V_{q^*}(p,\ms(p))}{w-\ms(p)} \leq 0,
\end{eqnarray*}
where the inequality follows from the concavity of $V_{\overline{q}^1}$ with respect to $w$, for all $w \geq m_{\q}(p)$. (Recall that $m_{q^*}(p)=\m(p)$ for all $p \leq q^*$.)
\medskip

Second, assume that $p >q^*$. We show in detail how to make use of Observation A to deduce the result. We repeatedly use similar computations later on. We have

\begin{eqnarray*}
V_{q^*}(p,w') &= &\la(p,w')V_{q^*}(\ph(p,w'),\ms(\ph(p,w'))) \\
& = & \la(p,w') \frac{1-\ph(p,w')}{1-\ph(p,w)} V_{q^*}\left(\ph(p,w),\frac{1-\ph(p,w)}{1-\ph(p,w')}\ms(\ph(p,w'))+ \left(1-\frac{1-\ph(p,w)}{1-\ph(p,w')}\right)\ms(1)\right) \\
& = & \la(p,w) V_{q^*}\left(\ph(p,w),\frac{\la(p,w')}{\la(p,w)}\ms(\ph(p,w'))+ \left(1-\frac{\la(p,w')}{\la(p,w)}\right)\ms(1)\right)\\
& = & \la(p,w) V_{q^*}\left(\ph(p,w),\ms(\ph(p,w))+\frac{w'-w}{\la(p,w)}\right), 
\end{eqnarray*}
where the first line follows from the construction of $V_{q^*}$, the second line from Observation A, the third line from the definition of the functions 
$\la$ and $\ph$ and the last line from the following computations:
\begin{eqnarray*}
\frac{\la(p,w')}{\la(p,w)}\ms(\ph(p,w'))+ \left(1-\frac{\la(p,w')}{\la(p,w)}\right)\ms(1) & = & \frac{1}{\la(p,w)}w' + \left(1-\frac{1}{\la(p,w)}\right)\ms(1)\\
& = &  \frac{1}{\la(p,w)}w' + \left(1-\frac{1}{\la(p,w)}\right)\left[\frac{w-\la(p,w)\ms(\ph(p,w))}{1-\la(p,w)}\right]\\
& = & \ms(\ph(p,w))+\frac{w'-w}{\la(p,w)}.
\end{eqnarray*}
Thus, we are able to express $V_{q^*}(p,w')$ as  $\la(p,w) V_{q^*}(\ph(p,w), \tilde{w})$, with $\tilde{w}$ the above expression. Moreover, 
$\ph(p,w) \leq q^*$ as $w \geq \ms(p)$. We can use the (already established) concavity of $V_{q^*}$ in $w$ for each $p\leq q^*$ to deduce the desired result. More precisely, we have that:
\begin{eqnarray*}
\frac{V_{q^*}(p,w') - V_{q^*}(p,w)}{w'-w}&= &\frac{\la(p,w)\left(V_{q^*}\left(\ph(p,w),\ms(\ph(p,w))+\frac{w'-w}{\la(p,w)}\right)-V_{q^*}\left(\ph(p,w),\ms(\ph(p,w))\right)  \right)}{w'-w} \\
&\leq &0,
\end{eqnarray*}
where the inequality follows from the concavity of $V_{q^*}$ in $w$ at all $p \leq q^*$.
\medskip

Lastly, since $V_{q^*}(p,w)=V_{q^*}(p,\ms(p))$ for all $w \in [m(p),\ms(p)]$, the result immediately follows for all $(w,w')$, with $w \in [m(p),\ms(p)]$.

\subsubsection{Proposition \ref{V-concave}(b)}
We prove the concavity of $V_{q^*}$ with respect to both arguments $(p,w)$. 

Let $\overline{\mathcal{W}}=\{(p,w): w \geq \overline{m}_{q^*}(p) \}$.   Let $(p,w) \in \overline{\mathcal{W}}$, $(p',w') \in \overline{\mathcal{W}}$ and $\alpha \in [0,1]$. Write $(p_{\alpha},w_{\alpha})$ for
\[\alpha\begin{pmatrix} p \\w \end{pmatrix} + (1-\alpha) \begin{pmatrix}p' \\w' \end{pmatrix}.\]  Without loss of generality, assume that $p \leq p'$. 
We have that:
\begin{flalign*}
&\alpha V_{q^*}(p,w) + (1-\alpha)V_{q^*}(p',w') \\  
&= \alpha \frac{1-p}{1-p'}V_{q^*}\Big(p',\underbrace{\frac{1-p'}{1-p} w + \frac{p'-p}{1-p}\ms(1)}_{\geq \ms(p')}\Big)+ (1-\alpha) V_{q^*}(p',w')  \\
& \leq \left(\alpha \frac{1-p}{1-p'}+(1-\alpha)\right)V_{q^*}\left(p',\frac{\alpha \frac{1-p}{1-p'}\left(\frac{1-p'}{1-p} w + \frac{p'-p}{1-p}\ms(1)\right)+(1-\alpha)w'}{\alpha \frac{1-p}{1-p'}+(1-\alpha)}\right)\\
&= \frac{1-p_{\alpha}}{1-p'} V_{q^*}\left(p',\frac{1-p'}{1-p_{\alpha}}w_{\alpha}+ \frac{p'-p_{\alpha}}{1-p_{\alpha}}\ms(1)\right)\\
 &=  V_{q^*}(p_{\alpha},w_{\alpha} ),
\end{flalign*}
where the inequality follows from the concavity of $V_{\overline{q}^1}$ with respect to $w$ for each $p$ and the property that $V_{q^*}(p,w)=V_{\overline{q}^1}(p,w)$ for all $(p,w)$ such that $w \geq \ms(p)$. Notice that we use twice Observation A. 

Finally, for all $(p,w) \in \mathcal{W}$, for all $(p',w') \in \mathcal{W}$ and for all $\alpha$, we have that:
\begin{eqnarray*}
\alpha V_{q^*}(p,w) + (1-\alpha)V_{q^*}(p',w') &= & \alpha V_{q^*}(p,\max(w,\ms(p))) + (1-\alpha)V_{q^*}(p',\max(w',\ms(p'))) \\
& \leq & V_{q^*}(p_{\alpha}, \alpha \max(w,\ms(p)) + (1-\alpha) \max(w,\ms(p')))\\
& \leq & V_{q^*}(p_{\alpha},w_{\alpha}),
\end{eqnarray*}
since $\alpha \max(w,\ms(p)) + (1-\alpha) \max(w,\ms(p')) \geq w_{\alpha}$ and the fact that $V_{q^*}$ is decreasing in $w$ for all $p$. This completes the proof of concavity. 

\subsubsection{Proposition \ref{V-concave} (c)}
We prove that $V_{q^*}(p,m(p)) \geq (1-\delta) v(a^*,p)+ \delta V_{q^*}(p,\bold{w}(p))$ for all $p \in Q^1$.

 The statement is true for all $p \leq q^*$ by definition since $V_{q^*}(p,w) =V_{\overline{q}^1}(p,w)$ for all $w$. 
 
 Assume that $p >q^*$. From Lemma \ref{lem:existence of interval}, there exists $\overline{q}$ such that  $\varphi(p,\mathbf{w}(p)) \geq \varphi(p',\mathbf{w}(p'))$ for all $p' \geq p \geq \overline{q}$. Moreover, it follows from A.6.3 and A.6.4 that $V(p,m(p)) \geq V(p,w)$ for all $w$, for all $p \leq \overline{q}$. Therefore, we must have that $q^* \geq \overline{q}$. It follows that $\varphi(p,\mathbf{w}(p)) <   \varphi(q^*,\mathbf{w}(q^*)) \leq q^*$, hence $\mathbf{w}(p) \geq \overline{m}_{q^*}(p)$. We therefore have that 
$V_{q^*}(p,\bold{w}(p))=V_{\overline{q}^1}(p,\bold{w}(p))$.

Since $V_{\overline{q}^1}(p,m(p))= (1-\delta) v(a^*,p)+ \delta V_{\overline{q}^1}(p,\bold{w}(p))$ for all $p \in Q^1$ and $V_{q^*}(p,m(p))=V_{q^*}(p,\ms(p))=V_{\overline{q}^1}(p,\ms(p))$, it is enough to prove that $V_{\overline{q}^1}(p, \ms(p)) \geq V_{\overline{q}^1}(p,m(p))$.

Clearly, there is nothing prove if $\ms(p)=m(p)$ for all $p \in Q^1$, i.e., if $q^*=\overline{q}^1$ (remember that $\m(p)=m(p)$ for all $p \in Q^1$).

\medskip

So, assume that $\ms(p) > m(p)$ for some $p \in (q^*,\overline{q}^1)$, hence $\ms(p) > m(p)$ for all $p \in (q^*, \overline{q}^1)$. We now argue that if $V_{\overline{q}^1}(p,w) > V_{\overline{q}^1}(p,m(p))$ for some $w \geq \ms(p)$, then 
\[V_{\overline{q}^1}(p',m(p')) < \frac{1-p'}{1-p}V_{\overline{q}^1}(p,w), \]
for all $p'>p$. To see this, observe that $w >m(p)$ and, accordingly, 
\[\frac{1-p'}{1-p}w+ \frac{p'-p}{1-p}m(1) - m(p')>0,\]
since $m$ is convex. Hence, 
\begin{eqnarray*}
0 & < & \frac{V_{\overline{q}^1}(p,w)-V_{\overline{q}^1}(p,m(p))}{w -m(p)} \\
& = & \frac{\frac{1-p}{1-p'}\left[V_{\overline{q}^1}\left(p',\frac{1-p'}{1-p}w+ \frac{p'-p}{1-p}m(1)\right)-V_{\overline{q}^1}\left(p',\frac{1-p'}{1-p}m(p)+ \frac{p'-p}{1-p}m(1)\right)\right]}{w-m(p)}\\
& \leq &  \frac{V_{\overline{q}^1}\left(p',\frac{1-p'}{1-p}w+ \frac{p'-p}{1-p}m(1)\right)- V_{\overline{q}^1}\left(p',m(p')\right)}{\frac{1-p'}{1-p}w+ \frac{p'-p}{1-p}m(1) - m(p')},
\end{eqnarray*}
where the equality follows Observation A and the inequality from the concavity of $V_{\overline{q}^1}$ in $w$ for each $p$. Since 
\[V_{\overline{q}^1}(p,w)=\frac{1-p}{1-p'}V_{\overline{q}^1}\left(p',\frac{1-p'}{1-p}w+ \frac{p'-p}{1-p}m(1)\right), \]
we have the desired result. 

Finally, from the definition of $q^*$,  for all $n >0$, there exist $p_n \in (q^*, \min(q^*+ \frac{1}{n}, \overline{q}^1)]$ and $w_n \geq m(p_n)$ such that $V_{\overline{q}^1}(p_n, m(p_n)) < V_{\overline{q}^1}(p_n,w_n)$. From the concavity of $V_{\overline{q}^1}$ in $w$ for all $p$, $V_{\overline{q}^1}(p_n, m(p_n)) < V_{\overline{q}^1}(p_n,\ms(p_n))$ for all $n$.

From the above argument, for all $p$, for all $n$ sufficiently large, i.e., such that $p_n <p$, we have that
\[V_{\overline{q}^1}(p,m(p)) < \frac{1-p}{1-p_n} V_{\overline{q}^1}(p_n, \ms(p_n)). \]
Taking the limit as $n \rightarrow \infty$, we obtain that 
\[V_{\overline{q}^1}(p,m(p)) < \frac{1-p}{1-q^*}V_{\overline{q}^1}(q^*,\ms(q^*))=V_{\overline{q}^1}(p,\ms(p)),\]
which completes the proof.

\subsection{Concavity of $V_{\overline{q}^1}$ with respect to $w$ for each $p$.}\label{sec:concave-in-w}
\begin{lemma}\label{V-concave-in-w}
The value function $V_{\overline{q}^1}(p,\cdot): [\m(p), M(p)] \rightarrow \mathbb{R}$ is concave in $w$, for each $p$. 
\end{lemma}
This section proves that  $V_{\overline{q}^1}$ is concave in $w$ for each $p$. To do so, we prove that 
\begin{equation*}
\begin{split}
\frac{V_{\overline{q}^1}(p,\m(p)+ \eta(\m(1)-u(a^*,1))) -  V_{\overline{q}^1}(p,\m(p))}{\eta}   \geq \\ 
 \frac{V_{\overline{q}^1}(p,\m(p)+\eta'(\m(1)-u(a^*,1)))-V_{\overline{q}^1}(p,\m(p))}{\eta'}, 
\end{split}
\end{equation*}
for all $(\eta,\eta')$ such that $\eta' \geq \eta$. (See the observations on concave functions.) We start with some preliminary results. 

\subsubsection{Preliminary Results}
We study how the function $\varphi(p,\mathbf{w}(p))$ varies with $p$.

\begin{lemma}\label{lem:existence of interval}
    There exists a non-empty interval $[\underline{q},\overline{q}]$ such that: 
    \begin{enumerate}
        \item For any $p'<p\leq \underline{q}$ or $p'>p\geq \bar{q}$, $\varphi(p,\mathbf{w}(p)) \geq \varphi(p',\mathbf{w}(p'))$,
        \item The ratio $\frac{m(1)-m(\varphi(p,\mathbf{w}(p))}{1-\varphi(p,\mathbf{w}(p))}$ is constant for all $p \in [\underline{q},\overline{q}]$.
    \end{enumerate}
\end{lemma}

\begin{proof}[Proof of Lemma \ref{lem:existence of interval}]
Observe that 
\[\frac{m(1)-\mathbf{w}(p)}{1-p}= \frac{m(1)-m(\varphi(p,\mathbf{w}(p))}{1-\varphi(p,\mathbf{w}(p))}.\] 
Therefore, the convexity of $m$ implies that if  $\frac{m(1)-\mathbf{w}(p)}{1-p} < \frac{m(1)-\mathbf{w}(p')}{1-p'}$, then $\varphi(p,\mathbf{w}(p)) < \varphi(p',\mathbf{w}(p'))$.

Consider the function $h:  [0,1] \rightarrow \mathbb{R}$, defined by  $h(p) = \frac{m(1)-\mathbf{w}(p)}{1-p}$. We argue that $h$ is quasi-concave. For all $(p,p')$ and $\alpha \in [0,1]$, we have that
\begin{eqnarray*}
\frac{m(1)-\mathbf{w}(\alpha p + (1-\alpha)p')}{\alpha(1-p)+ (1-\alpha)(1-p')} & \geq & \frac{\alpha (m(1)- \mathbf{w}(p))  + (1-\alpha) (m(1)- \mathbf{w}(p'))}{\alpha(1-p)+ (1-\alpha)(1-p')} \\
& = & \frac{\alpha (1-p)}{\alpha(1-p)+ (1-\alpha)(1-p')}\frac{m(1)- \mathbf{w}(p)}{1-p} +  \\
& & \indent \frac{(1-\alpha) (1-p')}{\alpha(1-p)+ (1-\alpha)(1-p')}\frac{m(1)- \mathbf{w}(p')}{1-p'}\\
& \geq & \min\left(\frac{m(1)- \mathbf{w}(p)}{1-p} ,\frac{m(1)- \mathbf{w}(p')}{1-p'}  \right),
\end{eqnarray*}
where the first inequality follows form the convexity of $\mathbf{w}$. (Note that the inequality is strict if $\mathbf{w}(\alpha p + (1-\alpha)p') < \alpha \mathbf{w}(p) + (1-\alpha) \mathbf{w}(p')$.)

It follows that if $h(p')  \geq h(p)$, then it is also true for all $p^{''} \in (p,p')$.  Since $h$ is quasi-concave and continuous, the set of maxima is a non-empty convex set $[\underline{q},\overline{q}]$, and the function is increasing for all $p< \underline{q}$ and decreasing for all $p > \overline{q}$. (Note that $m(1) - \mathbf{w}(1)= \frac{(1-\delta)(u(a^*,1)-m(1))}{\delta} <0$, hence the function is equal to $-\infty$ at $p=1$.)\end{proof}

We can make few additional observations about the interval $[\underline{q},\overline{q}]$.  Let $k^*:=\sup\{k: Q^k \neq \emptyset\}$. Since $\varphi(\overline{q}^k,\mathbf{w}(\overline{q}^k))= \underline{q}^k$, the function $h$ is decreasing for all $p \geq \overline{q}^{k^*}$. 
Similarly, since $\varphi(\underline{q}^k,\mathbf{w}(\underline{q}^k))= \underline{q}^{k-1}$, the function $h$ is increasing for all $p \leq \underline{q}^{k^*}$.  Therefore, $[\underline{q},\overline{q}] \subset Q^{k^*}$. 

If $P \neq \emptyset$, so that $k^*=\infty$, then for all $p \in P$, the function $h$ is increasing by convexity of $m$ since $\mathbf{w}(p)=m(p)$. (This is clearly true since $\varphi(p,m(p))= p$ in that region.) Therefore, $\overline{p} \leq \underline{q}$ if $P \neq \emptyset$. 

Finally, let $\tilde{p}:=\inf\{p: m(p)=u(a^1,p)\}$. By construction, $m$ is linear from $\tilde{p}$ to 1, i.e., $[\tilde{p},1]$ is the utmost right linear piece of $m$. We have that $\overline{q}< \tilde{p}$. To see this, observe that for all $p \geq \tilde{p}$, 
\[ \frac{m(1)-\mathbf{w}(p)}{1-p}= \frac{(1-\delta)(\overbrace{u(a^*,1)-u(a^1,1)}^{ < 0})}{1-p}+ \frac{(u(a^1,0)-u(a^1,1)) - (1-\delta)(u(a^*,0)-u(a^*,1))}{\delta},   \]
hence it is decreasing in $p$. (If there are multiple optimal actions at $p=1$, the argument applies to all of them and, therefore, to the one that induces the smallest $\tilde{p}$.) 
\medskip 

The second preliminary result is technical.  For any $p\in (0,1)$ and any $\eta\in \left[0,\frac{M(p)-\m(p)}{\m(1)-u(a^*,1)}\right]$, define  $w(p;\eta)$ as 
\[\m(p) + \eta \left[ {\m(1) - u({a^*},1)} \right],\]
and write $(\lambda_\eta,\varphi_\eta)$ for  $(\la (p,w(p;\eta )), \ph (p,w(p;\eta )))$. To ease notation, we do not explicitly write the dependence of $(\lambda_\eta,\varphi_\eta)$ on $p$. We have the following:

\begin{lemma}
    \label{lem:decreasing in eta}
    $\varphi_{\eta}$, $\lambda_{\eta}$, and $\frac{{1 - \lambda_\eta }}{\eta }$ are all decreasing in $\eta$.  
\end{lemma}
The proof follows directly from the definition of $(\lambda_\eta,\varphi_{\eta})$ and is omitted. 
\medskip 

Finally, we conclude with the following implication of Observation A, which wel use throughout. For all $(p,w,w')$ with $w\leq w'$, 
we have that: 
\begin{eqnarray*}
V_{\overline{q}^1}(p,w)-V_{\overline{q}^1}(p,w')= \la(p,w)\left[V_{\overline{q}^1}(\ph(p,w),\m(p,w))- V_{\overline{q}^1}\left(\ph(p,w),\m(p,w)+ \frac{w'-w}{\la(p,w)}\right)\right].
\end{eqnarray*} 

\subsubsection{Proof of Lemma \ref{V-concave-in-w}}We now prove that the gradient $\mathcal{G}(p;\eta):=\frac{V_{\overline{q}^1}(p,\m(p)) - V_{\overline{q}^1}(p,w(p;\eta))}{\eta}$ is increasing in $\eta\in \left[0,\frac{M(p)-\m(p)}{\m(1)-u(a^*,1)}\right]$, for all $p$. We prove it on three separate intervals $\mathcal{I}_1$, $\mathcal{I}_2$ and $\mathcal{I}_3$. If $P= \emptyset$, the three intervals are $[0,\underline{q}]$, $(\underline{q},\overline{q}]$ and $(\overline{q},1]$, respectively. If $P\neq \emptyset$, the three intervals are $[0,\underline{p}]$, $(\underline{p},\overline{q}^{\infty}]$ and $(\overline{q}^\infty,1]$, respectively.

\subsubsection{For all $p \in \mathcal{I}_1$, $\mathcal{G}(p;\eta)$ is increasing in $\eta$.}

We limit attention to the case $P \neq \emptyset$. (The case $P =\emptyset$ is identical.) The proof is by induction. First, consider the interval $[0,\underline{q}^1]$. Remember that at $\underline{q}^1$, we have a closed-form solution for $V_{\overline{q}^1}(\underline{q}^1,w)$ for all $w$ given by 
\[V_{\overline{q}^1}(\underline{q}^1,w)=\frac{M(\q)-w}{M(\q)-u(a^*,\q)}v(a^*,\q). \]
Therefore, 
     \[\begin{gathered}
  \frac{{V_{\overline{q}^1}({{\underline{q} }^1},\m({{\underline{q} }^1})) - V_{\overline{q}^1}({{\underline{q} }^1},{w(\underline{q}^1;\eta)})}}{\eta } = \frac{1}{\eta }\left[ {\frac{{M({{\underline{q} }^1}) - \m({{\underline{q} }^1})}}{{M({{\underline{q} }^1}) - u({a^*},{{\underline{q} }^1})}}v({a^*},{{\underline{q} }^1}) - \frac{{M({{\underline{q} }^1}) - {w(\underline{q}^1;\eta)}}}{{M({{\underline{q} }^1}) - u({a^*},{{\underline{q} }^1})}}v({a^*},{{\underline{q} }^1})} \right] \hfill \\
   = \frac{{v({a^*},{{\underline{q} }^1})}}{{M({{\underline{q} }^1}) - u({a^*},{{\underline{q} }^1})}}\frac{{[\m({{\underline{q} }^1}) + \eta (\m(1) - u({a^*},1))] -\m({{\underline{q} }^1})}}{\eta } \hfill \\
   = \frac{{{{\underline{q} }^1}v({a^*},1) + (1 - {{\underline{q} }^1})v({a^*},0)}}{{{{\underline{q} }^1}[\m(1) - u({a^*},1)] + (1 - {{\underline{q} }^1})[\m(0) - u({a^*},0)]}}\frac{{{w(\underline{q}^1;\eta)} - \m({{\underline{q} }^1})}}{\eta } \hfill \\
   = v({a^*},1)\underbrace{\frac{{{{\underline{q} }^1} + (1 - {{\underline{q} }^1})\frac{{v({a^*},0)}}{{v({a^*},1)}}}}{{{{\underline{q} }^1} + (1 - {{\underline{q} }^1})\frac{{\m(0) - u({a^*},0)}}{{\m(1) - u({a^*},1)}}}}}_{\geq 1 \text{\,since\,} \frac{v({a^*},0)}{v({a^*},1)} \geq \frac{\m(0)-u(a^*,0)}{\m(1)-u(a^*,1)}} \geq v(a^*,1). \hfill \\ 
\end{gathered} \]

We now consider any $p \in[0,\q)$. From Observation A, we have that: 
\[\left\{ \begin{gathered}
  V_{\overline{q}^1}(p,\m(p)) = \frac{{1 - p}}{{1 - {{\underline{q} }^1}}}V_{\overline{q}^1}\left( {{{\underline{q} }^1},\frac{{1 - {{\underline{q} }^1}}}{{1 - p}}\m(p) + \left( {1 - \frac{{1 - {{\underline{q} }^1}}}{{1 - p}}} \right)\m(1)} \right) \hfill \\
  V_{\overline{q}^1}(p,w(p;\eta )) = \frac{{1 - p}}{{1 - {{\underline{q} }^1}}}V_{\overline{q}^1}\left( {{{\underline{q} }^1},\frac{{1 - {{\underline{q} }^1}}}{{1 - p}}\m(p) + \left( {1 - \frac{{1 - {{\underline{q} }^1}}}{{1 - p}}} \right)\m(1) + \frac{{1 - {{\underline{q} }^1}}}{{1 - p}}\eta \left[ {\m(1) - u({a^*},1)} \right]} \right) \hfill \\ 
\end{gathered}  \right.\]
It follows that
\begin{scriptsize}
    \[\begin{gathered}
  \frac{{V_{\overline{q}^1}(p,\m(p)) - V_{\overline{q}^1}(p,w(p;\eta ))}}{\eta } \hfill \\
   = \frac{{1 - p}}{{1 - {{\underline{q} }^1}}}\frac{{V_{\overline{q}^1}\left( {{{\underline{q} }^1},\frac{{1 - {{\underline{q} }^1}}}{{1 - p}}\m(p) + \left( {1 - \frac{{1 - {{\underline{q} }^1}}}{{1 - p}}} \right)\m(1)} \right) - V_{\overline{q}^1}\left( {{{\underline{q} }^1},\frac{{1 - {{\underline{q} }^1}}}{{1 - p}}\m(p) + \left( {1 - \frac{{1 - {{\underline{q} }^1}}}{{1 - p}}} \right)\m(1) + \frac{{1 - {{\underline{q} }^1}}}{{1 - p}}\eta \left[ {\m(1) - u({a^*},1)} \right]} \right)}}{\eta } \hfill \\
   = \frac{1-p}{1-\underline{q}^1}\frac{1-\underline{q}^1}{1-p}\frac{\m(1)-u(a^*,1)}{M(\underline{q}^1)-u(a^*,\underline{q}^1}v(a^*,\underline{q}^1) =\frac{1-p}{1-\underline{q}^1}\frac{1-\underline{q}^1}{1-p} v({a^*},1)\frac{{{{\underline{q} }^1} + (1 - {{\underline{q} }^1})\frac{{v({a^*},0)}}{{v({a^*},1)}}}}{{{{\underline{q} }^1} + (1 - {{\underline{q} }^1})\frac{{\m(0) - u({a^*},0)}}{{\m(1) - u({a^*},1)}}}}\hfill\\
   \geqslant \frac{{1 - p}}{{1 - {{\underline{q} }^1}}}\frac{{1 - {{\underline{q} }^1}}}{{1 - p}}v({a^*},1) = v({a^*},1). \hfill \\ 
\end{gathered} \]
\end{scriptsize}
Therefore, $\mathcal{G}(p;\eta) \geq v(a^*,1)$ for all $\eta$, for all $p \in [0,\q]$. Moreover, the gradient $\mathcal{G}(p;\eta) $ is independent of $\eta$ for all $p \in [0,\q]$, hence is (weakly) increasing. 

\medskip 

By induction, assume that $\mathcal{G}(p;\eta) \geq v(a^*,1)$ for all $p \in [0,\underline{q}^k]$ and is increasing in $\eta$, we want to prove that both properties also hold for all $p \in (\underline{q}^k,\underline{q}^{k+1}]$.\medskip

We rewrite $V_{\overline{q}^1}(p,w(p;\eta))$ as follows: 
\[\begin{gathered}
  V_{\overline{q}^1}(p,w(p;\eta) ) = {\lambda _\eta }V_{\overline{q}^1}({\varphi _\eta },\m({\varphi _\eta })) = {\lambda _\eta }\left[ {(1 - \delta )v({a^*},{\varphi _\eta }) + \delta V_{\overline{q}^1}({\varphi _\eta },\mathbf{w}({\varphi _\eta }))} \right] \hfill \\
   = (1 - \delta ){\lambda _\eta }v({a^*},{\varphi _\eta }) + \delta {\lambda _\eta }V_{\overline{q}^1}({\varphi _\eta },\mathbf{w}({\varphi _\eta })) \hfill \\
   = (1 - \delta ){\lambda _\eta }v({a^*},{\varphi _\eta }) + \delta V_{\overline{q}^1}\left( {p,{\lambda _\eta }\mathbf{w}({\varphi _\eta }) + [1 - {\lambda _\eta }]\m(1)} \right) \hfill \\
   = (1 - \delta ){\lambda _\eta }v({a^*},{\varphi _\eta }) + \delta V_{\overline{q}^1}\left( {p,\mathbf{w}(p) + \frac{{\eta  - (1 - \delta )(1 - {\lambda _\eta })}}{\delta }[\m(1) - u({a^*},1)]} \right). \hfill \\ 
\end{gathered} \]
The second to last equality follows from Observation A, while the last equality follows from:
\begin{small}
    \[\begin{gathered}
  {\lambda _\eta }\mathbf{w}({\varphi _\eta }) + [1 - {\lambda _\eta }]\m(1) = {\lambda _\eta }\frac{{ - (1 - \delta )u({a^*},{\varphi _\eta }) + \m({\varphi _\eta })}}{\delta } + [1 - {\lambda _\eta }]\m(1) \hfill \\
   = \frac{{ - (1 - \delta )}}{\delta }{\lambda _\eta }u({a^*},{\varphi _\eta }) + \frac{1}{\delta }{\lambda _\eta }\m({\varphi _\eta }) + [1 - {\lambda _\eta }]\m(1) \hfill \\
   = \frac{{ - (1 - \delta )}}{\delta }\left[ {u({a^*},p) - (1 - {\lambda _\eta })u({a^*},1)} \right] + \frac{1}{\delta }\left[ {{w(p;\eta) } - (1 - {\lambda _\eta })\m(1)} \right] + [1 - {\lambda _\eta }]\m(1) \hfill \\
   = \frac{{ - (1 - \delta )}}{\delta }\left[ {u({a^*},p) - (1 - {\lambda _\eta })u({a^*},1)} \right] + \frac{1}{\delta }\left[ {\m(p) + \eta (\m(1) - u({a^*},1)) - (1 - {\lambda _\eta })\m(1)} \right] + [1 - {\lambda _\eta }]\m(1) \hfill \\
   = \left[ {\frac{{ - (1 - \delta )}}{\delta }u({a^*},p) + \frac{1}{\delta }\m(p)} \right] + \frac{{\eta  - (1 - \delta )(1 - {\lambda _\eta })}}{\delta }[\m(1) - u({a^*},1)]. \hfill \\ 
\end{gathered} \]
\end{small}

For future reference, recall that 
\begin{eqnarray*}
\lambda_{\eta} \mathbf{w}(\varphi _{\eta }) + (1 - \lambda _{\eta})\m(1) & = & \lambda_{\eta}\left[\la(\varphi_{\eta},\mathbf{w}(\varphi_{\eta})) \m(\ph(\varphi_{\eta},\mathbf{w}(\varphi_{\eta})))+  (1- \la(\varphi_{\eta},\mathbf{w}(\varphi_{\eta}))\m(1)\right] \\
& & + (1 - \lambda _{\eta })\m(1),
\end{eqnarray*}
so that 
\begin{eqnarray*}
\ph\left(p, \mathbf{w}(p)+\frac{\eta  - (1 - \delta )(1 - \lambda _{\eta})}{\delta }[\m(1) - u({a^*},1)]\right)= 
\ph(\varphi_{\eta},\mathbf{w}(\varphi_{\eta})),
\end{eqnarray*} 
and 
\begin{eqnarray*}
\la\left(p, \mathbf{w}(p)+\frac{\eta  - (1 - \delta )(1 - \lambda _{\eta})}{\delta }[\m(1) - u({a^*},1)]\right)= 
\lambda_{\eta}\la(\varphi_{\eta},\mathbf{w}(\varphi_{\eta})). 
\end{eqnarray*} 
Since $\varphi_\eta$ is decreasing in $\eta$, we have that $\varphi_{\eta'} \leq \varphi_{\eta}$ when $\eta'>\eta$ and, therefore, we have that $\ph(\varphi_{\eta},\mathbf{w}(\varphi_{\eta})) \leq \ph(\varphi_{\eta'},\mathbf{w}(\varphi_{\eta'}))$ since $\varphi_{\eta'} \leq \varphi_{\eta} \leq p \leq \underline{q}$. Similarly, since $\varphi_{\eta} < p \leq \underline{q}$, we have that 
$\ph(\varphi_{\eta},\mathbf{w}(\varphi_{\eta})) \leq \ph(p,\mathbf{w}(p))$ and, therefore, $\frac{\eta  - (1 - \delta )(1 - \lambda _{\eta})}{\delta }>0$.
\medskip

We now return to the computation of the gradient. We have:
\begin{equation}\label{eq:derivative}
    \begin{scriptsize}
    \begin{gathered}
   = \frac{{\left[ {(1 - \delta )v({a^*},p) + \delta V_{\overline{q}^1}(p,\mathbf{w}(p))} \right] - \left[ {(1 - \delta ){\lambda _\eta }v({a^*},{\varphi _\eta }) + \delta V_{\overline{q}^1}\left( {p,\mathbf{w}(p) + \frac{{\eta  - (1 - \delta )(1 - {\lambda _\eta })}}{\delta }[m(1) - u({a^*},1)]} \right)} \right]}}{\eta } \hfill \\
   = \frac{{(1 - \delta )}}{\eta }\left[ {v({a^*},p) - {\lambda _\eta }v({a^*},{\varphi _\eta })} \right] + \frac{\delta }{\eta }\left[ {V_{\overline{q}^1}(p,\mathbf{w}(p)) - V_{\overline{q}^1}\left( {p,\mathbf{w}(p) + \frac{{\eta  - (1 - \delta )(1 - {\lambda _\eta })}}{\delta }[m(1) - u({a^*},1)]} \right)} \right] \hfill \\
   = \frac{{(1 - \delta )}}{\eta }(1 - {\lambda _\eta })v({a^*},1) + \frac{\delta }{\eta }\left[ {V_{\overline{q}^1}(p,\mathbf{w}(p)) - V_{\overline{q}^1}\left( {p,\mathbf{w}(p) + \frac{{\eta  - (1 - \delta )(1 - {\lambda _\eta })}}{\delta }[m(1) - u({a^*},1)]} \right)} \right]. \hfill \\ 
\end{gathered} 
\end{scriptsize}
\end{equation}

We further develop the above expression. To ease notation, we write $(\varphi(p),\lambda(p))$ for $(\ph(p,\mathbf{w}(p)),\la(p,\mathbf{w}(p)))$. Note that $\varphi(p) \in (\underline{q}^{k-1}, \underline{q}^k]$, since $p \in (\underline{q}^k,\underline{q}^{k+1}]$. As ${\frac{{\eta  - (1 - \delta )(1 - {\lambda _\eta })}}{\delta }}>0$,
we have that:
\begin{scriptsize}
    \[\begin{gathered}
   = \frac{{(1 - \delta )}}{\eta }(1 - {\lambda _\eta })v({a^*},1) + \frac{\delta }{\eta }\left[ {V_{\overline{q}^1}(p,\mathbf{w}(p)) - V_{\overline{q}^1}\left( {p,\mathbf{w}(p) + \frac{{\eta  - (1 - \delta )(1 - {\lambda _\eta })}}{\delta }[m(1) - u({a^*},1)]} \right)} \right] \hfill \\
   = \frac{{(1 - \delta )}}{\eta }(1 - {\lambda _\eta })v({a^*},1) + \frac{\delta }{\eta }\frac{{\eta  - (1 - \delta )(1 - {\lambda _\eta })}}{\delta }\frac{{V_{\overline{q}^1}(p,\mathbf{w}(p)) - V_{\overline{q}^1}\left( {p,\mathbf{w}(p) + \frac{{\eta  - (1 - \delta )(1 - {\lambda _\eta })}}{\delta }[m(1) - u({a^*},1)]} \right)}}{{\frac{{\eta  - (1 - \delta )(1 - {\lambda _\eta })}}{\delta }}} \hfill \\
   = \frac{{(1 - \delta )}}{\eta }(1 - {\lambda _\eta })v({a^*},1) + \left[ {1 - \frac{{\left( {1 - \delta } \right)(1 - {\lambda _\eta })}}{\eta }} \right]\frac{{\lambda (p)\left[ {V_{\overline{q}^1}(\varphi (p),\m(\varphi (p))) - V_{\overline{q}^1}\left( {\varphi (p),\m(\varphi (p)) + \frac{{\eta  - (1 - \delta )(1 - {\lambda _\eta })}}{{\delta \lambda (p)}}[m(1) - u({a^*},1)]} \right)} \right]}}{{\frac{{\eta  - (1 - \delta )(1 - {\lambda _\eta })}}{\delta }}} \hfill \\
   = \frac{{(1 - \delta )}}{\eta }(1 - {\lambda _\eta })v({a^*},1) + \left[ {1 - \frac{{\left( {1 - \delta } \right)(1 - {\lambda _\eta })}}{\eta }} \right]\frac{{V_{\overline{q}^1}(\varphi (p),\m(\varphi (p))) - V_{\overline{q}^1}\left( {\varphi (p),\m(\varphi (p)) + \frac{{\eta  - (1 - \delta )(1 - {\lambda _\eta })}}{{\delta \lambda (p)}}[m(1) - u({a^*},1)]} \right)}}{{\frac{{\eta  - (1 - \delta )(1 - {\lambda _\eta })}}{{\delta \lambda (p)}}}} \hfill \\
   \geqslant \frac{{(1 - \delta )}}{\eta }(1 - {\lambda _\eta })v({a^*},1) + \left[ {1 - \frac{{\left( {1 - \delta } \right)(1 - {\lambda _\eta })}}{\eta }} \right]v({a^*},1)
   = v({a^*},1), \hfill \\ 
\end{gathered} \]
\end{scriptsize}
where we use Observation A and the induction step. 
\medskip 

We now show that the gradient is increasing in $\eta$. To start with, note that $\frac{{\eta  - (1 - \delta )(1 - {\lambda _\eta })}}{\delta }$ is increasing in $\eta$ since  $\frac{{1 - {\lambda _\eta }}}{\eta }$ is decreasing in $\eta$ (see Lemma \ref{lem:decreasing in eta}). 
For any $\eta>\eta^\prime$, we have the following 
\[\begin{gathered}
  \frac{{V_{\overline{q}^1}(p,\mathbf{w}(p)) - V_{\overline{q}^1}\left( {p,\mathbf{w}(p) + \frac{{\eta  - (1 - \delta )(1 - {\lambda _\eta })}}{\delta }\left[ {\m(1) - u({a^*},1)} \right]} \right)}}{{\frac{{\eta  - (1 - \delta )(1 - {\lambda _\eta })}}{\delta }}} \hfill \\
   = \frac{{\lambda (p)V_{\overline{q}^1}(\varphi (p),\m(\varphi (p))) - \lambda (p)V_{\overline{q}^1}\left( {\varphi (p),\m(\varphi (p)) + \frac{{\eta  - (1 - \delta )(1 - {\lambda _\eta })}}{{\delta \lambda (p)}}\left[ {\m(1) - u({a^*},1)} \right]} \right)}}{{\frac{{\eta  - (1 - \delta )(1 - \lambda )}}{\delta }}} \hfill \\
   = \frac{{V_{\overline{q}^1}(\varphi (p),\m(\varphi (p))) - V_{\overline{q}^1}\left( {\varphi (p),\m(\varphi (p)) + \frac{{\eta  - (1 - \delta )(1 - {\lambda _\eta })}}{{\delta \lambda (p)}}\left[ {\m(1) - u({a^*},1)} \right]} \right)}}{{\frac{{\eta  - (1 - \delta )(1 - \lambda )}}{{\delta \lambda (p)}}}} \hfill \\
   \geqslant \frac{{V_{\overline{q}^1}(\varphi (p),\m(\varphi (p))) - V_{\overline{q}^1}\left( {\varphi (p),\m(\varphi (p)) + \frac{{{\eta ^\prime } - (1 - \delta )(1 - {\lambda _{{\eta ^\prime }}})}}{{\delta \lambda (p)}}\left[ {\m(1) - u({a^*},1)} \right]} \right)}}{{\frac{{\eta^\prime  - (1 - \delta )(1 - \lambda_{\eta^\prime} )}}{{\delta \lambda (p)}}}} \hfill \\ 
   =\frac{{V_{\overline{q}^1}(p,\mathbf{w}(p)) - V_{\overline{q}^1}\left( {p,\mathbf{w}(p) + \frac{{{\eta ^\prime } - (1 - \delta )(1 - {\lambda _{{\eta ^\prime }}})}}{\delta }\left[ {\m(1) - u({a^*},1)} \right]} \right)}}{{\frac{{\eta^\prime  - (1 - \delta )(1 - {\lambda _{{\eta ^\prime }}})}}{\delta }}}, \hfill 
\end{gathered} \]
where the  inequality follows from the fact that $\varphi(p)\in (\underline{q}^{k-1},\underline{q}^k]$ and, therefore, the gradient $\mathcal{G}(\varphi(p);\eta)$ being increasing in $\eta$ by the induction hypothesis. \medskip 

Finally, we have that
\begin{scriptsize}
    \[\begin{gathered}
  \frac{1}{\eta }\left[ {V_{\overline{q}^1}(p,\m(p)) - V_{\overline{q}^1}(p,{w(p;\eta) })} \right] = \hfill \\
  \frac{{\left( {1 - \delta } \right)(1 - {\lambda _\eta })}}{\eta }v({a^*},1) + \left[ {1 - \frac{{\left( {1 - \delta } \right)(1 - {\lambda _\eta })}}{\eta }} \right]\frac{{V_{\overline{q}^1}(p,\mathbf{w}(p)) - V_{\overline{q}^1}\left( {p,\mathbf{w}(p) + \frac{{\eta  - (1 - \delta )(1 - {\lambda _\eta })}}{\delta }\left[ {m(1) - u({a^*},1)} \right]} \right)}}{{\frac{{\eta  - (1 - \delta )(1 - {\lambda _\eta })}}{\delta }}} \hfill \\
   \geqslant \frac{{\left( {1 - \delta } \right)(1 - {\lambda _\eta })}}{\eta }v({a^*},1) + \left[ {1 - \frac{{\left( {1 - \delta } \right)(1 - {\lambda _\eta })}}{\eta }} \right]\frac{{V_{\overline{q}^1}(p,\mathbf{w}(p)) - V_{\overline{q}^1}\left( {p,\mathbf{w}(p) + \frac{{{\eta ^\prime } - (1 - \delta )(1 - {\lambda _{{\eta ^\prime }}})}}{\delta }\left[ {\m(1) - u({a^*},1)} \right]} \right)}}{{\frac{{{\eta ^\prime } - (1 - \delta )(1 - {\lambda _{{\eta ^\prime }}})}}{\delta }}} \hfill \\
   = \frac{{\left( {1 - \delta } \right)(1 - {\lambda _{{\eta ^\prime }}})}}{{{\eta ^\prime }}}v({a^*},1) + \left[ {1 - \frac{{\left( {1 - \delta } \right)(1 - {\lambda _{{\eta ^\prime }}})}}{{{\eta ^\prime }}}} \right]\frac{{V_{\overline{q}^1}(p,\mathbf{w}(p)) - V_{\overline{q}^1}\left( {p,\mathbf{w}(p) + \frac{{{\eta ^\prime } - (1 - \delta )(1 - {\lambda _{{\eta ^\prime }}})}}{\delta }\left[ {\m(1) - u({a^*},1)} \right]} \right)}}{{\frac{{{\eta ^\prime } - (1 - \delta )(1 - {\lambda _{{\eta ^\prime }}})}}{\delta }}} \hfill \\
   + \left[ {\frac{{\left( {1 - \delta } \right)(1 - {\lambda _{{\eta ^\prime }}})}}{{{\eta ^\prime }}} - \frac{{\left( {1 - \delta } \right)(1 - {\lambda _\eta })}}{\eta }} \right]\left[ {\frac{{V_{\overline{q}^1}(p,\mathbf{w}(p)) - V_{\overline{q}^1}\left( {p,\mathbf{w}(p) + \frac{{{\eta ^\prime } - (1 - \delta )(1 - {\lambda _{{\eta ^\prime }}})}}{\delta }\left[ {\m(1) - u({a^*},1)} \right]} \right)}}{{\frac{{{\eta ^\prime } - (1 - \delta )(1 - {\lambda _{{\eta ^\prime }}})}}{\delta }}} - v({a^*},1)} \right] \hfill \\
   \geqslant \frac{1}{{{\eta ^\prime }}}\left[ {V_{\overline{q}^1}(p,\m(p)) - V_{\overline{q}^1}(p,w(p;\eta^\prime))} \right] \hfill \\
   + \left[ {\frac{{\left( {1 - \delta } \right)(1 - {\lambda _{{\eta ^\prime }}})}}{{{\eta ^\prime }}} - \frac{{\left( {1 - \delta } \right)(1 - {\lambda _\eta })}}{\eta }} \right] \hfill \\
   \times \left[ {\frac{{V_{\overline{q}^1}(\varphi (p),\m(\varphi (p))) - V_{\overline{q}^1}\left( {\varphi (p),\m(\varphi (p)) + \frac{{{\eta ^\prime } - (1 - \delta )(1 - {\lambda _{{\eta ^\prime }}})}}{{\delta \lambda (p)}}\left[ {\m(1) - u({a^*},1)} \right]} \right)}}{{\frac{{{\eta ^\prime } - (1 - \delta )(1 - {\lambda _{{\eta ^\prime }}})}}{{\delta \lambda (p)}}}} - v({a^*},1)} \right]  \\
   \geqslant \frac{1}{{{\eta ^\prime }}}\left[ {V_{\overline{q}^1}(p,\m(p)) - V_{\overline{q}^1}(p,w(p;\eta^\prime))} \right]. \hfill \\ 
\end{gathered} \]
\end{scriptsize}
The last inequality follows from the fact that the gradient in the second bracket is weakly larger than $v(a^*,1)$ by the induction hypothesis and the fact that $\frac{1-\lambda_{\eta}}{\eta} < \frac{1-\lambda_{\eta^\prime}}{\eta^\prime}$ (Lemma \ref{lem:decreasing in eta}).

Since $\lim_{k \rightarrow \infty} \underline{q}^{k}=\underline{p}$ when $P \neq \emptyset$, this completes the proof that the gradient is greater than $v(a^*,1)$ for all $p \in [0,\underline{p}]$.

\subsubsection{For all $p \in \mathcal{I}_2$, $\mathcal{G}(p;\eta)$ is increasing in $\eta$.}

We first treat the case $P \neq \emptyset$. Recall that for all $p \in (\underline{p},\overline{q}^{\infty}]$, we have an explicit definition of the value function $V_{\overline{q}^1}(p,\m(p))$ as:
\[v(a^*,p)-\frac{\m(p)-u(a^*,p)}{\m(1)-u(a^*,1)}v(a^*,1).\]
Define $\bar{\eta}(p)$ as the solution to ${\varphi _{{\bar{\eta} }(p)}} = \ph (p,w(p;{\bar{\eta}(p))) = \underline{p}}$. Note that for any $p \in  (\underline{p},\overline{q}^{\infty}]$, for any $\eta \leq \bar{\eta}$, $\varphi_{\eta} \in [\underline{p},\overline{q}^{\infty}]$. Therefore, 
\begin{eqnarray*}
V_{\overline{q}^1}(p,w(p;\eta))= \lambda_{\eta} V_{\overline{q}^1}(\varphi_{\eta},\m(\varphi_{\eta}))&=& \lambda_{\eta}
\left[v(a^*,\varphi_{\eta})-\frac{\m(\varphi_{\eta})-u(a^*,\varphi_{\eta})}{\m(1)-u(a^*,1)}v(a^*,1)\right]\\
&=& v(a^*,p)-\frac{w(p;\eta)-u(a^*,p)}{\m(1)-u(a^*,1)}v(a^*,1).
\end{eqnarray*}
It follows that the gradient is equal to $v(a^*,1)$ for all $p\in (\underline{p},p^*]$, for all $\eta \leq \bar{\eta}$. 

Consider now $\eta > \bar{\eta}$. We rewrite the gradient $\mathcal{G}(p;\eta)$ as follows:
    \[\begin{gathered}
  \frac{{V_{\overline{q}^1}(p,\m(p)) - V_{\overline{q}^1}(p,{w(p;\eta) })}}{\eta } \\
  = \frac{{V_{\overline{q}^1}(p,\m(p)) - V_{\overline{q}^1}(p,w(p;\eta_1(p)))}}{\eta } +
   \frac{{V_{\overline{q}^1}(p,w(p;\eta_1(p))) - V_{\overline{q}^1}(p,w(p;\eta) )}}{\eta } \hfill\\
   = \frac{{{\eta _1}(p)}}{\eta }\frac{{V_{\overline{q}^1}(p,\m(p)) - V_{\overline{q}^1}(p,w(p,{\eta _1}(p)))}}{{{\eta _1}(p)}} + \frac{{\eta  - {\eta _1}(p)}}{\eta }\frac{{V_{\overline{q}^1}(p,w(p;\eta_1(p))) - V_{\overline{q}^1}(p,{w(p;\eta) })}}{{\eta  - {\eta _1}(p)}} \hfill \\
 = \frac{{{\eta _1}(p)}}{\eta }v({a^*},1) + \frac{{\eta  - {\eta _1}(p)}}{\eta }\frac{{\frac{{1 - p}}{{1 - \underline{p} }}\left[ {V_{\overline{q}^1}(\underline{p} ,\m(\underline{p} )) - V_{\overline{q}^1}\left( {\underline{p} ,{w}\left(\underline{p};{\frac{{\eta  - {\eta _1}(p)}}{{\frac{{1 - p}}{{1 - \underline{p} }}}}} \right)} \right)} \right]}}{{\eta  - {\eta _1}(p)}}\hfill \\
   = \frac{{{\eta _1}(p)}}{\eta }v({a^*},1) + \frac{{\eta  - {\eta _1}(p)}}{\eta }\mathcal{G}\left(\underline{p};\frac{\eta-\eta_1(p)}{\frac{1-p}{1-\underline{p}}}\right).\hfill \\ 
   \end{gathered} \]
Since we have already shown that $\mathcal{G}(\underline{p};\eta)$ is increasing in $\eta$ and weakly larger than $v(a^*,1)$, we have that the gradient $\mathcal{G}(p;\eta)$ is also weakly increasing in $\eta$ (and greater than $v(a^*,1)$).\medskip

We now treat the case $P =\emptyset$. Define $\bar{\eta}(p)$ as the solution to ${\varphi _{{\bar{\eta} }(p)}} = \ph (p,w(p;{\bar{\eta}(p))) = \underline{q}}$. Note that for any $p \in  [\underline{q},\overline{q}]$, for any $\eta \leq \bar{\eta}$, $\varphi_{\eta} \in [\underline{q},\overline{q}]$. Therefore, for all $\eta \leq \bar{\eta}$, $\eta  = (1 - \delta )(1 - {\lambda _\eta })$ since the ratio $\frac{\m(1)-\mathbf{w}(\varphi_{\eta})}{1-\varphi_{\eta}}$ is constant in $\eta$ and so is $\ph(\varphi_{\eta},\mathbf{w}(\varphi_{\eta}))$. (Recall that we vary $\eta$ at a fixed $p$.) It follows then from Equation \eqref{eq:derivative} that 
    \begin{small}
        \[\begin{gathered}
  \mathcal{G}(p;\eta ) = \frac{{(1 - \delta )}}{\eta }(1 - {\lambda _\eta })v({a^*},1) + \frac{\delta }{\eta }\left[ {V_{\overline{q}^1}(p,\mathbf{w}(p)) - V_{\overline{q}^1}\left( {p,\mathbf{w}(p) + \frac{{\eta  - (1 - \delta )(1 - {\lambda _\eta })}}{\delta }[m(1) - u({a^*},1)]} \right)} \right], \hfill \\
   = \frac{{(1 - \delta )}}{\eta }(1 - {\lambda _\eta })v({a^*},1) = v({a^*},1). \hfill \\ 
\end{gathered} \]
    \end{small}  
We have that the gradient $\mathcal{G}(p;\eta)$ is equal to $v(a^*,1)$ for all $p\in (\underline{q},\overline{q}]$, for all $\eta \leq \bar{\eta}$. Finally, when $\eta > \bar{\eta}$, the same decomposition as in the case $P \neq \emptyset$ completes the proof.

\subsubsection{For all $p \in \mathcal{I}_3$, the gradient $\mathcal{G}(p;\eta)$ is increasing in $\eta$.}\hfill

We only treat the case $P \neq \emptyset$. (The case $P = \emptyset$ is treated analogously.) Define $\bar{\eta}(p)$ as the solution to ${\varphi _{{\bar{\eta} }(p)}} = \ph (p,w(p;{\bar{\eta}(p))) = \overline{q}^{\infty}}$. By construction, for all $p \in (\overline{q}^{\infty},1]$, for all $\eta \leq \bar{\eta}(p)$, we have that $\varphi_{\eta} \in (\overline{q}^\infty,1]$. Therefore, $\varphi_{\eta} > \overline{q}$.

Choose $\bar{\eta}(p) \leq \eta' \leq \eta$. We have that $\varphi_{\eta'} \geq \varphi_{\eta} \geq \overline{q}$ since $\overline{q}^{\infty} \geq \overline{q}$ and, therefore, 
\begin{eqnarray*}
\ph\left(p, \mathbf{w}(p)+\frac{\eta  - (1 - \delta )(1 - \lambda _{\eta})}{\delta }[\m(1) - u({a^*},1)]\right)  =  \ph(\varphi_{\eta},\mathbf{w}(\varphi_{\eta})) \geq \\
 \ph(\varphi_{\eta'},\mathbf{w}(\varphi_{\eta'})=
   \ph\left(p, \mathbf{w}(p)+\frac{\eta'  - (1 - \delta )(1 - \lambda _{\eta'})}{\delta }[\m(1) - u({a^*},1)]\right).
\end{eqnarray*}

Also, since $\overline{q} \leq \varphi_{\eta} \leq p $,  we have that 
$\ph(\varphi_{\eta},\mathbf{w}(\varphi_{\eta})) \geq \ph(p,\mathbf{w}(p))$ and, therefore, $\frac{\eta  - (1 - \delta )(1 - \lambda _{\eta})}{\delta }\leq 0$. The same applies to $\eta'$. Finally, as already shown, 
    \[\frac{{\eta  - (1 - \delta )(1 - {\lambda _\eta })}}{\delta } < \frac{{{\eta ^\prime } - (1 - \delta )(1 - {\lambda _{{\eta ^\prime }}})}}{\delta }.\]
    
To ease notation, define $(\tilde{\lambda}_\eta,\tilde{\varphi}_\eta)$ as follows:
    \begin{equation}
        \label{eq:definition of varphitilde and lambdatilde}
        \left\{ \begin{gathered}
  {{\tilde \lambda }_\eta } = \lambda \left( {p,\mathbf{w}(p) - \frac{{(1 - \delta )(1 - {\lambda _\eta }) - \eta }}{\delta }\left[ {m(1) - u({a^*},1)} \right]} \right) \hfill \\
  {{\tilde \varphi }_\eta } = \varphi \left( {p,\mathbf{w}(p) - \frac{{(1 - \delta )(1 - {\lambda _\eta }) - \eta }}{\delta }\left[ {m(1) - u({a^*},1)} \right]} \right) \hfill \\ 
\end{gathered}  \right.
    \end{equation}
Notice that $\tilde{\varphi}_{\eta} = \ph(\varphi_{\eta},\mathbf{w}(\varphi_{\eta})) \in \mathcal{I}_1$ since $\varphi_{\eta} > \overline{q}^{\infty}$.    
\medskip

The rest of the proof is purely algebraic and mirrors the case $p \in \mathcal{I}_1$. First, we have the following:
    \[\begin{gathered}
  \frac{{V_{\overline{q}^1}(p,\mathbf{w}(p)) - V_{\overline{q}^1}\left( {p,\mathbf{w}(p) - \frac{{(1 - \delta )(1 - {\lambda _\eta }) - \eta }}{\delta }\left[ {\m(1) - u({a^*},1)} \right]} \right)}}{{\frac{{(1 - \delta )(1 - {\lambda _\eta }) - \eta }}{\delta }}} \hfill \\
   = \frac{{{{\tilde \lambda }_\eta }V_{\overline{q}^1}\left( {{{\tilde \varphi }_\eta },\m({{\tilde \varphi }_\eta }) + \frac{{(1 - \delta )(1 - {\lambda _\eta }) - \eta }}{{\delta {{\tilde \lambda }_\eta }}}\left[ {\m(1) - u({a^*},1)} \right]} \right) - {{\tilde \lambda }_\eta }V_{\overline{q}^1}\left( {{{\tilde \varphi }_\eta },\m({{\tilde \varphi }_\eta })} \right)}}{{\frac{{(1 - \delta )(1 - {\lambda _\eta }) - \eta }}{\delta }}} \hfill \\
   = \frac{{V_{\overline{q}^1}\left( {{{\tilde \varphi }_\eta },{w}\left( \tilde{\varphi}_{\eta};{\frac{{(1 - \delta )(1 - {\lambda _\eta }) - \eta }}{{\delta {{\tilde \lambda }_\eta }}}} \right)} \right) - V_{\overline{q}^1}\left( {{{\tilde \varphi }_\eta },\m({{\tilde \varphi }_\eta })} \right)}}{{\frac{{(1 - \delta )(1 - {\lambda _\eta }) - \eta }}{{\delta {{\tilde \lambda }_\eta }}}}}, \hfill \\ 
\end{gathered} \]
where we again use Observation A. Similarly, we have:
\[\begin{gathered}
  \frac{{V_{\overline{q}^1}(p,w(p)) - V_{\overline{q}^1}\left( {p,w(p) - \frac{{(1 - \delta )(1 - {\lambda _{{\eta ^\prime }}}) - {\eta ^\prime }}}{\delta }\left[ {\m(1) - u({a^*},1)} \right]} \right)}}{{\frac{{(1 - \delta )(1 - {\lambda _{{\eta ^\prime }}}) - {\eta ^\prime }}}{\delta }}} \hfill \\
   = \frac{{{{\tilde \lambda }_\eta }V_{\overline{q}^1}\left( {{{\tilde \varphi }_\eta },{w}\left(\tilde{\varphi}_\eta; {\frac{{(1 - \delta )(1 - {\lambda _\eta }) - \eta }}{{\delta {{\tilde \lambda }_\eta }}}} \right)} \right) - {{\tilde \lambda }_\eta }V_{\overline{q}^1}\left( {{{\tilde \varphi }_\eta },{w}\left(\tilde{\varphi}_\eta; {\frac{{(1 - \delta )(1 - {\lambda _\eta }) - \eta }}{{\delta {{\tilde \lambda }_\eta }}} - \frac{{(1 - \delta )(1 - {\lambda _{{\eta ^\prime }}}) - {\eta ^\prime }}}{{\delta {{\tilde \lambda }_\eta }}}} \right)} \right)}}{{\frac{{(1 - \delta )(1 - {\lambda _{{\eta ^\prime }}}) - {\eta ^\prime }}}{\delta }}} \hfill \\
   = \frac{{V_{\overline{q}^1}\left( {{{\tilde \varphi }_\eta },{w}\left(\tilde{\varphi}_\eta; {\frac{{(1 - \delta )(1 - {\lambda _\eta }) - \eta }}{{\delta {{\tilde \lambda }_\eta }}}} \right)} \right) - V_{\overline{q}^1}\left( {{{\tilde \varphi }_\eta },{w}\left(\tilde{\varphi}_\eta; {\frac{{(1 - \delta )(1 - {\lambda _\eta }) - \eta }}{{\delta {{\tilde \lambda }_\eta }}} - \frac{{(1 - \delta )(1 - {\lambda _{{\eta ^\prime }}}) - {\eta ^\prime }}}{{\delta {{\tilde \lambda }_\eta }}}} \right)} \right)}}{{\frac{{(1 - \delta )(1 - {\lambda _{{\eta ^\prime }}}) - {\eta ^\prime }}}{{\delta {{\tilde \lambda }_\eta }}}}}, \hfill \\ 
\end{gathered} \]
where again we use Observation A and the fact
\[{\frac{{(1 - \delta )(1 - {\lambda _\eta }) - \eta }}{{\delta {{\tilde \lambda }_\eta }}} > \frac{{(1 - \delta )(1 - {\lambda _{{\eta ^\prime }}}) - {\eta ^\prime }}}{{\delta {{\tilde \lambda }_\eta }}}}.\]
\medskip

Since $\tilde{\varphi}_{\eta} \in \mathcal{I}_1$, we have that: 
\[\begin{gathered}
  \frac{{V_{\overline{q}^1}\left( {{{\tilde \varphi }_\eta },w\left( {{{\tilde \varphi }_\eta };\frac{{(1 - \delta )(1 - {\lambda _\eta }) - \eta }}{{\delta {{\tilde \lambda }_\eta }}}} \right)} \right) - V_{\overline{q}^1}\left( {{{\tilde \varphi }_\eta },w\left( {{{\tilde \varphi }_\eta };\frac{{(1 - \delta )(1 - {\lambda _\eta }) - \eta }}{{\delta {{\tilde \lambda }_\eta }}} - \frac{{(1 - \delta )(1 - {\lambda _{{\eta ^\prime }}}) - {\eta ^\prime }}}{{\delta {{\tilde \lambda }_\eta }}}} \right)} \right)}}{{\frac{{(1 - \delta )(1 - {\lambda _{{\eta ^\prime }}}) - {\eta ^\prime }}}{{\delta {{\tilde \lambda }_\eta }}}}} \hfill \\
   \leqslant \frac{{V_{\overline{q}^1}\left( {{{\tilde \varphi }_\eta },w\left( {{{\tilde \varphi }_\eta };\frac{{(1 - \delta )(1 - {\lambda _\eta }) - \eta }}{{\delta {{\tilde \lambda }_\eta }}}} \right)} \right) - V_{\overline{q}^1}\left( {{{\tilde \varphi }_\eta },\m({{\tilde \varphi }_\eta })} \right)}}{{\frac{{(1 - \delta )(1 - {\lambda _\eta }) - \eta }}{{\delta {{\tilde \lambda }_\eta }}}}},  \hfill \\ 
\end{gathered} \]
where the inequality follows from our previous argument on the interval $\mathcal{I}_1$.

It follows that:
          \begin{eqnarray*}
          \frac{{V_{\overline{q}^1}(p,w(p)) - V_{\overline{q}^1}\left( {p,w(p) - \frac{{(1 - \delta )(1 - {\lambda _{{\eta ^\prime }}}) - {\eta ^\prime }}}{\delta }\left[ {\m(1) - u({a^*},1)} \right]} \right)}}{{\frac{{(1 - \delta )(1 - {\lambda _{{\eta ^\prime }}}) - {\eta ^\prime }}}{\delta }}}\\
           \leqslant \frac{{V_{\overline{q}^1}(p,w(p)) - V_{\overline{q}^1}\left( {p,w(p) - \frac{{(1 - \delta )(1 - {\lambda _\eta }) - \eta }}{\delta }\left[ {\m(1) - u({a^*},1)} \right]} \right)}}{{\frac{{(1 - \delta )(1 - {\lambda _\eta }) - \eta }}{\delta }}}.\end{eqnarray*}
From Equation \eqref{eq:derivative}, we then have that
\begin{scriptsize}
    \[\begin{gathered}
  \frac{1}{\eta }\left[ {V_{\overline{q}^1}(p,\m(p)) - V_{\overline{q}^1}(p,\mathbf{w}(p;\eta))} \right] = \\
  \frac{{\left( {1 - \delta } \right)(1 - {\lambda _\eta })}}{\eta }v({a^*},1) + \left[ {\frac{{\left( {1 - \delta } \right)(1 - {\lambda _\eta })}}{\eta } - 1} \right]\frac{{V_{\overline{q}^1}(p,w(p)) - V_{\overline{q}^1}\left( {p,w(p) - \frac{{(1 - \delta )(1 - {\lambda _\eta }) - \eta }}{\delta }\left[ {m(1) - u({a^*},1)} \right]} \right)}}{{\frac{{(1 - \delta )(1 - {\lambda _\eta }) - \eta }}{\delta }}} \hfill \\
   \geqslant \frac{{\left( {1 - \delta } \right)(1 - {\lambda _\eta })}}{\eta }v({a^*},1) + \left[ {\frac{{\left( {1 - \delta } \right)(1 - {\lambda _\eta })}}{\eta } - 1} \right]\frac{{V_{\overline{q}^1}(p,w(p)) - V_{\overline{q}^1}\left( {p,w(p) - \frac{{(1 - \delta )(1 - {\lambda _{{\eta ^\prime }}}) - {\eta ^\prime }}}{\delta }\left[ {\m(1) - u({a^*},1)} \right]} \right)}}{{\frac{{(1 - \delta )(1 - {\lambda _{{\eta ^\prime }}}) - {\eta ^\prime }}}{\delta }}} \hfill \\
   = \frac{{\left( {1 - \delta } \right)(1 - {\lambda _\eta })}}{\eta }v({a^*},1) + \left[ {1 - \frac{{\left( {1 - \delta } \right)(1 - {\lambda _\eta })}}{\eta }} \right]\frac{{V_{\overline{q}^1}\left( {p,w(p) - \frac{{(1 - \delta )(1 - {\lambda _{{\eta ^\prime }}}) - {\eta ^\prime }}}{\delta }\left[ {\m(1) - u({a^*},1)} \right]} \right) - V_{\overline{q}^1}(p,w(p))}}{{\frac{{(1 - \delta )(1 - {\lambda _{{\eta ^\prime }}}) - {\eta ^\prime }}}{\delta }}} \hfill \\
   = \frac{{\left( {1 - \delta } \right)(1 - {\lambda _{{\eta ^\prime }}})}}{{{\eta ^\prime }}}v({a^*},1) + \left[ {1 - \frac{{\left( {1 - \delta } \right)(1 - {\lambda _{{\eta ^\prime }}})}}{{{\eta ^\prime }}}} \right]\frac{{V_{\overline{q}^1}\left( {p,w(p) - \frac{{(1 - \delta )(1 - {\lambda _{{\eta ^\prime }}}) - {\eta ^\prime }}}{\delta }\left[ {\m(1) - u({a^*},1)} \right]} \right) - V_{\overline{q}^1}(p,w(p))}}{{\frac{{(1 - \delta )(1 - {\lambda _{{\eta ^\prime }}}) - {\eta ^\prime }}}{\delta }}} \hfill \\
   + \left[ { \frac{{\left( {1 - \delta } \right)(1 - {\lambda _{{\eta ^\prime }}})}}{{{\eta ^\prime }}}-\frac{{\left( {1 - \delta } \right)(1 - {\lambda _\eta })}}{\eta } } \right]\left[ {\frac{{V_{\overline{q}^1}\left( {p,w(p) - \frac{{(1 - \delta )(1 - {\lambda _{{\eta ^\prime }}}) - {\eta ^\prime }}}{\delta }\left[ {\m(1) - u({a^*},1)} \right]} \right) - V_{\overline{q}^1}(p,w(p))}}{{\frac{{(1 - \delta )(1 - {\lambda _{{\eta ^\prime }}}) - {\eta ^\prime }}}{\delta }}} - v({a^*},1)} \right] \hfill \\
   \geqslant \frac{1}{{{\eta ^\prime }}}\left[ {V_{\overline{q}^1}(p,\m(p)) - V_{\overline{q}^1}(p,w(p;\eta^\prime))} \right], \hfill \\ 
\end{gathered} \]
\end{scriptsize}
where the last inequality follows from:
    \[\begin{gathered}
  \frac{{V_{\overline{q}^1}\left( {p,w(p) - \frac{{(1 - \delta )(1 - {\lambda _{{\eta ^\prime }}}) - {\eta ^\prime }}}{\delta }\left[ {\m(1) - u({a^*},1)} \right]} \right) - V_{\overline{q}^1}(p,w(p))}}{{\frac{{(1 - \delta )(1 - {\lambda _{{\eta ^\prime }}}) - {\eta ^\prime }}}{\delta }}} \hfill \\
   = \frac{{{{\tilde \lambda }_{{\eta ^\prime }}}V_{\overline{q}^1}({{\tilde \varphi }_{{\eta ^\prime }}},\m({{\tilde \varphi }_{{\eta ^\prime }}})) - {{\tilde \lambda }_{{\eta ^\prime }}}V_{\overline{q}^1}\left( {{{\tilde \varphi }_{{\eta ^\prime }}},w\left( {{{\tilde \varphi }_{{\eta ^\prime }}};\frac{{(1 - \delta )(1 - {\lambda _{{\eta ^\prime }}}) - {\eta ^\prime }}}{{\delta {{\tilde \lambda }_{{\eta ^\prime }}}}}} \right)} \right)}}{{\frac{{(1 - \delta )(1 - {\lambda _{{\eta ^\prime }}}) - {\eta ^\prime }}}{\delta }}} \geqslant v({a^*},1). \hfill \\ 
\end{gathered} \]

We now show that the the gradient $\mathcal{G}(p;\eta)$ is smaller than $v(a^*,1)$ for any $\eta\leq \bar{\eta} (p)$. From Equation \eqref{eq:derivative}, we have that: 
\begin{scriptsize}
    \[\begin{gathered}
  \frac{1}{\eta }\left[ {V_{\overline{q}^1}(p,\m(p)) - V_{\overline{q}^1}(p,\mathbf{w}(p;\eta))} \right] \hfill \\
   = \frac{{(1 - \delta )(1 - {\lambda _\eta })}}{\eta }v({a^*},1) - \left[ {\frac{{\left( {1 - \delta } \right)(1 - {\lambda _\eta })}}{\eta } - 1} \right]\frac{{V_{\overline{q}^1}\left( {p,w(p) - \frac{{(1 - \delta )(1 - {\lambda _\eta }) - \eta }}{\delta }\left[ {m(1) - u({a^*},1)} \right]} \right) - V_{\overline{q}^1}(p,w(p))}}{{\frac{{(1 - \delta )(1 - {\lambda _\eta }) - \eta }}{\delta }}} \hfill \\
   = v({a^*},1) - \left[ {\frac{{\left( {1 - \delta } \right)(1 - {\lambda _\eta })}}{\eta } - 1} \right]\left[ {\frac{{V_{\overline{q}^1}\left( {p,w(p) - \frac{{(1 - \delta )(1 - {\lambda _\eta }) - \eta }}{\delta }\left[ {m(1) - u({a^*},1)} \right]} \right) - V_{\overline{q}^1}(p,w(p))}}{{\frac{{(1 - \delta )(1 - {\lambda _\eta }) - \eta }}{\delta }}} - v({a^*},1)} \right] \hfill \\
   =v({a^*},1) - \left[ {\frac{{\left( {1 - \delta } \right)(1 - {\lambda _\eta })}}{\eta } - 1} \right]\left[ {\frac{{{{\tilde \lambda }_\eta }V_{\overline{q}^1}({{\tilde \varphi }_\eta },\m({{\tilde \varphi }_\eta })) - {{\tilde \lambda }_\eta }V_{\overline{q}^1}\left( {{{\tilde \varphi }_\eta },w\left( {{{\tilde \varphi }_\eta };\frac{{(1 - \delta )(1 - {\lambda _\eta }) - \eta }}{{\delta {{\tilde \lambda }_\eta }}}} \right)} \right)}}{{\frac{{(1 - \delta )(1 - {\lambda _\eta }) - \eta }}{\delta }}} - v({a^*},1)} \right] \hfill \\
    = v({a^*},1) - \underbrace{\left[ {\frac{{\left( {1 - \delta } \right)(1 - {\lambda _\eta })}}{\eta } - 1} \right]}_{\geq 0}\underbrace{\left[ {\frac{{V_{\overline{q}^1}({{\tilde \varphi }_\eta },\m({{\tilde \varphi }_\eta })) - V_{\overline{q}^1}\left( {{{\tilde \varphi }_\eta },w\left( {{{\tilde \varphi }_\eta };\frac{{(1 - \delta )(1 - {\lambda _\eta }) - \eta }}{{\delta {{\tilde \lambda }_\eta }}}} \right)} \right)}}{{\frac{{(1 - \delta )(1 - {\lambda _\eta }) - \eta }}{{\delta {{\tilde \lambda }_\eta }}}}} - v({a^*},1)} \right]}_{\geq 0} \hfill \\
   \leqslant v({a^*},1), \hfill \\ 
\end{gathered} \]
\end{scriptsize}
where the inequality follows from the fact that $\tilde{\varphi}_{\eta} \leq \underline{p}$ (therefore, from our arguments on the interval $
\mathcal{I}_1$, where we show that the gradient is larger than $v(a^*,1)$). 

Finally, we can use a similar decomposition as in the case $p \in \mathcal{I}_2$ to prove that the gradient is increasing for all $\eta$.

\section{A Formal Discussion of Other Policies}
\label{app:other-policies}

\subsection{Non-uniqueness and comparison with the KG's policy.} Our policy is not always uniquely optimal. We demonstrate the non-uniqueness with the help of a simple example and then discuss how our policy compares with the KG's policy (for Kamenica-Gentzkow's policy). 

\textit{\textbf{Example 2.}} The agent has two possible actions $a_0$ and $a_1$, with $a_0$ (resp., $a_1$) the agent's optimal action when the state is $\omega_0$ (resp., $\omega_1$). The principal wants to induce $a_0$ as often as possible, i.e., $a^*=a_0$. The discount factor is $1/2$. The payoffs are in Table \ref{tab:ex2}, with the first coordinate corresponding to the principal's payoff.
\begin{table}[h]
\centering \caption{Payoff table of Example 2}\label{tab:ex2}
\begin{tabular}{|c|c|c|}
\hline
 & $a_0$ & $a_1$ \\  \hline
 $\omega_0$ & $1,1$ & $0,0$ \\ \hline
 $\omega_1$ & $1,0$ & $0,1$  \\ \hline
 \end{tabular}
 \medskip

\end{table}

In Example 2, we have that: $m(p)=\max(1-p,p)$, $M(p)=1$ and $u(a^*,p)=1-p$. Thus, $a^*$ is optimal for all $p \in P=[0,1/2]$. Moreover, $Q^1= [0,2/3]$ and $\bold{w}(p)=3p-1$ for $p \in( 1/2,2/3]$. 

We now provide an explicit characterization of the value function.  We first compute the value function $V_{\overline{q}^1}(\cdot,m(\cdot))$ and check whether it is concave. For $p \in [0,1/2]$, the policy recommends $a^*$ and promises a continuation payoff of $m(p)$.  That is, since $a^*$ is optimal, the principal does not need to incentivize the agent. For $p \in (1/2,2/3]$, the policy recommends $a^*$ and promises a continuation payoff of $\bold{w}(p)$. At $(p,\bold{w}(p))$ with $p \in (1/2,2/3]$, the policy splits $p$ into $\varphi(p,\bold{w}(p))$ and $1$, with probability $\lambda(p,\bold{w}(p))$ and $1-\lambda(p,\bold{w}(p))$ respectively. (See Equation (\ref{eq1}).) 

We obtain that $\lambda(p,\bold{w}(p))= (3-4p)$ and $\varphi(p,\bold{w}(p))=\frac{2-3p}{3-4p}$. Note that $\varphi(p,\bold{w}(p))=\frac{2-3p}{3-4p}<\frac{1}{2}$  since $p \in (1/2,2/3]$. After splitting $p$ into $\varphi(p,\bold{w}(p))$, the principal therefore obtains a payoff of 1 in all subsequent periods. It follows that the principal's expected payoff is
\[\frac{1}{2}+\frac{1}{2}\lambda(p,\bold{w}(p))= 2(1-p).\]
Finally, if $p\in (2/3,1]$, the policy splits $p$ into $2/3$ and $1$ with probability $3(1-p)$ and $(1-3(1-p))$,  respectively. The principal's expected payoff is then
\[3(1-p) \times \Big[\frac{1}{2}+\frac{1}{2}\lambda\left(\frac{2}{3},\bold{w}\left(\frac{2}{3}\right)\right)\Big]=3(1-p)\times 2\left(1-\frac{2}{3}\right)=2(1-p). \]
So, the value function $V_{\overline{q}^1}$ induced by the policy $\tau_{\q}$ is such $V_{\overline{q}^1}(p,m(p))= 1$ for all $p \in [0,1/2]$ and $V_{\overline{q}^1}(p,m(p))= 2(1-p)$ for all $p \in(1/2,1]$. Since it is concave in $p$, this guarantees that $q^*=\overline{q}^1$ and, thus, the policy is indeed optimal.

We now consider another policy, which we call the KG's policy. The aim of the KG's policy is to persuade the agent to choose $a^*$ as often as possible by disclosing information at the initial stage only. The best payoff the principal can obtain with a KG's policy is:
\[\max_{(\lambda_s,p_s,a_s)}\sum_{s}\lambda_s v(a_s,p_s), \]
subject to \[\forall s, u(a_s,p_s) \geq m(p_s), \text{and} \sum_s \lambda_s p_s=p.\]

In Example 2, the KG's policy differs from our policy only when $p \geq 1/2$, and consists in  splitting $p$ into $1/2$ and $1$, with probability $2(1-p)$ and $1-2(1-p)$ respectively. The KG's policy induces the same value function as our policy, hence is also optimal.  We now prove that this is not accidental. \medskip

Suppose that there are only two actions, $a_0$ and $a_1$, such that $a_0$ (resp., $a_1$) is optimal at state $\omega_0$ (resp., $\omega_1$). The principal aims at implementing $a_0$ as often as possible, i.e., $a^*=a_0$.\footnote{If $a^*=a_1$, then $0=m(1)-u(a_1,1) \geq (m(0)-u(a_1,0))\frac{v(a^*,1)}{v(a^*,0)}=(u(a_0,0)-u(a_1,0))\frac{v(a^*,1)}{v(a^*,0)}\geq 0$, i.e., $a_1$ is also optimal when the agent believes that the state is $\omega_0$ with probability 1.} Remember that $a_0$ is optimal at all beliefs in 
$[\underline{p},\overline{p}]$. Since $a_0$ is optimal at $0$, $\underline{p}=0$.  To streamline the exposition, assume that the prior $p_0 > \overline{p}$. (If $p_0 \leq \overline{p}$, an optimal policy is to never reveal any information.)  It is then immediate to see that  the KG's policy consists in splitting the prior $p_0$ into $\overline{p}$ and $1$, with probability $\frac{1-p_0}{1-\overline{p}}$ and $1- \frac{1-p_0}{1-\overline{p}}$, respectively. Intuitively, the principal designs a binary experiment, with one signal perfectly informing the agent that the state is $\omega_1$ and the other partially informing the agent so that his posterior beliefs is $\overline{p}$. 

We can contrast the KG's policy with our policy. Unlike the KG's policy, our policy does not reveal information to the agent at the first period, and only reveals information to the agent if he plays $a_0$. If the agent plays $a_0$ at the first period, the policy splits $p_0$ into $\varphi(p_0,\bold{w}(p_0))$ and $1$ with probability $\lambda(p_0,\bold{w}(p_0))$ and $1-\lambda(p_0,\bold{w}(p_0))$, respectively. Note that $\varphi(p_0,\bold{w}(p_0)) \leq \overline{p}$ since $\bold{w}(p_0)  \geq m(p_0)$. Thus, our policy guarantees that the agent plays $a^*$ for sure at the first period. However, this comes at a cost: the principal needs to reveal more information to the agent at the next period and, consequently, inducing the agent to play $a_0$ with a lower probability. Somewhat surprisingly, both policies are optimal, regardless of the discount factor. 
\begin{corollary}\label{KG-policy}
If there are only two actions, then the KG's policy is also optimal. 
\end{corollary}

\medskip 
%

As Example 1 shows, the KG's policy is not always optimal. Yet, if $a^*$ is not strictly dominated and the function $m$ is linear from $\overline{p}$ to 1, then the KG's policy is also optimal at all priors above $\overline{p}$. (A proof is available upon request.) More generally, whenever the value function $V^*$ is linear in $(p,w)$, the KG's policy is also optimal. We conjecture, however, that the value function $V^*$ is generically non-linear. 

\medskip

\subsection{Comparison with the ``random disclosure'' policy.}  Remember that the policy of fully disclosing the state with delay plays a prominent role in the work of \citet{Ball2019dynamic} and \citet{orlov2018persuading}. Since we study a discrete time model, we do not directly compare our policy with the policy of fully disclosing the state with delay, but with the ``random disclosure'' policy. The ``random disclosure'' policy consists in fully disclosing the state with probability $\alpha$ at period $t+1$ (and to withhold all information with the complementary probability) if the agent plays $a^*$ at period $t$.\footnote{In continuous time, the policy of fully disclosing the state with delay yields the same payoff as the ``random disclosure'' policy.} 

We first compute the principal's payoff if he commits to the best ``random disclosure'' policy. To ease the exposition, we assume that $a^*$ is not optimal at the belief $p=0$.\footnote{When $a^*$ is optimal at $p=0$, we need to add the term  $\delta \alpha (1-p)v(a^*,p)$ to the objective, which corresponds to the payoff the principal obtains when the disclosed state is $\omega_0$.} Assume that $p \in Q^1$. The best ``random disclosure'' policy is solution to the maximization problem:
\[V= \max_{\alpha \in [0,1]} (1-\delta)v(a^*,p)+ \delta(1-\alpha) V,\]
subject to 
\[U=(1-\delta)u(a^*,p)+ \delta \left[\alpha M(p) + (1-\alpha)U\right] \geq m(p).\]
The optimal solution is 
\[\alpha^*=\frac{\mathbf{w}(p)-m(p)}{M(p)-m(p)}= \frac{1-\delta}{\delta}\frac{m(p)-u(a^*,p)}{M(p)-m(p)}, \] 
inducing the value
\[(1-\delta)\sum_{t}\delta^t \left(\frac{M(p)-\mathbf{w}(p)}{M(p)-m(p)}\right)^t v(a^*,p)=\frac{M(p)-m(p)}{M(p)-u(a^*,p)}v(a^*,p). \]
The formula has a natural interpretation. Whenever the agent is recommended to play $a^*$, no information has been revealed yet, so that the maximal value of information the principal can create is $M(p)-m(p)$.  To incentivize the agent, the principal needs to promise a continuation payoff of $\mathbf{w}(p)$ in the future and thus needs to create an information value of $\mathbf{w}(p) -m(p)$. To create an information value of $\mathbf{w}(p) -m(p)$, the principal commits to fully disclose the state with some probability, hence foregoing the opportunity to incentivize the agent to play $a^*$ in the future. Therefore, the highest probability with which the principal can incentivize the agent to play $a^*$ is $(M(p)-\mathbf{w}(p))/(M(p)-m(p))$.\medskip

To understand why and when the principal can do better than following the ``random recursive policy,'' we study the \emph{relaxed} version of our problem, where only the (ex-ante) participation constraint needs to be satisfied. Consider the following policy. The principal discloses information at the ex-ante stage, i.e., chooses a splitting $(\lambda_s,p_s)_s$ of $p$, and recommends the agent to play $a^*$ at all periods with probability $\beta_s$ when the realized signal is $s$. We continue to assume that $p \in Q^1$. The policy satisfies the participation constraint if
\begin{equation*}
\sum_{s}\lambda_s\left[\beta_s u(a^*,p_s)+ (1-\beta_s)m(p_s)\right] \geq m(p). 
\end{equation*}
We can rewrite the participation constraint as:
\begin{equation}
\sum_{s} \lambda_s(1-\beta_s)(m(p_s)-u(a^*,p_s)) \geq m(p)-u(a^*,p), \label{eq2}
\end{equation}
where $m(p_s)-u(a^*,p_s)$ is the opportunity cost of following the recommendation at belief $p_s$. The principal maximizes $\sum_s
\lambda_s \beta_s v(a^*,p_s)$ subject to the participation constraint. Clearly, the  participation constraint binds at a maximum. Moreover, since $m$ is convex, the best for the principal is to fully disclose all information, i.e., to split $p$ into $0$ and $1$. 

Note that if the principal recommends $a^*$ with the same probability at all $s$, his payoff is 
\[\frac{M(p)-m(p)}{M(p)-u(a^*,p)}v(a^*,p),\]
which is precisely the payoff of the ``random recursive'' policy.\footnote{When $a^*$ is optimal at $p=0$, we need to add the term $(1-p)\left(1-\frac{M(p)-m(p)}{M(p)-u(a^*,p)}\right)v(a^*,p)$.}  

The principal can do better by exploiting the difference in opportunity costs at the two extreme beliefs $0$ and $1$. Writing $\beta_1$ (resp., $\beta_0$) for the probability of recommending $a^*$ conditional on the posterior being $1$ (resp., $0$), the principal maximizes $p \beta_1v(a^*,1)+ (1-p)\beta_0v(a^*,0)$ subject to:
\[p \beta_1 (m(1)-u(a^*,1))+ (1-p) \beta_0 (m(0)-u(a^*,0)) \leq M(p)-m(p). \] 
The right-hand side is the maximal value of information the principal can create, while the left-hand side is the expected opportunity cost of following the recommendation. As with the ``random disclosure'' policy, the principal needs to generate the maximal value of information; this is the maximal value the principal can use to incentivize the agent. However, unlike the ``random disclosure'' policy, the principal needs to use the surplus created asymmetrically, as it is easier to incentivize the agent in state $\omega_0$ than $\omega_1$.  

More precisely, the problem is linear in $(\beta_0,\beta_1)$. Therefore, since the slope $\frac{v(a^*,0)}{v(a^*,1)}$ is larger than the slope $\frac{m(0)-u(a^*,0)}{m(1)-u(a^*,1)}$, the optimal solution is to set $\beta_0$ as high as possible. For instance, if $M(p)-m(p) \leq (1-p)(m(0)-u(a^*,0))$, the best is to set $(\beta_0,\beta_1)= ( \frac{M(p)-m(p)}{(1-p)(m(0)-u(a^*,0))},0)$, resulting in a payoff of 
\[\frac{M(p)-m(p)}{m(0)-u(a^*,0)}v(a^*,0) \geq \frac{M(p)-m(p)}{M(p)-u(a^*,p)}v(a^*,p), \]
with a strict inequality if the opportunity cost is strictly higher in state $\omega_1$.\footnote{See Appendix \ref{app-first-best} for the full characterization.}  This is the solution to the relaxed constraint.

\medskip

While our policy also needs to incentivize the agent to follow the recommendation, it exploits the same asymmetries in opportunity costs as the above policy, which explains why it outperforms the ``random disclosure'' policy. 

To conclude, note that if $\frac{v(a^*,0)}{v(a^*,1)} = \frac{m(0)-u(a^*,0)}{m(1)-u(a^*,1)}$, then the recursive policy solves the relaxed problem and, therefore, is also optimal. 

\subsection{Proof of Corollary \ref{KG-policy}}

We first compute the principal's payoff induced by our policy. To ease notation, we write $\varphi$ for $\varphi(p,\bold{w}(p))$.  We first assume that $q^*=\q$, compute the value function $V_{\overline{q}^1}(p,m(p))$ for all $p$ and check that it is concave. By construction, the principal's payoff satisfies: 
\[V_{\overline{q}^1}(p,m(p))= (1-\delta)v(a^*,p)+ \delta V_{\overline{q}^1}(p,\bold{w}(p))= (1-\delta)v(a^*,p)+ \delta \frac{1-p}{1-\varphi}v(a^*,\varphi).\]
Remember that
\begin{eqnarray*}
\bold{w}(p)=\frac{m(p)-(1-\delta)u(a_0,p)}{\delta}= \frac{1-p}{1-\varphi}m(\varphi)+\frac{p-\varphi}{1-\varphi}m(1).
\end{eqnarray*}
Since $\bold{w}(p) = m(p)=u(a_0,p)$ when $p \leq \overline{p}$, we have that $\varphi =p$ and, therefore, the principal payoff is $1$ when $p \leq \overline{p}$. Assume that $p > \overline{p}$. We have that: 
\[\bold{w}(p)= \frac{u(a_1,p)-(1-\delta)u(a_0,p)}{\delta}= \frac{1-p}{1-\varphi}u(a_0,\varphi)+\frac{p-\varphi}{1-\varphi}u(a_1,1),\]
since $m(\varphi)=u(a_0,\varphi)$ and $\varphi \leq \overline{p}$. (To see this, if $\varphi > \overline{p}$, then $m(\varphi)= u(a_1,\varphi)$, hence $\bold{w}(p)=m(p)$, a contradiction with $\bold{w}(p) > m(p)$ when $p > \overline{p}$.) 
The above equation is equivalent to:
\[(1-\varphi)[u(a_1,p) - (1-\delta)u(a_0,p)]  = \delta[(1-p)u(a_0,\varphi)+(p-\varphi)u(a_1,1)]. \]
Observing that $u(a,p) = (1-p)(u(a,0)-u(a,1))+u(a,1)$ for all $a$ and, similarly, for $\varphi$, we can simplify the above expression to 
\[ \delta\frac{1-p}{1-\varphi}= \delta-p + (1-p)\frac{u(a_0,0)-u(a_1,0)}{u(a_1,1)-u(a_0,1)}.\]
Lastly, remember that the threshold $\overline{p}$ is given by:  
\[1-\overline{p}=\frac{u(a_1,1)-u(a_0,1)}{u(a_0,0)-u(a_0,1) + u(a_1,1)-u(a_1,0)}, \]
and, therefore, 
\begin{eqnarray*}
V_{\overline{q}^1}(p,m(p))& = & v(a^*,p) + \delta \left(1- \frac{1-p}{1-\varphi}\right)v(a^*,1) \\
& = & \frac{1-p}{1-\overline{p}}v(a^*,\overline{p}) + \left[1-\frac{1-p}{1-\overline{p}}+ \delta \left(1- \frac{1-p}{1-\varphi}\right) \right]v(a^*,1)\\
& = &  \frac{1-p}{1-\overline{p}}v(a^*,\overline{p}). 
\end{eqnarray*}

Since the KG's policy induces the same payoff, it is also optimal. 

\subsection{First best}\label{app-first-best}
This section provides detail about the solution to the first-best problem, which we study when comparing our policy with the random disclosure policy.  Let 
\[\alpha_1^*=1-\frac{m(p)-u(a^*,p)}{p(m(1)-u(a^*,1))}= \frac{M(p)-m(p) - (1-p)(m(0)-u(a^*,0))}{p(m(1)-u(a^*,1))}.\]
Note that $\alpha_1^* \leq 1$, with equality if $m(p)=u(a^*,p)$), and $\alpha_1^* <0$ if $M(p)-m(p) - (1-p)(m(0)-u(a^*,0)) <0$.

At an optimum, the participation constraint clearly binds. If $m(0)-u(a^*,0)=0$, the solution is clearly $(1,\frac{M(p)-m(p)}{p(m(1)-u(a^*,1))})$. Assume that $m(0)-u(a^*,0)>0$. We can rewrite the principal's objective as a function of $\alpha_1$: 
\begin{eqnarray*}
\begin{cases}
p\alpha_1v(a^*,1)+(1-p)v(a^*,0) &\text{\;if\;}\alpha_1 \leq \max(0,\alpha_1^*),\\
p \alpha_1\left(v(a^*,1)- v(a^*,0)\frac{m(1)-u(a^*,1)}{m(0)-u(a^*,0)}\right )+ \frac{M(p)-m(p)}{m(0)-u(a^*,0)}v(a^*,0) & \text{\;if\;} \max(0,\alpha^*_1) \leq \alpha_1 \leq \frac{M(p)-m(p)}{p(m(1)-u(a^*,1))},\\
-\infty & \text{otherwise}.
\end{cases}
\end{eqnarray*}
Note that the objective is continuous in $\alpha_1$. The optimal payoff is therefore:
\[p \max(0,\alpha_1^*)v(a^*,1)+ (1-p)\max\left(\frac{M(p)-m(p)}{(1-p)(m(0)-u(a^*,0))}, 1\right)v(a^*,0), \]
obtained with $(\alpha_0,\alpha_1)= \left(\frac{M(p)-m(p)}{(1-p)(m(0)-u(a^*,0))}, 0\right)$ if  $\frac{M(p)-m(p)}{(1-p)(m(0)-u(a^*,0))} \leq 1$ and 
$(\alpha_0,\alpha_1)= (1,\alpha_1^*)$, otherwise.

\bibliographystyle{ecta}
\bibliography{references-renou.bib}

\end{document}